\documentclass[twoside,11pt]{article}

\usepackage[abbrvbib, preprint]{jmlr2e}

\newtheorem{condition}{Condition}

\usepackage{amsmath}

\usepackage{booktabs}

\def\hat{\widehat}
\def\Relief{Reliever}
\def\A{\mathcal{A}}
\def\Lcal{\mathcal{L}}
\def\L{\mathcal{L}}
\def\I{\mathcal{I}}
\def\R{\mathcal{R}}
\def\M{\mathcal{M}}
\def\hM{\widehat{\mathcal{M}}}

\def\ebf{\mathbf{e}}
\def\gbf{\mathbf{g}}
\def\rbf{\mathbf{r}}
\def\ubf{\mathbf{u}}
\def\vbf{\mathbf{v}}
\def\xbf{\mathbf{x}}
\def\zbf{\mathbf{z}}
\def\Dbf{\mathbf{D}}
\def\Hbf{\mathbf{H}}

\def\Zbb{\mathbb{Z}}
\def\Rbb{\mathbb{R}}
\def\Sbb{\mathbb{S}}
\def\Ebb{\mathbb{E}}
\def\Pbb{\mathbb{P}}

\def\Bcal{\mathcal{B}}
\def\Fcal{\mathcal{F}}
\def\Gcal{\mathcal{G}}
\def\Hcal{\mathcal{H}}
\def\Kcal{\mathcal{K}}
\def\Ncal{\mathcal{N}}
\def\Scal{\mathcal{S}}
\def\Tcal{\mathcal{T}}
\def\Ucal{\mathcal{U}}

\def\balpha{\boldsymbol{\alpha}}

\def\btheta{\boldsymbol{\theta}}

\def\htheta{\widehat{\btheta}}

\newcommand{\set}[1]{\{#1\}}

\newcommand{\Bigset}[1]{\Bigl\{#1\Bigr\}}

\newcommand{\norm}[1]{\lVert #1\rVert}
\newcommand{\bignorm}[1]{\bigl\lVert #1\bigr\rVert}
\newcommand{\Bignorm}[1]{\Bigl\lVert #1\Bigr\rVert}
\newcommand{\biggnorm}[1]{\biggl\lVert#1\biggr\rVert}

\newcommand{\abs}[1]{|#1|}
\newcommand{\bigabs}[1]{\bigl| #1\bigr|}
\newcommand{\Bigabs}[1]{\Bigl| #1\Bigr|}

\newcommand{\size}[1]{\vert#1\vert}

\DeclareMathOperator*{\argmin}{\mathrm{arg\,min}}
\DeclareMathOperator*{\argmax}{\mathrm{arg\,max}}

\newcommand{\kmin}{\underline{\kappa}}

\DeclareMathOperator*{\Var}{\mathrm{Var}}
\DeclareMathOperator*{\supp}{\mathrm{supp}}
\newcommand{\truecps}{\Tcal^\ast}
\newcommand{\estcps}{\widehat{\Tcal}}
\newcommand{\cpsset}{\Tcal}
\newcommand{\zero}[0]{\boldsymbol{0}}
\newcommand{\id}{\boldsymbol{1}}

\usepackage{lastpage}
\jmlrheading{26}{2025}{1-\pageref{LastPage}}{7/24}{8/25}{24-1108}{Chengde Qian, Guanghui Wang and Changliang Zou}
\ShortHeadings{Reliever: Relieving the Burden of Costly Model Fits for Changepoint Detection}{Qian, Wang and Zou}

\firstpageno{1}

\usepackage{bm}
\usepackage{enumitem}

\usepackage{hyperref}

\usepackage{tikz}
\usepackage{subcaption}

\usepackage{multirow}
\usepackage{threeparttable}
\usepackage{float}

\newif\ifcompact\compactfalse
\ifcompact
\mypageset
\IfFileExists{./color.tex}{\input{color.tex}}{}
\fi

\begin{document}

\title{Reliever: Relieving the Burden of Costly Model Fits for Changepoint Detection}

\author{\name Chengde Qian \email qianchd@gmail.com \\
       \addr School of Mathematical Sciences \\
       Shanghai Jiao Tong University \\
       Shanghai 200240, China
       \AND
       \name Guanghui Wang \email ghwang.nk@gmail.com \\
       \addr School of Statistics and Data Science, LPMC, KLMDASR, and LEBPS \\
       Nankai University \\
       Tianjin 300071, China
       \AND
       \name Changliang Zou \email zoucl@nankai.edu.cn \\
       \addr NITFID, School of Statistics and Data Science, LPMC, KLMDASR, and LEBPS \\
       Nankai University \\
       Tianjin 300071, China
       }

\editor{Ji Zhu}

\maketitle

\begin{abstract}
Changepoint detection typically relies on a grid-search strategy for optimal data segmentation. When model fitting itself is expensive, repeatedly fitting a model on every candidate segment dominates the computation. Existing approaches mitigate this by pruning the grid, thus reducing the number of segments (and model fits). We propose Reliever, which instead cuts the number of model fits directly and nests seamlessly within standard grid-search routines. Reliever fits a small, deterministic collection of proxy models and reuses them wherever they apply, making it compatible with a wide range of existing algorithms. For high-dimensional regression with changepoints, coupling Reliever with an optimal grid-search method yields changepoint and coefficient estimators that are rate-optimal up to a logarithmic factor. Extensive numerical experiments demonstrate that Reliever rapidly and accurately detects changepoints across a wide range of high-dimensional and nonparametric models.
\end{abstract}


\section{Introduction}\label{sec:intro}

Changepoint detection serves to identify changes in statistical properties such as mean, variance, slope, or distribution within ordered observations. This technique finds applications in diverse domains including time series analysis, signal processing, finance, neuroscience, and environmental monitoring.

To identify the number and locations of changepoints, a common method is to conduct a \textit{grid search} to optimize data segmentation by minimizing (or maximizing) a specific criterion. This criterion often integrates a sum of segment-wise losses (or gains, respectively) with a penalty for excessive segmentation. Grid-search algorithms are broadly categorized into \textit{optimal} and \textit{greedy} strategies. Optimal strategies use dynamic programming \citep{MR978902,jackson2005algorithm,killick_optimal_2012-1} to find the global minimum, while greedy strategies, such as binary segmentation \citep{fryzlewicz_wild_2014,baranowski_narrowest_2019,kovacs_seeded_2022} and moving windows \citep{Hao+SelenaNiu+Zhang-2013,MR3706768}, iteratively approximate this minimum. These algorithms necessitate repeatedly fitting models and evaluating loss functions across numerous data segments. Table \ref{tab:complexity} outlines the computational complexity of the grid-search step \textit{in isolation}---that is, it treats the required model fits and loss values as if they were already available---so that one can compare how each algorithm scales with the sample size $n$. These algorithms include segment neighborhood \citep[SN,][]{MR978902}, optimal partitioning \citep[OP,][]{jackson2005algorithm}, pruned exact linear time \citep[PELT,][]{killick_optimal_2012-1}, wild binary segmentation \citep[WBS,][]{fryzlewicz_wild_2014}, and seeded binary segmentation \citep[SeedBS,][]{kovacs_seeded_2022}. For an extensive review of grid-search algorithms, please refer to \cite{CHO2021}.

\begin{table}[htbp]
\setlength\tabcolsep{0pt}
\begin{threeparttable}
\caption{\small Computational complexity of grid-search algorithms in isolation and total model-fitting operations, comparing the original implementations and the proposed Reliever implementations. The notation $a_n$ denotes the complexity of fitting a single model on an interval of length $n$. The set $\R$ is the pre-specified, deterministic collection of intervals on which Reliever actually fits models, with its cardinality $|\R|=O(n)$; see Definition~\ref{itv_construction}.}
\label{tab:complexity}
\centering
\begin{tabular*}{.97\linewidth}{c@{\extracolsep{\fill}}*{5}{c}}
\toprule
& \multicolumn{3}{c}{Optimal} & \multicolumn{2}{c}{Greedy}\\
\cline{2-4}
\cline{5-6}
Grid-search algorithm & SN & OP & PELT\tnote{$\dagger$} & WBS & SeedBS\\
\midrule
\multicolumn{6}{l}{Complexity of the grid-search step (in isolation)}\\
& $O(Kn^2)$\tnote{$\ddagger$} & $O(n^2)$ & $O(n)$ & $O(Mn)$\tnote{$\mathsection$} & $O(n\log n)$ \vspace{0.7em}\\
\multicolumn{6}{l}{Total model-fitting operations}\\
Original & $O(n^2a_{n})$ & $O(n^2a_{n})$ & $O(na_{n})$ & $O(Mna_{n})$ & $O\big(n(\log n)a_{n}\big)$\\
Reliever & $O(|\R|a_{n})$ & $O(|\R|a_{n})$ & $O(|\R|a_{n})$ & $O(|\R|a_{n})$ & $O(|\R|a_{n})$\\
\bottomrule
\end{tabular*}
\begin{tablenotes}[flushleft]\footnotesize
    \item[$\dagger$] For cases when the pruning condition is met \citep[][Eq. (4)]{killick_optimal_2012-1}; if pruning fails, PELT reduces to OP.
    \item[$\ddagger$] $K$: user-specified upper bound on the number of changepoints.
    \item[$\mathsection$] $M$: number of random intervals in WBS.
\end{tablenotes}
\end{threeparttable}
\end{table}

To evaluate the loss on any candidate interval $I\subset(0,n]$ with integer endpoints, we must first fit a model $\hM_I$ on that segment. Across the set of candidate intervals determined by the grid-search algorithm, the resulting sequence of model fits $\{\hM_I\}$ often dominates the runtime in modern changepoint procedures, far outweighing both the associated loss evaluations $\{\Lcal(I;\hM_I)\}$ and the modest overhead of iterating through the interval grid. For instance, consider high-dimensional linear models with changepoints, estimated using the lasso \citep{MR3453652,leonardi_computationally_2016,kaul_efficient_2019-1,wang_statistically_2021,xu2022change}. Fitting the lasso on an interval of length $n$ by coordinate descent requires $O(np)$ operations per iteration, so the number of variables $p$ directly drives runtime. If the penalty parameter is selected through cross-validation, the cost of a single model fit increases multiplicatively. Additionally, unlike classical mean-change models---where the sample mean can be updated incrementally \citep{MR978902}---high-dimensional fits cannot be adjusted cheaply when observations are added or removed.
Consequently, the sequence of model fits dominates the overall complexity; see Table \ref{tab:complexity}. Similar computational bottlenecks arise in changepoint models that incorporate graphical structures \citep{MR4313475}, vector autoregressive dynamics \citep{MR4399083,doi:10.1080/01621459.2022.2079514}, network topologies \citep{MR4206675}, nonparametric frameworks \citep{MR3210993,MR4515551,MR4567802}, and mechanisms for handling missing data \citep{MR4460584}.

\subsection{Our Idea}\label{sec:main_idea}

Our approach---\textit{Reliever}---operates as follows. For each candidate interval $I\subset(0,n]$, a standard grid-search algorithm $\A$ involves fitting an interval-specific model $\widehat{M}_I$ and evaluating the loss $\mathcal{L}(I; \widehat{M}_I)$. Reliever replaces this costly step with a proxy fit: we pair $I$ with an interval $R_I$, chosen from a pre-specified deterministic collection $\R$ (see Definition \ref{itv_construction}). We fit the model $\widehat{M}_{R_I}$ on $R_I$ and evaluate the loss on the target interval $I$ via $\mathcal{L}(I; \widehat{M}_{R_I})$. This substitution continues until the algorithm $\A$ has visited every candidate interval. Figure \ref{fig:reliever} illustrates the procedure.

\begin{figure}[htb]
\centering
\includegraphics[width=0.96\linewidth]{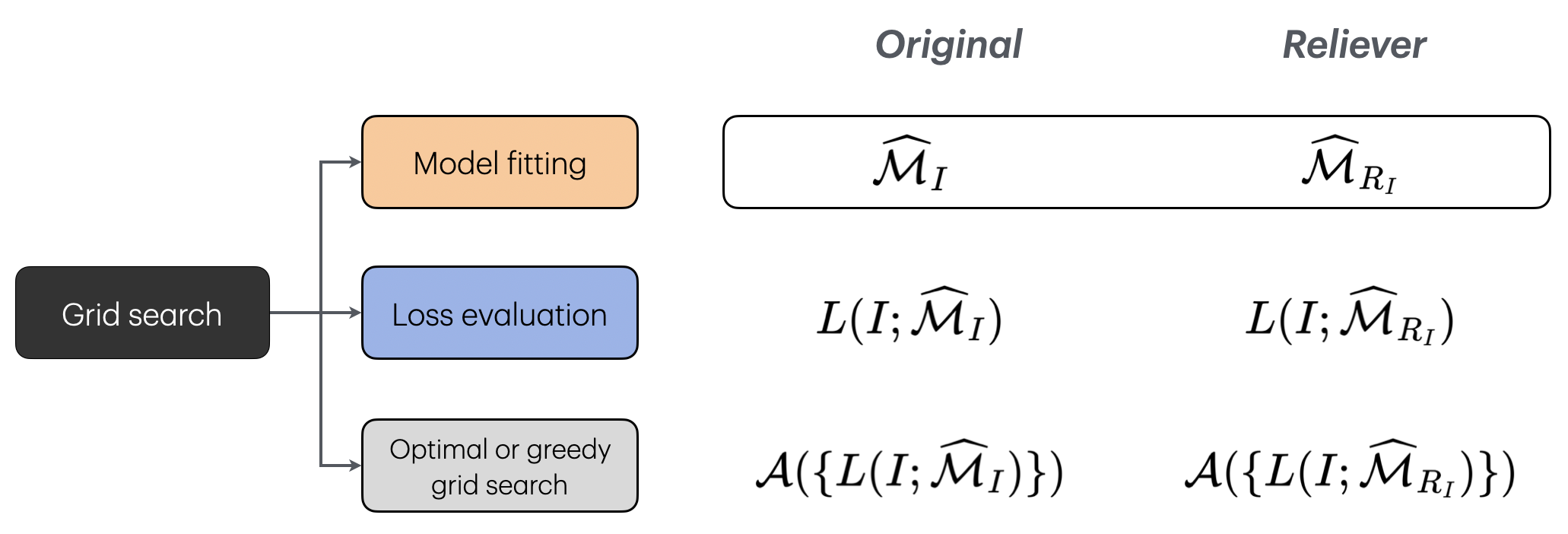}
\caption{Workflow comparison between a standard grid-search algorithm and the same algorithm equipped with Reliever.
}
\label{fig:reliever}
\end{figure}

The collection $\R$ is intentionally small---$|\R|=O(n)$ in our construction---so only $O(|\R|)$ actual model fits are required overall. Consequently, the total fitting cost drops to at most $O(na_n)$, where $a_n$ denotes the operations needed to fit a single model on an interval of length $n$; see Table \ref{tab:complexity}. In practice the cost is often lower because the algorithm $\A$ may visit only a subset of intervals in $\R$. Besides controlling the size of $\R$, each interval $R_I$ is selected such that $R_I\subset I$ and the reminder $I \setminus R_I$ is short. With this design, replacing the original loss sequence $\{\Lcal(I;\hM_I)\}$ by its proxy counterpart $\{\Lcal(I;\hM_{R_I})\}$ in the grid-search algorithm would preserve changepoint-detection accuracy, regardless of how many changepoints lie within any individual interval $I$.

To illustrate the benefits of Reliever, we consider a high-dimensional linear model with multiple changepoints (see Section \ref{subsec:hdlinear}), using $n=600$ observations and $p=100$ variables. Figure \ref{fig:comp_time}(a) compares the average computational time spent on model fits (including loss evaluations) with and without Reliever, alongside the average time spent solely on the grid search for each algorithm. Clearly, the primary computational burden arises from model fits, and employing Reliever substantially reduces this burden. Figure \ref{fig:comp_time}(b) further shows the reduction in the average number of model fits required along the search path. Finally, Figure \ref{fig:comp_time}(c) presents a boxplot of changepoint detection error, measured by the Hausdorff distance (see Section \ref{sec:sim}), confirming that Reliever significantly reduces computational cost without sacrificing detection accuracy.

\begin{figure}[htbp]
\centering
\includegraphics[width=0.99\linewidth]{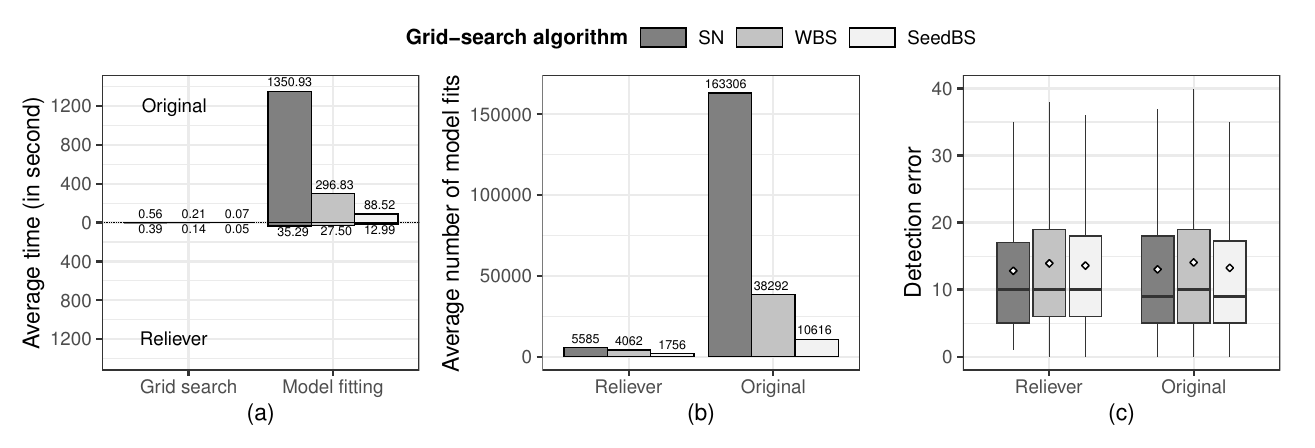}
\caption{Runtime savings and accuracy retention provided by Reliever in a high-dimensional linear model. (a) Average time spent on grid search in isolation and on model fits for the SN, WBS, and SeedBS algorithms, with and without Reliever. (b) Average number of model fits executed along each search path. (c) Changepoint detection error (Hausdorff distance); circles mark mean values.}
\label{fig:comp_time}
\end{figure}

\subsection{Our Contributions}\label{subsec:contribution}

We introduce Reliever, a highly flexible framework that speeds up changepoint detection whenever model fitting is the main cost. While earlier work cuts runtime by reducing the number of intervals that a grid-search algorithm visits, Reliever follows a complementary route: it fits models on only a pre-specified deterministic set of $O(n)$ intervals and reuses those fits wherever possible. Because the grid-search logic itself is left untouched, Reliever can be dropped straight into common optimal and greedy algorithms like SN, OP, PELT, WBS and SeedBS.

This sharp cut in model fits greatly shortens running time in both high-dimensional and nonparametric settings, while maintaining detection accuracy.
In the context of high-dimensional linear models with multiple changepoints---a topic that has garnered significant research interest---we demonstrate that Reliever combined with the OP algorithm \citep[for example,][]{leonardi_computationally_2016}, produces estimators for both changepoints and corresponding regression coefficients that are rate-optimal, up to a logarithmic factor.

\subsection{Related Works}\label{subsec:literature}

\noindent\textbf{Comparison with grid-pruning methods}. The collection of pre-specified and deterministic intervals used by Reliever resembles the \textit{seeded intervals} in SeedBS \citep{kovacs_seeded_2022}. These two acceleration ideas, however, serve different aims. \cite{kovacs_seeded_2022} refined WBS by replacing its random intervals with seeded intervals, targeting near-linear scaling of the grid search relative to the sample size (as outlined in Table \ref{tab:complexity}). A similar approach, involving the use of deterministic intervals for constructing scan statistics, is explored in \citet{MR3076173}. In contrast, our approach retains existing grid-search schemes but re-uses a proxy model fitted on one deterministic interval whenever that model is relevant to another data segment. This strategy allows Reliever to adapt to any grid-search algorithm, including SeedBS and WBS. Because the two kinds of deterministic intervals target different goals, their construction principles also differ; see Remark \ref{rmk:seedbs}.

\noindent\textbf{Comparison with two-step methods}. Our method's strategy to reduce intensive model fitting relates to two-step procedures that use a preliminary set of changepoint candidates. In the context of high-dimensional linear models with a single changepoint, \cite{kaul_efficient_2019-1} proposed an approach involving initial fitting of two regression models, one for data before and another after an initial changepoint estimator, followed by searching for the best split to minimize training error. To achieve nearly-optimal convergence rates for the resulting estimator, the initial estimator must be consistent.
For multiple changepoint scenarios, \cite{kaul_detection_2019-1} extended this approach by incorporating multiple initial candidates and employing a simulated annealing algorithm to allocate available model fits. This method presupposes proximities of all true changepoints to some of the initial candidates.
\cite{Cho+Owens-2023a} uses moving window on a coarse grid to scan for initial changepoint candidates and then refine their locations, whereas \cite{li2023divide} employs dynamic programming on a coarse grid to obtain initial candidates before refinement.
In the context of univariate mean change models, \cite{lu_intelligent_2020} introduced a method that leverages a sparse subsample to derive pilot changepoint estimators; for these pilot estimators to yield optimal changepoint estimators, they must accurately reflect both the number and locations of the changepoints.
In contrast, our Reliever framework does not rely on consistent initial estimators. It offers broad applicability and can be integrated as a foundational component in a variety of existing changepoint detection algorithms.

\subsection{Notation}\label{subsec:notation}

The $L_q$ norm of a vector $\zbf \in \Rbb^p$ is defined by $\norm{\zbf}_{q} = (\sum_{j=1}^p z_j^q)^{1/q}$. For a $p$-by-$p$ positive semi-definite matrix $\mathbf{A}$, we denote $\norm{\zbf}_{\mathbf{A}} = (\zbf^\top \mathbf{A} \zbf)^{1/2}$. The sub-Gaussian norm of a sub-Gaussian random variable $X$ is $\|X\|_{\Psi_2}=\inf\{t>0: \Ebb\{\exp(X^2/t^2)\}\le 2\}$.
The sub-Exponential norm of a sub-Exponential random variable $X$ is defined as $\norm{X}_{\Psi_1} = \inf\set{t > 0: \Ebb\exp(|X|/t) \le 2}$.
For a vector $\bm X\in\mathbb{R}^p$, define $\|\bm X\|_{\Psi_j}=\sup_{\bm v\in\mathbb{S}^{p-1}}\|\bm v^\top \bm X\|_{\Psi_j}$, where $\mathbb{S}^{p-1}$ is the unit sphere in $\Rbb^p$ and $j = 1, 2$.

\section{Methodology}\label{sec:method}

In this section, we first describe the general multiple changepoint models and algorithms with examples. Then we formally introduce the construction of the Reliever procedure.

\subsection{Changepoint Models and Grid-Search Algorithms}\label{subsec:model}

Consider a dataset $\set{\zbf_i}_{i=1}^n$ from a multiple changepoint model
\begin{equation}\label{MCP}
    \zbf_i \sim \M_k^\ast,\ \tau_{k-1}^\ast < i \le \tau_k^\ast,\ k=1,\ldots,K^\ast+1;\ i=1,\ldots,n,
\end{equation}
where $K^\ast$ and $\{\tau_k^\ast\}$ denote the number and locations of changepoints, respectively, with $\tau_0^\ast = 0$ and $\tau_{K^\ast+1}^\ast = n$. The notations $\{\M_k^\ast\}$ represent the models governing each data segment, ensuring that $\M_{k-1}^\ast\ne \M_k^\ast$. These models may describe nonparametric distributions of $\{\zbf_i\}$ or specific parametric forms with parameters $\{\btheta_k^\ast\}$, where $\btheta_{k-1}^\ast\ne \btheta_k^\ast$. For specific instances of these models, please refer to Examples \ref{ex:convex}--\ref{ex:nonpara}.

Changepoint detection typically proceeds through a grid-search process, involving a model fitting procedure, a loss function to evaluate fit quality, and a grid-search algorithm to determine the optimal segmentation, as illustrated in Figure \ref{fig:reliever}. For a candidate interval $I\subset(0,n]$, a model fitting procedure yields a fitted model $\hM_I$ (or $\widehat{\btheta}_{I}$ in parametric scenarios) based on the data segment $\{\zbf_i:i\in I\}$. The quality of this fit is evaluated using a loss function $\Lcal(I;\hM_I)$ (or $\Lcal(I;\widehat{\btheta}_{I})$ for parametric models).

\begin{example}[Parametric models, convex M-estimation]\label{ex:convex}
Consider the general parametric changepoint models within the framework of (\ref{MCP}), where each $\zbf_i\in\Rbb^p$ and
\begin{equation*}
    \zbf_i\ \text{has distribution}\ P_{\btheta_k^\ast},\ \tau_{k-1}^\ast < i \le \tau_k^\ast,\ k=1,\ldots,K^\ast+1;\ i=1,\ldots,n.
\end{equation*}
In scenarios with small $p$ and large $n$, one uses M-estimation for model fitting, which yields $\hat{\btheta}_I = \argmin_{\btheta \in \Theta} \sum_{i \in I} \ell(\zbf_i, \btheta)$, where $\ell(\zbf,\btheta)$ is a convex function with respect to $\btheta \in \Theta \subset \Rbb^p$. The loss evaluation is defined as $\Lcal(I;\hat{\btheta}_I)=\sum_{i \in I} \ell(\zbf_i, \hat{\btheta}_I)$. Employing convex losses, such as the absolute deviation or the Huber loss, is particularly effective in managing heavy-tailed observations or outliers for changepoint detection \citep{MR3941246}.
\end{example}

\begin{example}[High-dimensional linear models, lasso]\label{ex:lin_reg}
In (\ref{MCP}), each sample pair $\zbf_i=(y_i,\xbf_i)$ consists of a response $y_i \in \Rbb$ and covariates $\xbf_i \in \Rbb^p$, modeled by
\begin{equation}\label{MCP:LM}
    y_i = \xbf_i^\top \btheta_k^\ast + \epsilon_i,\ \tau_{k-1}^\ast < i \le \tau_k^\ast,\ k=1,\ldots,K^\ast+1;\ i=1,\ldots,n,
\end{equation}
where $\{\btheta_k^\ast\}$ are regression coefficients and $\{\epsilon_i\}$ denote random noises.
In high-dimensional settings where both $p$ and $n$ are large, lasso is used for model fitting, that is,
$$\hat{\btheta}_I = \argmin_{\btheta \in \Rbb^p}\{\sum_{i\in I} (y_i - \xbf_i^\top \btheta)^2 + \lambda_I \norm{\btheta}_1\},$$
with $\lambda_I$ as a tuning parameter. The loss evaluation function is specified as $\Lcal(I; \hat{\btheta}_I) = \sum_{i\in I} (y_i - \xbf_i^\top \hat{\btheta}_I)^2$. Detecting changes in high-dimensional linear models has recently garnered considerable attention; see, for example, \cite{leonardi_computationally_2016}, \cite{rinaldo_localizing_2021}, \cite{wang_statistically_2021} and \cite{xu2022change}.
\end{example}

\begin{example}[Nonparametric distributions]\label{ex:nonpara}
Within the framework of (\ref{MCP}), consider the samples $\zbf_i\equiv z_i\in\Rbb$ and
\begin{equation*}
    z_i\ \text{has a distribution function}\ F_k^\ast,\ \tau_{k-1}^\ast < i \le \tau_k^\ast,\ k=1,\ldots,K^\ast+1;\ i=1,\ldots,n.
\end{equation*}
Model fitting is performed using the empirical distribution function $\{\hat{F}_I(t):t\in\mathbb{R}\}$ for the data subset $\{z_i:i\in I\}$. The loss evaluation function, proposed by \cite{MR3210993}, is the negative of an integrated nonparametric maximum log-likelihood.
\end{example}

A grid-search algorithm is then employed to minimize a specific criterion over all possible segmented data sequences. This criterion typically comprises the sum of losses evaluated for each segment, along with a penalty that accounts for the complexity of the segmentation. To be specific, let $\mathcal{T}_K(\delta_{\mathsf{m}})=\{(\tau_1,\ldots,\tau_K):0\equiv\tau_0<\tau_1<\cdots<\tau_K<\tau_{K+1}\equiv n, \tau_{k+1}-\tau_k\ge \delta_{\mathsf{m}},\ k=0,\ldots,K\}$ be the set of $K$ candidate changepoints, where $\delta_{\mathsf{m}}>0$ is the minimal-spacing parameter; see Remark \ref{rmk:OnImplementations}. For $(\tau_1,\ldots,\tau_K)\in\mathcal{T}_K(\delta_{\mathsf{m}})$ that partitions the data into $K+1$ segments, the criterion is generally formulated as
\begin{align}\label{criterion}
    \sum_{k=1}^{K+1} \Lcal((\tau_{k-1},\tau_k];\hM_{k}) + \gamma K,
\end{align}
where $\gamma\ge 0$ controls the level of penalization to avoid overfitting. Optimal-kind algorithms (for example, SN, OP, or PELT, see Section \ref{sec:intro}) aim to find the exact minimizer over the entire search space $\mathcal{T}_K(\delta_{\mathsf{m}})$. This involves evaluating a sequence of losses (and fitting the corresponding models) for all $O(n^2)$ intervals $I\subset(0,n]$ satisfying $\size{I}\ge \delta_{\mathsf{m}}$, sequentially explored using a dynamic programming scheme. Although PELT uses a pruning strategy to skip certain intervals and thus reduces this complexity to $O(n)$, this reduction does not always apply (see Eq. (4) in \cite{killick_optimal_2012-1}). In contrast, greedy-kind algorithms, such as binary segmentation (BS), WBS, narrowest-over-threshold, or SeedBS, consider only a subset of these intervals in a sequential and greedy manner, aiming to reach a local minimizer. To illustrate, consider the BS algorithm. This algorithm begins by solving (\ref{criterion}) with $K=1$, which involves approximately $O(n)$ intervals. The resulting changepoint divides the data sequence into two segments. The algorithm then repeats the same procedure within each segment to identify new changepoints. This iterative process continues until a segment contains fewer observations than $\delta_{\mathsf{m}}$ or until a stopping rule is triggered. Overall, BS involves approximately $O(n\log n)$ intervals.
{Throughout, we regard any interval $I=(a,b]$ with integers $0 \le a < b \le n$ as a \textit{search interval} and collect all such candidates in the set $\I = \set{I: I\subset (0,n]}$. Given the collection of loss values $\{\Lcal(I;\hM_I):I\in\I\}$, the grid-search algorithm can be regarded as the operator $\A = \A(\{\Lcal(I;\hM_I):I\in\I\})$, which maps these evaluations to a final segmentation. In fact, $\A$ inspects only a subset of $\I$---either because intervals shorter than a minimal-spacing parameter $\delta_{\mathsf{m}}$ are excluded or because greedy strategies (for example, BS) deliberately restrict the search. Whenever it is necessary to distinguish between the two, we denote this subset actually explored by the algorithm by $\I_\A\subset\I$.}

The grid-search process becomes computationally demanding when model fits along the search path, $\{\hM_I:I\in\I_{\A}\}$, becomes costly. This is particularly evident in scenarios like those described in Examples \ref{ex:convex}--\ref{ex:nonpara}, where a single model fit requires substantial computational effort, and updating neighboring model fits by adding or removing observations remains elusive.

\subsection{Relief Intervals}\label{subsec:intervals}

Our approach is straightforward yet highly adaptable, and it integrates seamlessly with any grid-search algorithm $\A$. We begin by constructing a set of deterministic intervals $\R$. During the search process, for a search interval $I\in\I$, a proxy or \textit{Relief} model, $\hM_{R_I}$, fitted using data from an interval $R_I\in\R$, replaces $\hM_I$ when evaluating the loss $\Lcal(I;\hM_I)$.
Each interval $R_I\in\R$ is referred to as a \textit{relief interval} to distinguish it from a search interval $I$.
It is possible for multiple search intervals to correspond to a single relief interval, and not all relief intervals may be visited during the search.
The notation $R_I$ indicates that the selection of a relief interval is dependent on the current search interval $I$. For simplicity, we will use $R$ interchangeably with $R_I$, which should not lead to confusion. The key to the construction of $\R$ lies in satisfying two properties.
First, it significantly reduces the number of intervals for which a sequence of models needs to be fitted, as opposed to fitting models for every search interval.
Second, it ensures that the losses computed using the relief models exhibit behaviors similar to those computed with the original models, thereby allowing for consistent changepoint detection.

\begin{definition}[Relief intervals]
\label{itv_construction}
Let $\delta_{\mathsf{m}} > 0$ be the minimal-spacing parameter. Define $0 < w \le 1$ as the \textit{wriggle} parameter and $b > 1$ as the \textit{growth} parameter. For each $0 \le k \le \lfloor \log_{b} \{(1 + w) n/\delta_{\mathsf{m}}\} \rfloor$, construct the $k$th layer of relief intervals, consisting of $n_k$ intervals of length $\ell_k$, evenly shifted by $s_k$, that is, $\R_k = \set{(qs_k, qs_k+\ell_k] + a_k: 0 \le q \le n_k}$, where $\ell_k = b^k \delta_{\mathsf{m}}/(1+w)$, $s_k = w \ell_k$, $n_k= \lfloor(n - \ell_k)/s_k\rfloor$, and $a_k = n/2 - (\ell_k + n_k s_k)/2$ is an adjustment factor to center the intervals in $\R_k$ around $n/2$. The complete set of relief intervals is $\R = \bigcup_{k=0}^{\lfloor \log_{b} \{(1 + w) n/\delta_{\mathsf{m}}\} \rfloor} \R_k$.
\end{definition}

\begin{figure}[htbp]
    \centering
    \begin{tikzpicture}[scale=0.07]
    \pgfmathsetmacro{\ticker}{0.5}
    \draw[->](0,0)--(200,0);
    \foreach \i/\texti  in {40,80,...,160} {
        \draw (\i,-\ticker) --(\i,\ticker);
    }
    \foreach \i/\texti  in {20,60,...,180} {
        \draw (\i,-\ticker) --(\i,\ticker) node[label=below:\texti]{};
    }
    \draw[-](0,3)--(32,3);
    \draw[-](8,4)--(40,4);
    \draw[-](16,5)--(48,5);
    \draw[-](24,6)--(56,6);
    \draw[-](32,7)--(64,7);
    \draw[-](40,3)--(72,3);
    \draw[-](48,4)--(80,4);
    \draw[-](56,5)--(88,5);
    \draw[-](64,6)--(96,6);
    \draw[-](72,7)--(104,7);
    \draw[-](80,3)--(112,3);
    \draw[-](88,4)--(120,4);
    \draw[-](96,5)--(128,5);
    \draw[-](104,6)--(136,6);
    \draw[-](112,7)--(144,7);
    \draw[-](120,3)--(152,3);
    \draw[-](128,4)--(160,4);
    \draw[-](136,5)--(168,5);
    \draw[-](144,6)--(176,6);
    \draw[-](152,7)--(184,7);
    \draw[-](160,3)--(192,3);
    \draw[-](168,4)--(200,4);
    \draw[densely dotted](0,10.5)--(40,10.5);
    \draw[densely dotted](10,11.5)--(50,11.5);
    \draw[densely dotted](20,12.5)--(60,12.5);
    \draw[densely dotted](30,13.5)--(70,13.5);
    \draw[densely dotted](40,14.5)--(80,14.5);
    \draw[densely dotted](50,10.5)--(90,10.5);
    \draw[densely dotted](60,11.5)--(100,11.5);
    \draw[densely dotted](70,12.5)--(110,12.5);
    \draw[densely dotted](80,13.5)--(120,13.5);
    \draw[densely dotted](90,14.5)--(130,14.5);
    \draw[densely dotted](100,10.5)--(140,10.5);
    \draw[densely dotted](110,11.5)--(150,11.5);
    \draw[densely dotted](120,12.5)--(160,12.5);
    \draw[densely dotted](130,13.5)--(170,13.5);
    \draw[densely dotted](140,14.5)--(180,14.5);
    \draw[densely dotted](150,10.5)--(190,10.5);
    \draw[densely dotted](160,11.5)--(200,11.5);
    \draw[-](0,18)--(50,18);
    \draw[-](12.5,19)--(62.5,19);
    \draw[-](25,20)--(75,20);
    \draw[-](37.5,21)--(87.5,21);
    \draw[-](50,22)--(100,22);
    \draw[-](62.5,18)--(112.5,18);
    \draw[-](75,19)--(125,19);
    \draw[-](87.5,20)--(137.5,20);
    \draw[-](100,21)--(150,21);
    \draw[-](112.5,22)--(162.5,22);
    \draw[-](125,18)--(175,18);
    \draw[-](137.5,19)--(187.5,19);
    \draw[-](150,20)--(200,20);
    \draw[densely dotted](6.25,25.5)--(68.75,25.5);
    \draw[densely dotted](21.875,26.5)--(84.375,26.5);
    \draw[densely dotted](37.5,27.5)--(100,27.5);
    \draw[densely dotted](53.125,28.5)--(115.625,28.5);
    \draw[densely dotted](68.75,29.5)--(131.25,29.5);
    \draw[densely dotted](84.375,25.5)--(146.875,25.5);
    \draw[densely dotted](100,26.5)--(162.5,26.5);
    \draw[densely dotted](115.625,27.5)--(178.125,27.5);
    \draw[densely dotted](131.25,28.5)--(193.75,28.5);
    \draw[-](2.34375,33)--(80.46875,33);
    \draw[-](21.875,34)--(100,34);
    \draw[-](41.40625,35)--(119.53125,35);
    \draw[-](60.9375,36)--(139.0625,36);
    \draw[-](80.46875,37)--(158.59375,37);
    \draw[-](100,33)--(178.125,33);
    \draw[-](119.53125,34)--(197.65625,34);
    \draw[densely dotted](2.34375,40.5)--(100,40.5);
    \draw[densely dotted](26.7578125,41.5)--(124.4140625,41.5);
    \draw[densely dotted](51.171875,42.5)--(148.828125,42.5);
    \draw[densely dotted](75.5859375,43.5)--(173.2421875,43.5);
    \draw[densely dotted](100,44.5)--(197.65625,44.5);
    \draw[-](8.447265625,48)--(130.517578125,48);
    \draw[-](38.96484375,49)--(161.03515625,49);
    \draw[-](69.482421875,50)--(191.552734375,50);
    \draw[densely dotted](4.632568359375,53.5)--(157.220458984375,53.5);
    \draw[densely dotted](42.779541015625,54.5)--(195.367431640625,54.5);
    \draw[-](4.632568359375,58)--(195.367431640625,58);
\end{tikzpicture}
\caption{Illustration of relief intervals with $n = 200, \delta_{\mathsf{m}}=50, w = 0.25$ and $b = 1.25$.}
\label{fig:itv_25}
\end{figure}

The rationale behind the construction of relief intervals is to ensure that for any search interval $I\in\I$ with $\size{I}\ge\delta_{\mathsf{m}}$, there is always a relief interval $R\in\R$ such that $R\subset I$ and $\size{R}/\size{I}$ is maximized.
We define the \textit{coverage ratio} as
$$r = \min_{I\in\I: \size{I}\ge\delta_{\mathsf{m}}} \max_{R\in\R; R\subset I} \frac{\size{R}}{\size{I}},\qquad 0<r\le 1.$$

\begin{proposition}\label{prop:relief_interval}
(i) In general, $\size{\R} \le c_{w,b} n/\delta_{\mathsf{m}}$ and $r\ge\{(1+w)b\}^{-1}$, where $c_{w,b}=\{(1+w)b\}/\{w(b-1)\}$.
(ii) {Setting $(1+w)=b=r^{-\frac{1}{2}}$, $\size{\R} = \{n (r^{-\frac{1}{2}} - 1)^2\} / (\delta_{\mathsf{m}} r)$. Additionally, if $\delta_{\mathsf{m}}$ and $r \in (0, 1)$ are fixed constants,} then $\size{\R} = O(n)$.
(iii) If $\delta_{\mathsf{m}} = C\log n$ for some constant $C>0$ and $w = b - 1 = \delta_{\mathsf{m}}^{-\frac{1}{2}}$, then $\size{\R}\le n\{1 + (C \log n)^{-\frac{1}{2}}\}^2 = O(n)$ and $r\ge\{1 + (C \log n)^{-\frac{1}{2}}\}^{-2}\approx 1 - 2(C \log n)^{-\frac{1}{2}}$.
\end{proposition}

Proposition \ref{prop:relief_interval} demonstrates that, by selecting appropriate wriggle and growth parameters, alongside a minimal-spacing parameter, the number of relief intervals approaches linearity with the sample size $n$ while achieving a nearly perfect coverage ratio. To facilitate practical applications, setting a single coverage ratio parameter $r\in(0, 1)$ is sufficient, with $1+w=b=r^{-\frac{1}{2}}$. {Figure \ref{fig:itv_25} illustrates the construction of relief intervals with $n = 200$, $\delta_{\mathsf{m}}=50$, $w = 0.25$, and $b = 1.25$ (corresponding to $r = 0.64$).} The parameter $r$ balances computational complexity and estimation accuracy; see Remark \ref{rmk:tradeoff}. Table \ref{tab:num_interval} reports, for $n=1200$ and $\delta_{\mathsf{m}}=30$, the number of search intervals examined by the original SN implementation ($r=1$; that is, $|\I_\A|$ for SN), as well as the corresponding number of relief intervals obtained for various coverage ratios of $r$.

\begin{table}[htb]
\setlength\tabcolsep{0pt}
\centering
\begin{threeparttable}
\caption{\small Number of model fits required the SN algorithm: Reliever for various coverage ratios $r$ versus the original implementation ($r=1$).}
\label{tab:num_interval}
\tabcolsep=0.44em
\begin{tabular*}{.97\linewidth}{c@{\extracolsep{\fill}}*{9}{c}}
\toprule
$r$ & 0.5 & 0.6 & 0.7 & 0.8 & 0.9 & 0.95 & 0.97 & 0.99 & 1\\
\midrule
Model fits & 440 & 762 & 1,298 & 2,744 & 12,227 & 31,699 & 57,522 & 196,395 & 686,206 \\
\bottomrule
\end{tabular*}
\end{threeparttable}
\end{table}

\begin{remark}[Comparison with seeded intervals]\label{rmk:seedbs}
The deterministic nature of our relief intervals is conceptually inspired by the seeded intervals in SeedBS \citep{kovacs_seeded_2022}, yet they diverge significantly in their design principles due to differing objectives. The seeded intervals were developed to replace random (wild) intervals in WBS, focusing on shorter intervals that typically contain a single changepoint to reduce the occurrence of longer intervals that may encompass multiple changepoints. In contrast, the relief intervals are designed to ensure that each search interval closely matches a relief interval of similar length, thereby producing comparable loss values. Our approach is compatible with various grid-search algorithms, including WBS and SeedBS.
A natural question arises: can seeded intervals serve as proxy intervals in place of relief intervals within the proposed Reliever framework? Figure \ref{fig:seeded}(a) displays the coverage ratio $r$ against the number of proxy intervals, revealing that starting at $r \approx 0.5$, our construction achieves a higher coverage ratio with the same number of intervals. Notably, using seeded intervals limits the coverage ratio to approximately $0.68$, even with an increased number of intervals. Figure \ref{fig:seeded}(b) compares the changepoint detection error between the two constructions, both utilizing $1,000$ proxy intervals, across various grid-search algorithms, under a high-dimensional linear model with multiple changepoints (as described in Section \ref{subsec:hdlinear}). This comparison demonstrates that our construction, with its higher coverage, typically yields better detection accuracy.
\end{remark}

\begin{figure}[!ht]
    \centering
    \begin{subfigure}{0.46\textwidth}
    \centering
    \includegraphics[width=.99\linewidth]{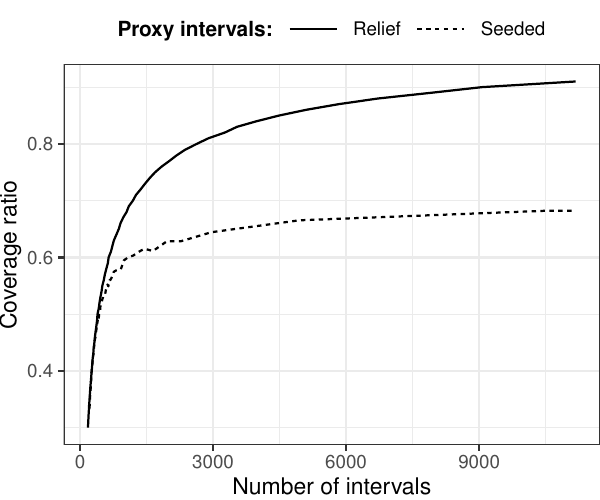}
    \caption{}
    \end{subfigure}
    \begin{subfigure}{0.46\textwidth}
    \centering
    \includegraphics[width=0.99\linewidth]{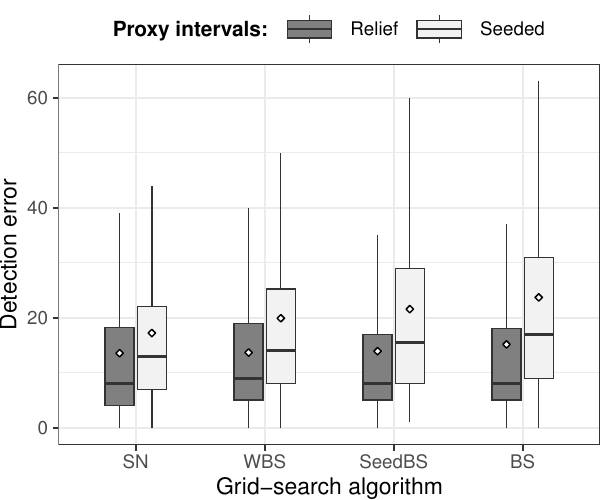}
     \caption{}
    \end{subfigure}
    \caption{Comparison of coverage ratio and changepoint detection error for two proxy interval strategies within the Reliever framework.}
    \label{fig:seeded}
\end{figure}

\subsection{The Reliever Procedure}\label{subsec:relief}

\begin{enumerate}[label=(\arabic*)]
    \item Require a gird search algorithm $\A=\A(\{\Lcal(I;\hM_I):I\in\I\})$ with a minimal-spacing parameter $\delta_{\mathsf{m}}\ge 0$ and a model-fitting procedure $\hM_I$ for any interval $I$ such that $\size{I}\ge\delta_{\mathsf{m}}$, and a coverage ratio parameter $r \in (0, 1]$;
    \item Create relief intervals $\R$ with the wriggle and growth parameters satisfying $1+w=b=r^{-\frac{1}{2}}$ according to Definition \ref{itv_construction};
    \item Execute the gird search with relief models, that is, $\A=\A(\{\Lcal(I;\hM_{R_I}):I\in\I\})$ with $R_I = \argmax\nolimits_{R \in \R, R \subset I} \size{R}$.
\end{enumerate}

The Reliever procedure can be used in conjunction with both optimal- and greedy-kind grid-search algorithms, as outlined in Section \ref{sec:intro}, represented by $\A=\A(\{\Lcal(I;\hM_I):I\in\I\})$. The primary distinction when using Reliever compared to the original implementation is the employment of a relief model $\hM_{R_I}$ to evaluate the loss function $\Lcal(I;\hM_I)$. This modification not only maintains versatility but also significantly reduces the number of model fits required, as each loss evaluation leverages a pre-fitted model from a relief interval, dramatically lowering computational overhead compared to the original implementations (refer to Table \ref{tab:complexity} for computational comparisons).

\begin{remark}[On implementations]\label{rmk:OnImplementations}
The minimal-spacing parameter $\delta_{\mathsf{m}}$ is an inherent component of the grid-search algorithm, ensuring that the model/parameter is estimatable on each segment $I$. For example, in linear models, \cite{bai1998estimating} set $\delta_{\mathrm{m}} = p$ (with $p$ the covariate dimension) so that the least-squares estimator is well-posed. Our construction of relief intervals imposes no additional practical restriction on $\delta_{\mathrm{m}}$.
Theoretically, large $\delta_{\mathrm{m}}$ guarantees that $\hM_{I}$ behaves well and converges to its population counterpart as $|I|$ grows---a standard requirement in nonparametric and high-dimensional settings. For instance, in multivariate nonparametric changepoint detection via kernel density, \cite{padilla2023change} use $\delta_{\mathrm{m}} = h^{-p} \log(n)$ (with $h$ the bandwidth and $p$ the data dimension). In high-dimensional linear models with temporal dependence, \cite{xu2022change} set $\delta_{\mathrm{m}} = O((s \log(p \vee n))^{2/\zeta - 1})$ (with $p$ the covariate dimension, $s$ the sparsity level, and $\zeta$ characterizing dependence and noise tails), which simplifies to $O(s \log(p \vee n))$ in the independent, sub-Gaussian case ($\zeta=1$). For numerical stability, we use $\delta_{\mathrm{m}} = 20$ in simulation studies (Section \ref{sec:sim}). In Section~\ref{sec:small_delta_m}, we investigate robustness to small choice of $\delta_{\mathsf m}$.

The construction of relief intervals also hinges on the coverage-ratio parameter $r$, which governs the trade-off between computational cost and estimation accuracy (see Remark \ref{rmk:tradeoff}). We recommend viewing $r$ as a \textit{budget} parameter. With ample computing time, choose $r\approx 1$ (which recovers the original algorithm). When runtime is critical, pick the largest $r$ that satisfies the time constraints. Extensive experiments show that $0.8 < r < 0.9$ typically cuts runtime substantially while maintaining satisfactory detection accuracy compared to the original implementation; see Section~\ref{sec:sim}.
When the set of grid-search intervals is known (or easily predicted) in advance, the total fitting time for Reliever can be approximated. Practitioners may then tune $r$ to match a target runtime with reasonable confidence. Implementation details are provided in Section~\ref{sec:estimated_time}.
\end{remark}

\section{Theoretical Justifications}\label{subsec:example}

Despite the broad applicability of Reliever across various changepoint detection algorithms and model settings, establishing a unified theoretical framework for analyzing detection accuracy is challenging without specific assumptions regarding the model, the fitting procedure, and the grid-search algorithm. Here, we first offer an indirect justification by examining the variations in loss values resulting from the application of Reliever. This examination focuses on parametric changepoint models with convex minimization routines for model fitting, as demonstrated in Example \ref{ex:convex}. Furthermore, in Section \ref{sec:hdLinear}, we present rigorous theoretical results concerning the estimation accuracy of multiple changepoints in the context of high-dimensional linear models employing lasso, as detailed in Example \ref{ex:lin_reg}. {Unless stated otherwise, we assume that the observations $\zbf_i$ are temporally independent.}

\subsection{Variations in Loss Values: Convex Minimization}

Continue with Example \ref{ex:convex}. In the original implementation of a grid-search algorithm $\A$, losses are evaluated as $\Lcal(I; \hat{\btheta}_I) = \sum_{i \in I} \ell(\zbf_i, \hat{\btheta}_I)$, where $\hat{\btheta}_I = \argmin_{\btheta \in \Theta} \sum_{i \in I} \ell(\zbf_i, \btheta)$. With the Reliever approach, these loss calculations are replaced by $\Lcal(I; \hat{\btheta}_{R}) = \sum_{i \in I} \ell(\zbf_i, \hat{\btheta}_{R})$, where $R \in \R$ is a relief interval corresponding to $I$.
To analyze the impact of this substitution, let $\overline{\Lcal}(I, \btheta) = \Ebb \Lcal(I; \btheta)$ and $\btheta_{I}^\circ = \argmin_{\btheta \in \Theta} \overline{\Lcal}(I, \btheta)$ denote population loss and its minimizer, respectively. Define $G_{I}(\balpha) = \size{I}^{-1} \sum_{i \in I} g(\zbf_i, \btheta_{I}^\circ + \balpha \size{I}^{-\frac{1}{2}})$ and $\overline{G}_{I}(\balpha) = \Ebb G_{I}(\balpha)$, where $g(\zbf, \btheta) = \nabla_{\btheta} \ell(\zbf, \btheta)$.

Under mild regularity conditions for convex M-estimation, Proposition \ref{thm:mest} shows that the difference between the losses $\Lcal(I; \hat{\btheta}_I)$ and $\Lcal(I; \hat{\btheta}_{R})$ remains uniformly controlled over all search intervals $I\in\I$, regardless of whether $I$ contains changepoints.

\begin{enumerate}[label=(\alph*)]
    \item $\ell(\cdot, \zbf)$ is convex on the domain $\Theta$ for all fixed $\zbf$ and $\Theta$ is a compact and convex subset of $\Rbb^p$.
    \item The expectation $\Ebb\ell(\zbf_i, \btheta)$ is finite for all $\zbf_i$ and fixed $\btheta \in \Theta$.
    \item The population minimizer $\btheta_{I}^\circ$ uniquely exists and is interior point of $\Theta$.
    \item $\norm{g(\zbf_i, \btheta)}_{\Psi_1} \le C_{A.1}$ for each $\btheta$ near $\btheta_{I}^\circ$.
    \item $\overline{\Lcal}(I, \btheta)$ is twice differentiable at $\btheta_{I}^\circ$ and $\Hbf_{I} \triangleq \size{I}^{-1}  \nabla_{\btheta}^2 \overline{\Lcal}(I, \btheta_{I}^\circ)$ is positive-define.
    \item $\abs{\overline{G}_I(\size{I}^{\frac{1}{2}}(\btheta - \btheta_{I}^\circ)) - \Hbf_{I} (\btheta - \btheta_{I}^\circ)} = C_{A.2} \norm{{\btheta - \btheta_{I}^\circ}}_{2}^2$.
    \item $\norm{g(\zbf_i, \btheta) - g(\zbf_i, \btheta_{I}^\circ)}_{\Psi_1} \le C_{A.3} \norm{\btheta - \btheta_{I}^\circ}_2$.
    \item $\size{I}^{-1} \overline{\Lcal}(I, \btheta)$ is $\rho$-strongly convex in the compact set $\Theta$.
    \item $\Ebb g(\zbf_i, \btheta)$ is $\zeta$-Lipschitz continuous w.r.t. $\btheta$.
    \item For $i \in I \setminus R$, $\norm{\btheta_{R}^\circ - \btheta_{i}^\circ}_2 \le \Delta_{\infty}$ where $\Delta_{\infty} > 0$ is a fixed constant.
    \item $\norm{\Hbf_{R}^{-1} - \Hbf_{I}^{-1}}_{\mathsf{op}} \le C_{A.4} \norm{\btheta_{R}^\circ - \btheta_{I}^\circ}_{2}$ and $\norm{\Hbf_{I}^{-1}}_{\mathsf{op}} \le C_{A.5}$ for any interval $I$.
\end{enumerate}

\begin{proposition}\label{thm:mest}
    Given that the conditions (a)--(k) hold, with probability at least $1 - n^{- C}$ for some constant $C>0$, the event
    \begin{equation}\label{mest:general}
        0 \le \frac{1}{\size{I}} \{\Lcal(I; \hat{\btheta}_R) - \Lcal(I; \hat{\btheta}_I)\} \le O\biggl(\biggnorm{\frac{1}{\size{I}} \sum_{i \in I \setminus R} \Ebb \nabla_{\btheta} \ell(\zbf_i, \btheta_{R}^\circ)}_2^2 + \frac{(1 - r)\log n}{r\size{I}} + \frac{(\log n)^2}{r^2\size{I}^2}\biggr)
    \end{equation}
    holds uniformly for all intervals $R \subset I \in \I$, where $\nabla_{\btheta} \ell(\cdot,\btheta)$ denotes the gradient or sub-gradient. In particular, in scenarios where $I = (s, e]$ contains no changepoint, or there is only one changepoint $\tau \in I$ such that $\min(\tau - s, e - \tau) = O(\log n)$, this event simplifies to
    \begin{equation}\label{mest:homo}
        0 \le \frac{1}{\size{I}}\{\Lcal(I; \hat{\btheta}_{R}) - \Lcal(I; \hat{\btheta}_I)\} \le C_1 \biggl(\frac{(1 - r)\log n}{r\size{I}} + \frac{(\log n)^2}{r^2\size{I}^2}\biggr),
   \end{equation}
   where $C_1 > 0$ is a constant.
\end{proposition}

\begin{remark}[On Conditions (a)--(k)]\label{rmk:condition_mest}
Conditions (a)--(k) parallel the regularity assumptions of \citet{niemiro_asymptotics_1992}, which analyzed M-estimators obtained through convex minimization under independent and identically distributed (i.i.d.) data. These conditions fundamentally concern the smoothness and convexity of the loss function $\ell$ and its expectation.
Conditions (a)--(f) are the standard ingredients for deriving estimation error bounds in classical convex M-estimation with i.i.d. samples. The additional conditions (g)--(k) address distributional heterogeneity created by changepoints and ensure uniform control of the difference terms $\norm{\btheta_I^\circ-\btheta_R^\circ}_2$ and $\norm{\hat{\btheta}_I-\hat{\btheta}_R}_2$ for all $I \in \I$.
The proof of Proposition \ref{thm:mest} relies on a novel non-asymptotic Bahadur-type representation for $\hat{\btheta
}_I - \hat{\btheta}_R$ in the presence of changepoints across all $I\in\I$, which may be of independent interest.
\end{remark}

In Proposition \ref{thm:mest}, Eq.(\ref{mest:homo}) indicates that the discrepancy between the Reliever-based loss $\Lcal(I; \hat{\btheta}_{R})$ and the original loss $\Lcal(I; \hat{\btheta}_I)$ vanishes when the data within $I$ are (nearly) homogeneous and $(\log n) / \size{I}$ goes to zero. This provides a justification for employing Reliever. Conversely, for heterogeneous $I$ containing changepoints distant from the boundaries, this vanishing property of the discrepancy may not hold. Surprisingly, the inequality $\Lcal(I; \hat{\btheta}_R)\ge \Lcal(I; \hat{\btheta}_I)$ in Eq.(\ref{mest:general}) becomes instrumental in excluding inconsistent changepoint estimators in such scenarios. Therefore, we can expect that Reliever effectively tracks the original search path. To gain some intuition, consider the scenario with a single changepoint $\tau^\ast$ such that $\min(\tau^\ast, n - \tau^\ast) \ge \delta_{\mathsf{m}}$ or $\tau^\ast\in\mathcal{T}_1(\delta_{\mathsf{m}})$. And the grid-search algorithm is specified as the first BS step. Define the changepoint estimator as $\hat{\tau}_{\mathrm{original}} = \argmin_{\tau\in\mathcal{T}_1(\delta_{\mathsf{m}})} S^{(I)}_{I}(\tau)$, where $S^{(I)}_{I}(\tau)=\Lcal(I_{1, \tau}, \hat{\btheta}_{I_{1, \tau}}) + \Lcal(I_{2, \tau}, \hat{\btheta}_{I_{2, \tau}})$, and for any $\tau$, $I_{1, \tau} = (0, \tau]$ and $I_{2, \tau} = (\tau, n]$. The Reliever-based changepoint estimator is then defined as $\hat{\tau} = \argmin_{\tau\in\mathcal{T}_1(\delta_{\mathsf{m}})} S^{(R)}_{I}(\tau)$, where $S^{(R)}_{I}(\tau)=\Lcal(I_{1, \tau}, \hat{\btheta}_{R_{1, \tau}}) + \Lcal(I_{2, \tau}, \hat{\btheta}_{R_{2, \tau}})$, and $R_{j, \tau}\subset I_{j, \tau}$ is the corresponding relief interval for $j=1,2$. We present the following corollary, which establishes the consistency of $\hat{\tau}$ in the sense that $|\hat{\tau} - \tau^\ast| / n \rightarrow 0$.

\begin{corollary}\label{cor:scp}
    Assume $\delta_{\mathsf{m}} = C_{\mathsf{m}} \log n$ for some constant $C_{\mathsf{m}} > 0$ and the event described in Proposition \ref{thm:mest} holds. If there exists a sufficiently large constant $C_2 > 0$ such that for any $\tau\in\mathcal{T}_1(\delta_{\mathsf{m}})$ satisfying $\abs{\tau - \tau^\ast} > \delta$ for a constant $\delta > 0$,
    \begin{align}\label{eq:losslogn}
        S^{(I)}_{I}(\tau) - S^{(I)}_{I}(\tau^\ast) > C_2 \log n
    \end{align}
    holds, then $\abs{\hat{\tau} - \tau^\ast} \le \delta$.
\end{corollary}

Corollary \ref{cor:scp} is a direct consequence of Proposition \ref{thm:mest}. Assume $\abs{\hat{\tau} - \tau^\ast} > \delta$. Since $\Lcal(I; \hat{\btheta}_R)\le \Lcal(I; \hat{\btheta}_I)$ according to Eq.(\ref{mest:general}), it follows that $S^{(R)}_{I}(\hat{\tau})\ge S^{(I)}_{I}(\hat{\tau})$. Utilizing Eq.(\ref{mest:homo}), we derive
\begin{align*}
    S^{(R)}_{I}(\tau^\ast)
    \le S^{(I)}_{I}(\tau^\ast) + 2C_1 \Big\{\frac{1-r}{r} + \frac{n \log n}{\tau^\ast (n - \tau^\ast)r^2}\Big\} \log n.
\end{align*}
Considering Eq.(\ref{eq:losslogn}), by setting $C_2 \ge 2C_1 \{(1 - r)r^{-1} + C_{\mathsf{m}}^{-1}r^{-2}\}$, we have $S^{(R)}_{I}(\hat{\tau})-S^{(R)}_{I}(\tau^\ast) > [C_2 - 2C_1 \{(1 - r)r^{-1} + C_{\mathsf{m}}^{-1}r^{-2}\}] \log n \ge 0$. Therefore, the assumption $\abs{\hat{\tau} - \tau^\ast} > \delta$ leads to a contradiction, establishing the validity of Corollary \ref{cor:scp}. Eq.(\ref{eq:losslogn}) imposes implicit constraints on the model, ensuring that the original grid-search algorithm produces a consistent changepoint estimator, that is, $\abs{\hat{\tau}_{\mathrm{original}} - \tau^\ast} \le \delta$. Verifying Eq.(\ref{eq:losslogn}) or establishing a lower bound for $S^{(I)}_{I}(\tau) - S^{(I)}_{I}(\tau^\ast)$ is a well-accepted technique for justifying the consistency of changepoint estimators \citep{MR2743035}. Corollary \ref{cor:scp} demonstrates that the consistency proof for the original grid-search algorithm can readily be extended to the Reliever estimator. This approach is also applicable to multiple changepoint detection tasks.

\subsection{Changepoint Detection Accuracy: Lasso Regression}\label{sec:hdLinear}

To deepen our understanding of how variations in loss values impact the accuracy of changepoint detection with the Reliever method, we delve into the detection of multiple changepoints in high-dimensional linear models (cf. Example \ref{ex:lin_reg}).

We use the OP algorithm \citep[for example,][]{leonardi_computationally_2016}. The standard implementation minimizes the criterion
\begin{align}\label{object_cp_full_linear}
    \sum_{k=1}^{K+1} \Lcal((\tau_{k-1},\tau_k];\hat{\btheta}_{(\tau_{k-1},\tau_k]}) + \gamma K,
\end{align}
over all candidate changepoints $(\tau_1,\ldots,\tau_K)\in\mathcal{T}_K(\delta_{\mathsf{m}})$. For any search interval $I\in\I$ with $\size{I}\ge\delta_{\mathsf{m}}$, model parameters are estimated as $\hat{\btheta}_I = \argmin_{\btheta \in \Rbb^p}\{\sum_{i\in I} (y_i - \xbf_i^\top \btheta)^2 + \lambda_I \norm{\btheta}_1\}$, and loss values are computed as $\Lcal(I; \hat{\btheta}_I) = \sum_{i\in I} (y_i - \xbf_i^\top \hat{\btheta}_I)^2$. The tuning parameter $\gamma$ helps avoid overestimating the number of changepoints. Specific values for $\lambda_I$ and $\gamma$ are established through a rigorous theoretical analysis discussed later in this section.

To integrate Reliever with OP, we construct a collection of relief intervals $\R$ with a fixed coverage ratio parameter $0<r<1$. The optimization criterion within the Reliever framework is thus reformulated as
\begin{align}\label{object_cp_part_linear}
    \sum_{k=1}^{K+1} \Lcal((\tau_{k-1},\tau_k];\hat{\btheta}_{R_k}) + \gamma K,\mbox{ with } R_k = \argmax\nolimits_{R \in \R, R \subset (\tau_{k-1},\tau_k]} \size{R}.
\end{align}
This criterion and the original in Eq.(\ref{object_cp_full_linear}) are specific instances of a more general optimization problem
\begin{align}\label{object}
    \min_{(\tau_1,\ldots,\tau_K)\in\mathcal{T}_K(\delta_{\mathsf{m}})}
    \left\{\sum_{k=1}^{K+1} \Lcal\Big((\tau_{k-1},\tau_k];\widetilde{\btheta}\big((\tau_{k-1},\tau_k]\big)\Big) + \gamma K\right\}.
\end{align}
Here, $\widetilde{\btheta}(I)$ can be any valid estimator of regression coefficients for $I\in\I$ such that $\size{I}\ge\delta_{\mathsf{m}}$. Setting $\widetilde{\btheta}(I)=\hat{\btheta}_I$ recovers the original criterion (\ref{object_cp_full_linear}). Alternatively, choosing $\widetilde{\btheta}(I)=\hat{\btheta}_{R_I}$ with $R_I = \argmax\nolimits_{R \in \R, R \subset I} \size{R}$, gives us the reformulated problem (\ref{object_cp_part_linear}). OP manages this optimization (\ref{object})  by integrating a sequence of parameter estimation $\{\widetilde{\btheta}(I)\}$ and loss evaluation $\{\Lcal_I\equiv \Lcal(I;\widetilde{\btheta}(I))\}$ steps along the search path, with dynamic ordering of intervals $I$ determined by OP itself.

We first establish a deterministic claim about the consistency and near rate-optimality of the resulting changepoint estimators, conditional on an event measuring the quality or \textit{goodness} of the evaluated losses.
We introduce some notations and conditions crucial for this analysis.
For any search interval $I\in\I$, denote $\btheta_I^\circ = \argmin_{\btheta \in \Rbb^p} \Ebb\{\Lcal(I; \btheta)\}$, and define $\Delta_{I} = ({\size{I}}^{-1}\sum_{i \in I} \norm{\btheta_i^\circ - \btheta_I^\circ}_{\Sigma}^2)^{\frac{1}{2}}$, where $\btheta_i^\circ = \btheta_{\set{i}}^\circ$ for $i=1,\ldots,n$. For $k=1,\ldots,K^\ast$, let $\Delta_k = \norm{\btheta_{k+1}^\ast - \btheta_{k}^\ast}_{\Sigma}$ be the magnitude of change at $\tau_k^\ast$, with $\Delta_0 = \Delta_{K^\ast + 1} = \infty$.

\begin{condition}[Change signals]\label{cond:change}
    There exists a sufficiently large constant $C_{\mathsf{snr}}>0$ such that for $k = 1,\dots,K^\ast+1$, $\tau_k^\ast - \tau_{k-1}^\ast \ge C_{\mathsf{snr}} s\log(p \vee n) (\Delta_{k-1}^{-2} + \Delta_k^{-2} + 1)$.
\end{condition}
\begin{condition}[Regression coefficients]\label{cond:parameter}
    (a) Sparsity: $\size{\Scal_k} \le s < p$, where $\Scal_k = \set{1\le j\le p: \theta_{k,j}^\ast \neq 0}$ and $\theta_{k,j}^\ast$ is the $j$th component of $\btheta_k^\ast$;
    (b) Boundness: $\abs{\theta_{k,j}^\ast} \le C_\theta$ for some constant $C_\theta>0$.
\end{condition}
\begin{condition}[Covariates and noises]\label{cond:dist}
    (a) Covariates $\set{\xbf_i}_{i=1}^n$ are i.i.d. sub-Gaussian with zero mean and covariance $\Sigma$, satisfying $0<\kmin\le \sigma_x^2<\infty$, where $\kmin=\lambda_{\min}(\Sigma)$ and $\sigma_x^2=\lambda_{\max}(\Sigma)$ are the minimum and maximum eigenvalues of $\Sigma$, respectively. Furthermore, $\norm{\Sigma^{-\frac{1}{2}}\xbf_i}_{\Psi_2} \le C_x$ for some constant $C_x>0$;
    (b) Noises $\set{\epsilon_i}_{i=1}^n$ are i.i.d. sub-Gaussian with zero mean, variance $\sigma_\epsilon^2$, and a sub-Gaussian norm $C_{\epsilon}$.
\end{condition}

These conditions are commonly adopted in the literature for changepoint detection in high-dimensional linear models \citep{leonardi_computationally_2016,wang_statistically_2021,rinaldo_localizing_2021,xu2022change}. Specifically, Condition \ref{cond:change} introduces a \textit{local multiscale} signal-to-noise ratio (SNR) requirement for the spacing between neighboring changepoints, providing greater flexibility compared to the global SNR condition in existing works like \cite{leonardi_computationally_2016} and \cite{wang_statistically_2021}.

\begin{lemma}[Goodness of loss evaluations]\label{lem:loc_err_g}
    Under Condition \ref{cond:change}, for the problem~(\ref{object}) with $\delta_{\mathsf{m}} = C_{\mathsf{m}} s \log(p \vee n)$ for a sufficiently large constant $C_{\mathsf{m}}>0$, and $\gamma = C_{\gamma} s \log(p \vee n)$ for a constant $C_{\gamma}>0$, the solution $(\hat{\tau}_1,\ldots,\hat{\tau}_{\hat{K}})$ satisfies:
    \begin{equation*}
        \widehat{K} = K^\ast \mbox{ and } \max_{1 \le k \le K^\ast} \min_{1 \le j \le \widehat{K}} \frac{1}{2} \Delta_{k}^2 \abs{\tau_k^\ast - \hat{\tau}_j} \le \widetilde{C} s \log(p \vee n),
    \end{equation*}
    for some constant $\widetilde{C} > 0$, conditional on the event $\mathbb{G}=\mathbb{G}_1\cap\mathbb{G}_2^- \cap \mathbb{G}_2^+ \cap\mathbb{G}_3$. Here,
    \begin{align*}
        \mathbb{G}_1 &= \left\{\mbox{for any }I\in E_1, \Bigabs{\Lcal_I - \sum_{i \in I} \epsilon_i^2} < C_{\ref*{lem:loc_err_g}.1} s \log(p \vee n)\right\},\\
        \mathbb{G}_2^- &= \left\{\mbox{for any }I\in E_2^-, \Lcal_I - \sum_{i \in I} \epsilon_i^2 - \Delta_{I}^2 \size{I} > - C_{\ref*{lem:loc_err_g}.2} s \log(p \vee n)\right\},\\
        \mathbb{G}_2^+ &= \left\{\mbox{for any }I\in E_2^+, \Lcal_I - \sum_{i \in I} \epsilon_i^2 - \Delta_{I}^2 \size{I} < C_{\ref*{lem:loc_err_g}.2} s \log(p \vee n)\right\},\\
        \mathbb{G}_3 &= \left\{\mbox{for any }I\in E_3, \Lcal_I - \sum_{i \in I} \epsilon_i^2 > (1 - C_{\ref*{lem:loc_err_g}.3}) \Delta_{I}^2 \size{I} \right\},
    \end{align*}
    where $E_1=\{I: \size{I} \ge \delta_{\mathsf{m}}, \Delta_{I} = 0\}$,
    $
E_2^-=\{I = (a, b]: \size{I} \ge \delta_{\mathsf{m}}; I \cap \truecps = \{\tau_k^\ast\}, \min(\tau_k^\ast - a, b - \tau_k^\ast) \le 2 \widetilde{C} \Delta_k^{-2} s \log(p \vee n)\}
$,
$
E_2^+ =\{I \in E_2^{-}: I \cap \truecps = \{\tau_k^\ast\}, \size{I} \ge (\Delta_{k}^2 \vee 1) C_{\mathsf{snr}} s \log(p \vee n)\}
$,
    and $E_3=\{I:  \size{I} \ge \delta_{\mathsf{m}}, \Delta_{I}^2 \size{I} \ge \widetilde{C} s \log(p \vee n)\}$.
    Here $C_{\ref*{lem:loc_err_g}.1}$, $C_{\ref*{lem:loc_err_g}.2}$ and $C_{\ref*{lem:loc_err_g}.3}$ are positive constants. In addition, the constants $C_{\gamma}$ and $\widetilde{C}$ only depends on $C_{\mathsf{snr}}$, $C_{\mathsf{m}}$, $C_{\ref*{lem:loc_err_g}.1}$, $C_{\ref*{lem:loc_err_g}.2}$, and $C_{\ref*{lem:loc_err_g}.3}$.
\end{lemma}

Lemma \ref{lem:loc_err_g} presents a deterministic result. The probabilistic conditions come into play when certifying that the event $\mathbb{G}$ holds with high probability for both the standard OP implementation with $\Lcal_I=\Lcal(I;\hat{\btheta}_I)$ and the accelerated Reliever version with $\Lcal_I=\Lcal(I;\hat{\btheta}_{R_I})$. This lemma offers new insights into the necessary conditions for the evaluated losses using a general model-fitting procedure $\hat{\btheta}(I)$ along the OP grid-search path to produce consistent and nearly rate-optimal changepoint estimators, which may be of independent interest. Theorem \ref{thm:localization_error} further asserts that this event $\mathbb{G}$ occurs with high probability under additional Conditions \ref{cond:parameter}--\ref{cond:dist}.

\begin{theorem}\label{thm:localization_error}
    Suppose that Conditions \ref{cond:change}--\ref{cond:dist} hold. Let $C_{\lambda}$ and $C_{\gamma}$ be positive constants, and $0 < C_{\mathsf{m}} < C_{\mathsf{snr}}$ be sufficiently large constants.
    The solution $(\hat{\tau}_1,\ldots,\hat{\tau}_{\hat{K}})$ of either Problem (\ref{object_cp_full_linear}) or Problem (\ref{object_cp_part_linear}) with $\delta_{\mathsf{m}} = C_{\mathsf{m}} s \log(p \vee n)$, $\lambda_I = C_{\lambda} C_x \sigma_x D_I \{\size{I} \log(p \vee n)\}^{\frac{1}{2}}$, and $\gamma = C_{\gamma} s \log(p \vee n)$, satisfies that
    \begin{align*}
    \Pbb\Bigset{
        \widehat{K} = K^\ast \mbox{ and } \max_{1 \le k \le K^\ast} \min_{1 \le j \le \widehat{K}} \frac{1}{2} \Delta_{k}^2 \abs{\tau_k^\ast - \hat{\tau}_j} \le \widetilde{C} s \log(p \vee n)
    } \ge 1 - (p \vee n)^{-c},
    \end{align*}
    where $D_I = (C_x^2\Delta_I^2 + C_{\epsilon}^2)^{\frac{1}{2}}$. The constants $C_{\gamma}$, $C_{\lambda}$, $\widetilde{C}$ and $c$ are independent of $(n, p, s, K^\ast)$.
    Moreover, under the same probability event, there exists a constant $C > 0$ such that for all $1 \le k \le K^\ast + 1$,
    \begin{equation*}
        \norm{\hat{\btheta}_{(\hat{\tau}_{k-1}, \hat{\tau}_k]} - \btheta_k^\ast}_2 \le C \biggl\{\frac{s \log(p \vee n)}{\tau_k^\ast - \tau_{k-1}^\ast}\biggr\}^{\frac{1}{2}}.
    \end{equation*}
\end{theorem}

Theorem \ref{thm:localization_error} demonstrates that, under mild conditions and with appropriately chosen tuning parameters $\gamma$ and $\lambda_I$, both the original and Reliever-enhanced implementations of OP consistently estimate the number of changepoints and achieve a state-of-the-art localization rate $n^{-1}\abs{\tau_k^\ast - \hat{\tau}_k} \le C\Delta_{k}^{-2} n^{-1} s \log(p \vee n) $ with high probability. This rate exhibits the phenomenon of \textit{superconsistency} for changepoint estimation in high-dimensional linear regression with multiple changepoints, extending a well-known result for single changepoint scenarios \citep{MR3453652}. Importantly, our theoretical analysis supports scenarios where $K^*$, the number of changepoints, varies with $n$ and may potentially diverge. When $K^*=O(1)$, our findings are consistent with those reported in \cite{rinaldo_localizing_2021} and \cite{xu2022change}, which use OP-type algorithms. \citet{wang_statistically_2021} allows for $K^*$ to diverge and derives this rate using a WBS-type algorithm. Additionally, it is noteworthy that the tuning parameter $\lambda_I$, which serves as the regularization factor for the lasso model within each interval $I$, not only scales with $|I|^{\frac{1}{2}}$ but is also modulated by the change magnitude $\Delta_I^2$. In fact, determining the rate of $\lambda_I$ involves examining the uniform bound of a sequence of mean-zero (sub-)gradients, where the variance is, however, influenced by $\Delta_I^2$. Previous works, such as those by \cite{wang_statistically_2021} and \cite{xu2022change}, typically assume $\sup_{I}\Delta_I^2=O(1)$, which simplifies the dependency of $\lambda_I$ to $|I|^{\frac{1}{2}}$ alone. Theorem \ref{thm:localization_error} underscores the nuanced, change-adaptive nature of the regularization parameter $\lambda_I$. While the comprehensive exploration of this parameter's dynamics is outside the scope of our current study, it marks a promising avenue for future research and merits further investigation.

\begin{remark}[Tradeoff between computational time and estimation accuracy]
\label{rmk:tradeoff}
At first glance, Reliever might appear to provide an advantage without a corresponding cost, as the localization rate initially appears to be unaffected by the coverage ratio $r$.
However, a deeper analysis of the underlying proofs reveals that $r$ subtly influences the constant $\tilde{C}$ in the localization rate, particularly since $r$ is held constant.
More precisely, the magnitude of $\tilde{C}$ is dependent on several constants including $C_{\mathsf{snr}}$, $C_{\mathsf{m}}$, $C_{\ref*{lem:loc_err_g}.1}$, $C_{\ref*{lem:loc_err_g}.2}$, and $C_{\ref*{lem:loc_err_g}.3}$, as detailed in Lemma \ref{lem:loc_err_g}.
By setting $C_{\mathsf{snr}}$ and $C_{\mathsf{m}}$ to sufficiently large values, we determine that $\widetilde{C} = 2 (1 - C_{\ref*{lem:loc_err_g}.3})^{-1} (3 C_{\ref*{lem:loc_err_g}.1} + 10 C_{\ref*{lem:loc_err_g}.2})$.
From the proof, it becomes evident that $C_{\ref*{lem:loc_err_g}.j} \propto r^{-2}$ for $j=1,2,3$.
Therefore, as $r$ decreases, the constants $C_{\ref*{lem:loc_err_g}.j}$ increase, which in turn elevates $\tilde{C}$ and leads to deteriorated localization rates for smaller values of $r$.
Similar relation can also be observed in Proposition~\ref{thm:mest} for single changepoint detection.
The observation illustrates a pivotal \textit{tradeoff}: lower values of $r$ enhance computational speed at the expense of localization precision.
As $r$ approaches 1, the distinction between Reliever and the original grid-search algorithm diminishes, indicating minimal computational gains in exchange for optimal localization accuracy.
It is also worth emphasizing that by choosing a fixed $0 < r < 1$, Reliever can always reduce the number of model fits to $O(n)$ as Proposition~\ref{prop:relief_interval} shows.
For detailed derivations and specific values of $C_{\ref*{lem:loc_err_g}.j}$, $j=1,2,3$, please refer to Corollary~\ref{cor:mix_err} in Section~\ref{subsec:cert_reliever}.
\end{remark}

\subsubsection{Temporal Dependence: Extending the Localization Theory}\label{rmk:temporal_dependence}

A closer inspection of the proof of Theorem~\ref{thm:localization_error} (independent case) reveals that independence is invoked only inside a Bernstein-type tail bound that establishes oracle inequalities (see Section~\ref{sec:proof_main_thm}).
Replacing this bound with a version suited to dependent data therefore suffices to extend the theory.

\citet{xu2022change} establish such a Bernstein bound for functionally dependent sequences.
Incorporating their bound yields Reliever's nearly rate-optimal localization guarantee for temporally dependent data. The main ingredients are summarized below; detailed proofs are deferred to Section~\ref{sec:proof_under_dependence}.

\begin{definition}[Functional dependent sequence \citep{wu2005nonlinear,xu2022change}]
For each $t \in \Zbb$, let
\begin{equation*}
    \xbf_t = \gbf_t(\Fcal_t^X),
\end{equation*}
where $\Fcal_t^X = \set{X_s}_{s \le t}$ is generated from i.i.d. elements $\set{X_s}_{s \in \Zbb}$, and $\gbf_t:\Fcal_t^X\to\Rbb^p$ is measurable. The sequence $\set{\xbf_t}_{t \in \Zbb}$ is then called functionally dependent, with dependence functions $\{\gbf_{t}\}_{t \in \Zbb}$ and generating elements $\set{X_t}_{t \in \Zbb}$. Let $\Fcal_{t, s}^X$ be the same as $\Fcal_t^X$ except that $X_s$ is replaced by an independent copy $\widetilde{X}_s$. Define, for $q\ge1$, the functional dependence measure and its cumulative version as
\begin{equation*}
    \delta_{s, q}^{\xbf} = \sup_{\vbf \in \Scal^{p-1},\, t \in \Zbb} [\Ebb \abs{\vbf^\top (\xbf_{t} - \xbf_{t - s})}^q]^{\frac{1}{q}} \mbox{~and~} \Delta_{m, q}^{\xbf} = \sum_{s = m}^{\infty} \delta_{s, q}^{\xbf},\, m \in \Zbb,
\end{equation*}
respectively.
\end{definition}

Within the functional-dependence framework, we impose the following conditions.

\begin{condition}[Change signals]\label{cond:change_temporal}
    There exists a sufficiently large constant $C_{\mathsf{snr}}>0$ such that for $k = 1,\dots,K^\ast+1$, $\tau_k^\ast - \tau_{k-1}^\ast \ge C_{\mathsf{snr}} s\log(p \vee n) [\Delta_{k-1}^{-2} + \Delta_k^{-2} + \{s\log(p \vee n)\}^{2/\zeta - 2}]$ where $\zeta \in (0, 1)$. We additionally assume $K = O(1)$ and $\sup_{1 \le k \le K^\ast} \Delta_k \le C_{\Delta}$ for some universal constant $C_{\Delta} > 0$.
\end{condition}
\begin{condition}[Regression coefficients]\label{cond:parameter_temporal}
$\size{\Scal_k} \le s < p$, where $\Scal_k = \set{1\le j\le p: \theta_{k,j}^\ast \neq 0}$.
\end{condition}
\begin{condition}[Covariates and noises]\label{cond:dist_temporal}
    Let $\zeta_1 > 0$ and $\zeta_2 \in (0, 2]$ be two constants such that $(\zeta_1^{-1} + \zeta_2^{-1})^{-1} = \zeta \in (0, 1)$.
    (a) \textbf{Covariates.} The sequence $\set{\xbf_i}_{i=1}^n$ is a consecutive subsequence of an infinite functionally dependent sequence $\set{\xbf_t}_{t \in \Zbb}\subset \Rbb^p$ with a time-invariant dependence function $\gbf^{\xbf}$ ($\gbf_t = \gbf^{\xbf}$ for all $t \in \Zbb$). Moreover, $\sup_{m \ge 0} \exp(c m^{\zeta_1}) \Delta_{m,4}^{\xbf} \le D_{x}$, for some constant $D_x > 0$.
    Assume each $\xbf_i$ has mean zero and covariance $\Sigma$, satisfying $0<\kmin\le \sigma_x^2<\infty$, where $\kmin=\lambda_{\min}(\Sigma)$ and $\sigma_x^2=\lambda_{\max}(\Sigma)$, respectively. Furthermore, $\norm{\Sigma^{-\frac{1}{2}}\xbf_i}_{\Psi_{\zeta_2}} \le C_x$ for some constant $C_x>0$.
    (b) \textbf{Noises}. The sequence $\set{\epsilon_i}_{i=1}^n$ is a consecutive subsequence of an infinite functionally dependent sequence $\set{\epsilon_t}_{t \in \Zbb} \subset \Rbb$ with a time-invariant dependence function $g^{\epsilon}$. Moreover, $\sup_{m \ge 0} \exp(c m^{\zeta_1}) \Delta_{m,4}^{\epsilon} \le D_{\epsilon}$, for some constant $D_\epsilon > 0$.
    Assume $\set{\epsilon_i}_{i=1}^n$ are independent of $\set{\xbf_i}_{i=1}^n$, and each $\epsilon_i$ has mean zero and variance $\sigma_{\epsilon}^2$. Furthermore, $\norm{\epsilon_i}_{\Psi_{\zeta_2}} \le C_{\epsilon}$.
\end{condition}

Conditions \ref{cond:change_temporal}--\ref{cond:dist_temporal} mirror their counterparts in the temporally independent setting (Conditions \ref{cond:change}--\ref{cond:dist}) in temporal independence scenarios; we additionally assume $K = O(1)$ and $\sup_{1 \le k \le K^\ast} \Delta_k \le C_{\Delta}$, as in \cite{xu2022change}.

\begin{corollary}\label{thm:localization_error_temporal}
    Suppose that Conditions \ref{cond:change_temporal}--\ref{cond:dist_temporal} hold. Let $C_{\lambda}$ and $C_{\gamma}$ be some positive constants, and $0 < C_{\mathsf{m}} < C_{\mathsf{snr}}$ be sufficiently large constants.
    The solution $(\hat{\tau}_1,\ldots,\hat{\tau}_{\hat{K}})$ of either Problem (\ref{object_cp_full_linear}) or Problem (\ref{object_cp_part_linear}) with $\delta_{\mathsf{m}} = C_{\mathsf{m}} \{s \log(p \vee n)\}^{2/\zeta - 1}$, $\lambda_I = C_{\lambda} C_x \sigma_x \{\size{I} \log(p \vee n)\}^{\frac{1}{2}}$, and $\gamma = C_{\gamma} s \log(p \vee n)$, satisfies that
    \begin{align*}
    \Pbb\Bigset{
        \widehat{K} = K^\ast \mbox{ and } \max_{1 \le k \le K^\ast} \min_{1 \le j \le \widehat{K}} \Delta_{k}^2 \abs{\tau_k^\ast - \hat{\tau}_j} \le \widetilde{C} s \log(p \vee n)
    } \ge 1 - (p \vee n)^{-c}.
    \end{align*}
    The constants $C_{\gamma}$, $C_{\lambda}$, $\widetilde{C}$ and $c$ are independent of $(n, p, s, K^\ast)$.
    Moreover, under the same probability event, there exists a constant $C > 0$ such that for all $1 \le k \le K^\ast + 1$,
    \begin{equation*}
        \norm{\hat{\btheta}_{(\hat{\tau}_{k-1}, \hat{\tau}_k]} - \btheta_k^\ast}_2 \le C \biggl\{\frac{s \log(p \vee n)}{\tau_k^\ast - \tau_{k-1}^\ast}\biggr\}^{\frac{1}{2}}.
    \end{equation*}
\end{corollary}

\section{Numerical Studies}\label{sec:sim}

To evaluate the effectiveness of the Reliever approach compared to the original implementation of various grid-search algorithms, we explore two scenarios: high-dimensional linear changepoint models (cf. Example \ref{ex:lin_reg}) and nonparametric changepoint models (cf. Example \ref{ex:nonpara}). The grid-search algorithms assessed include SN, WBS (with $M=100$ random intervals), and SeedBS (using a decay parameter $a = 2^{-{1}/{2}}$, as recommended by \citet{kovacs_seeded_2022}), implemented with a known number of changepoints for a fair comparison. For nonparametric models, we also consider OP and PELT, which do not presuppose the number of changepoints. The accuracy of changepoint estimation is quantified using the Hausdorff distance $\max\{\text{OE},\text{UE}\}$, where $\text{OE} = \max_{1 \le j \le \widehat{K}} \min_{1 \le k \le K^\ast} |\tau_k^\ast - \hat{\tau}_j|$ is the over-segmentation error and $\text{UE} = \max_{1 \le k \le K^\ast} \min_{1 \le j \le \widehat{K}} |\tau_k^\ast - \hat{\tau}_j|$ is the under-segmentation error. The following results are based on $500$ replications.

\subsection{High-Dimensional Linear Models}\label{subsec:hdlinear}

In this scenario, we examine changepoint detection in high-dimensional linear models as outlined in Example \ref{ex:lin_reg}, with $n\in\{300,600,900,1200\}$ and $p = 100$. The covariates $\set{\xbf_i}$ are i.i.d. from the standard multivariate Gaussian distribution, and the noises $\set{\epsilon_i}$ are i.i.d. from the standard Gaussian distribution $\mathcal{N}(0,1)$. We introduce three changepoints at $\set{\tau_k^\ast}_{k=1}^3 = \set{\lfloor 0.22 n \rfloor, \lfloor 0.55 n \rfloor, \lfloor 0.77 n \rfloor}$. The regression coefficients $\set{\btheta_k^\ast}$ are generated such that $\theta_{k,j}=0$ for $j=3,\ldots,p$, and $\theta_{k,1}$ and $\theta_{k,2}$ are uniformly sampled, satisfying the SNRs $\norm{\btheta_1}_2^2 / \Var(\epsilon_1) = 2$ and $\norm{\btheta_{k} - \btheta_{k-1}}_2^2 / \Var(\epsilon_1) = 2^{-1}$ for $k = 2, 3, 4$. Here $\theta_{k,j}$ denotes the $j$th element of $\btheta_{k}$. We set $\delta_{\mathsf{m}} = 20$ for numerical stability and employ lasso for parameter estimation using the \textsf{glmnet} package \citep{friedman_regularization_2010} in \textsf{R}. We apply three grid-search algorithms, SN, WBS, and SeedBS, assuming a known number of changepoints $\hat{K}=K^\ast=3$. For each algorithm, we scale the regularization parameter $\lambda_I = \lambda \size{I}^{\frac{1}{2}}$, with $\lambda$ ranging from a set of $30$ values. The changepoint detection error for each algorithm is reported as the minimal Hausdorff distance achieved across all values of $\lambda$.

Figures \ref{fig:lasso_err}--\ref{fig:lasso_time} display the changepoint detection error and the computational time for each grid-search algorithm across different values of the coverage ratio $r$.
The value $r=1$ corresponds to the original implementation. The results indicate that as $r$ approaches $1$, the performance of Reliever approaches that of the original implementation. For $r = 0.9$, Reliever delivers comparable results to the original implementations but with substantial reductions in computational time. Even when $r = 0.6$, the performance is still acceptable, considering the negligible running time.

\begin{figure}[htb]
    \centering
    \includegraphics[width=1\linewidth]{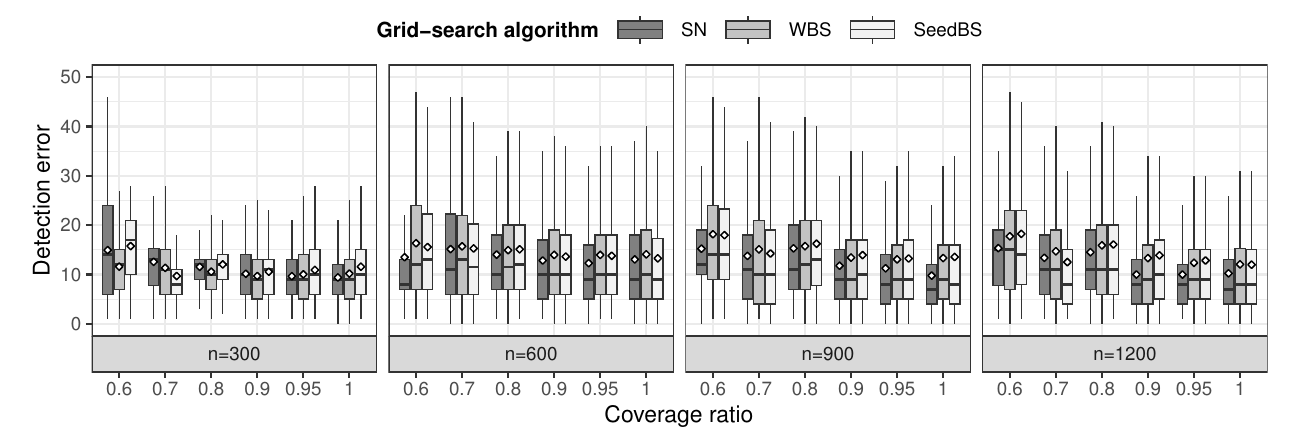}
\caption{Changepoint detection error for various grid-search algorithms across varying values of the coverage ratio, under the high-dimensional linear model.}
\label{fig:lasso_err}
\end{figure}

\begin{figure}[htb]
    \centering
    \includegraphics[width=1\linewidth]{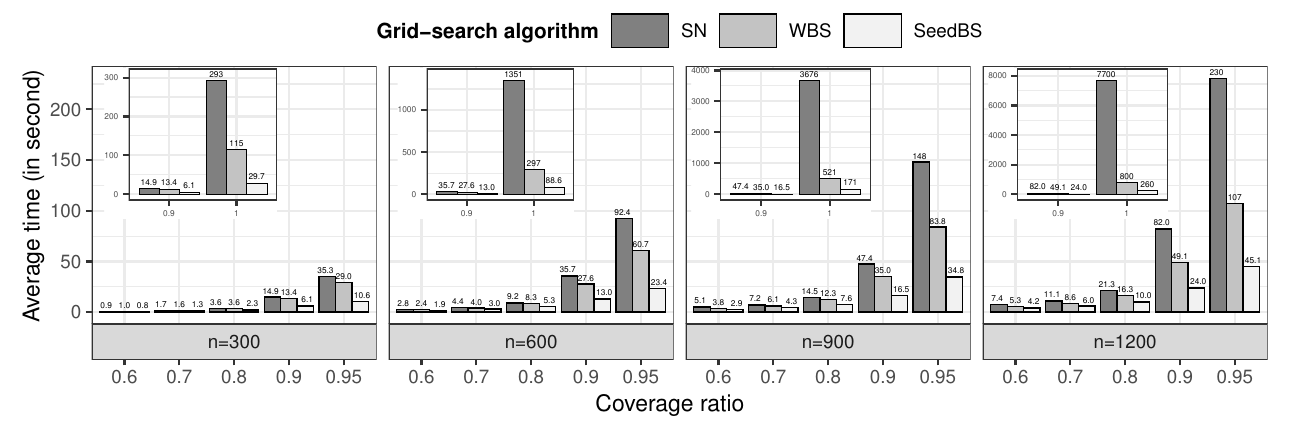}
\caption{Computational time for various grid-search algorithms across varying values of the coverage ratio, under the high-dimensional linear model.}
\label{fig:lasso_time}
\end{figure}

\subsection{Univariate Nonparametric Models}\label{subsec:nmcd}

In the second scenario, we explore changepoint detection for univariate nonparametric distributions as described in Example \ref{ex:nonpara}. We employ the same three-changepoint structure used in the first scenario. The data within the four segments are generated from different distributions, that is, $\mathcal{N}(0,1)$, $\chi_{(3)}^2$ (standardized to have unit variance), $\chi_{(1)}^2$ (likewise standardized), and $\mathcal{N}(0,1)$, respectively. We implement SN, WBS, and SeedBS to identify changepoints, assessing their effectiveness across varying coverage ratios $r$. Figures \ref{fig:nmcd_err}--\ref{fig:nmcd_time} summarize the changepoint detection error and computational time for each algorithm. The Reliever method demonstrates robust performance, particularly for $r$ values above $0.7$. Notably, SN maintains consistent accuracy and efficiency across a range of $r$ values.

\begin{figure}[htb]
    \centering
    \includegraphics[width=1\linewidth]{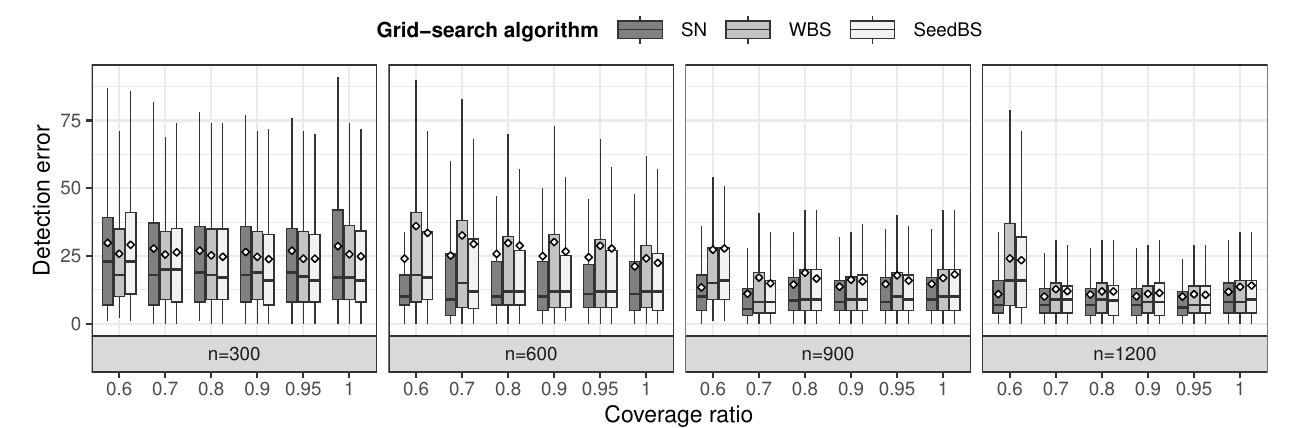}
\caption{Changepoint detection error for various grid-search algorithms across varying values of the coverage ratio, under the nonparametric model.}
\label{fig:nmcd_err}
\end{figure}

\begin{figure}[htb]
    \centering
    \includegraphics[width=1\linewidth]{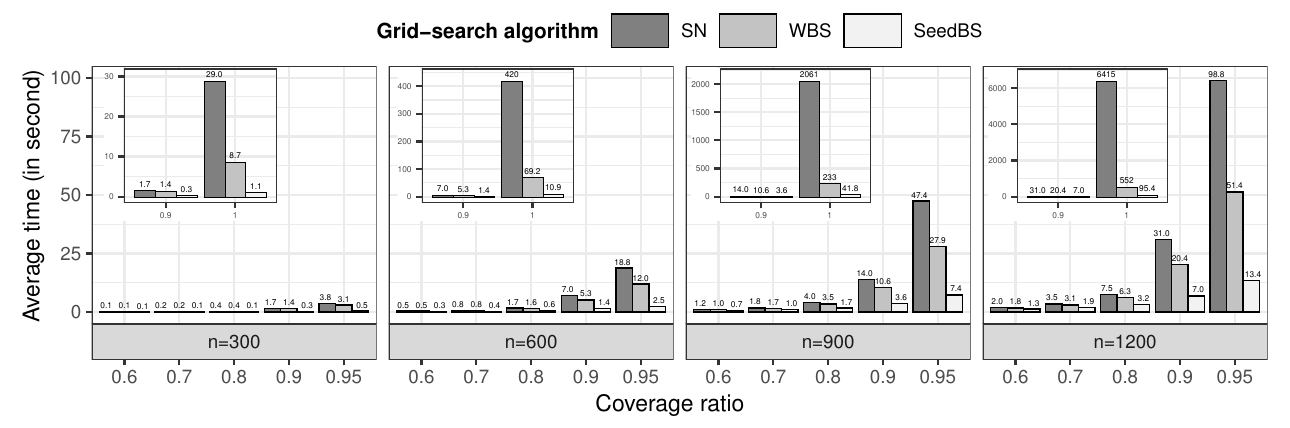}
\caption{Computational time for various grid-search algorithms across varying values of the coverage ratio, under the nonparametric model.}
\label{fig:nmcd_time}
\end{figure}

\subsection{Nonparametric Changepoint Detection via Kernel-Density Estimation}\label{subsec:multivariate_nonpara}

This section shows how Reliever integrates naturally with the kernel-density CUSUM framework of \citet{padilla2023change}. Consider observations $\{\zbf_i\}_{i=1}^n \subset \Rbb^p$. For any search interval $I=(a,b]\in\I$ and for a kernel $\Kcal(\cdot)$ with bandwidth $h > 0$, the density is estimated by $\hat{f}_{I}(\zbf) = \{\size{I}h^p\}^{-1} \sum_{i \in I} \Kcal\Bigl((\zbf - \zbf_i)/h\Bigr),\ \forall \zbf \in \Rbb^p$. Given a split point $t\in I$, define the pointwise CUSUM contrast $\tilde{f}_{I, t}(\zbf) = \sqrt{(b - t)(t - a)/(b - a)} \{\hat{f}_{(a, t]}(\zbf) - \hat{f}_{(t, b]}(\zbf)\}$ and the overall CUSUM contrast $\norm{\tilde{f}_{I, t}}_{n}^2 = {n}^{-1} \sum_{j=1}^n \tilde{f}_{I, t}^2(\zbf_j)$. The changepoint within $I$ is then estimated as $\argmax_{a+\delta_{\mathrm{m}}<t\le b - \delta_{\mathrm{m}}} \norm{\tilde{f}_{I, t}}_{n}^2$.

\cite{padilla2023change} adopt the SeedBS algorithm; to embed their method in our Reliever framework we translate the contrast-based CUSUM into a loss-based evaluation. Define the loss $\Lcal(I; \hat{f}_I) = \sum_{j=1}^n \sum_{i \in I} \Bigl\{{h^p}^{-1} K\Bigl((\zbf_j - \zbf_i)/{h}\Bigr) - \hat{f}_I(\zbf_j)\Bigr\}^2$. Maximizing the overall contrast is equivalent to $\argmin_{a+\delta_{\mathrm{m}}<t\le b - \delta_{\mathrm{m}}} \Lcal((a,t]; \hat{f}_{(a,t]}) + \Lcal((t,b]; \hat{f}_{(t,b]})$.

In the original SeedBS algorithm, $\hat{f}_{I}$ and $\Lcal(I; \hat{f}_{I})$ are computed for every search interval $I\in\I$ along the search path. With Reliever, we instead fit $\hat{f}_{R}$ on only $O(n)$ deterministic proxy intervals $R \in \R$, reuse these fits, and evaluate $\Lcal(I; \hat{f}_{R})$ for all required $I$.

We replicate the three-changepoint scenario of Sections 4.1--4.2 (with $n = 1200$, $p=5$, $\{\tau_k^\ast\}_{k=1}^3 = \{0.22 n, 0.55 n, 0.77 n\}$, $\delta_{\mathrm{m}} = 2$). Data are generated similar to those in Scenario 3 of \citet{padilla2023change}: two independent sequences $\{\ebf_{z, i}\} \subset \Rbb^p$ and $\{\ebf_{z, i}'\} \subset \Rbb^p$ with i.i.d. entries $\ebf_{z, i,j}\sim\mathrm{Pareto}(3, 1)$ and $\ebf_{z, i,j}'\sim\mathrm{Uniform}(-\sqrt{3}, \sqrt{3})$ drive AR(1) processes $\zbf_i = 0.3 \zbf_{i-1} + \ebf_{z, i}$ and $\zbf_i' = 0.3 \zbf_{i-1}' + \ebf_{z, i}'$ (starting with $\zbf_{0} = \zbf_{0}' = \boldsymbol{0}$). Then for $i \in (\tau_1^\ast, \tau_2^\ast] \cup (\tau_3^\ast, n]$, we reset $\zbf_i$ by $\zbf_i = \sqrt{0.8} \zbf_i - \sqrt{0.2} \zbf_i'$. We use a standard radial basis function (RBF) kernel $\Kcal$ with bandwidth $h = 1$, and set Reliever's coverage ratio to $r = 0.9$.

Table \ref{tab:multi_nonpara_cp} reports changepoint detection error and density‐fit time. Although one can update $\hat{f}_{I}$ and $\L(I; \hat{f}_I)$ in $O(p)$ per step---for instance, $\hat{f}_{(a, b+1]}(\zbf) = \frac{1}{b - a + 1}\{(b-a)\hat{f}_{(a, b]}(\zbf) + h^{-p} K((\zbf - \zbf_{b+1})/ h) \}$, Reliever still cuts total fitting time across WBS, SeedBS, and SN, with almost no change in detection error.

\begin{table}[htb]
\setlength\tabcolsep{0pt}
\centering
\begin{threeparttable}
\caption{\small Average detection error and density-fit time (centiseconds) for multivariate nonparametric changepoint detection via kernel density estimation, comparing original and Reliever-enabled versions (coverage ratio \(r=0.9\)) of WBS, SeedBS, and SN.}
\label{tab:multi_nonpara_cp}
\tabcolsep=0.44em
\begin{tabular*}{.97\linewidth}{c@{\extracolsep{\fill}}*{6}{r}}
\toprule
Grid search & \multicolumn{2}{c}{WBS} & \multicolumn{2}{c}{SeedBS} & \multicolumn{2}{c}{SN} \\
\cline{1-1} \cline{2-3}  \cline{4-5} \cline{6-7}
Model fitting & Original & Reliever & Original & Reliever & Original & Reliever \\
\midrule
Error & 38.9 & 39.6 & 41.7 & 41.6 & 34.8 & 34.8 \\
Time & 91.9 & 54.9 & 19.2 & 14.9 & 447.1 & 147.3 \\
\bottomrule
\end{tabular*}
\end{threeparttable}
\end{table}

\subsection{Integration of Reliever with OP and PELT}

To investigate how Reliever integrates with OP and PELT, without presupposing the number of changepoints, we revisit the nonparametric changepoint model discussed in Section \ref{subsec:nmcd}. This model accommodates the applicability of PELT, proposed by \cite{haynes_computationally_2017}. Additionally, we examine the data-generating process of Model 1 from \cite{MR3210993} and \cite{haynes_computationally_2017}, with $K^\ast = 11$ changepoints and Student-$t(3)$-distributed noises. Specifically,
\[
    z_i = \sum_{k = 1}^{K^\ast} h_k \id_{\set{i > \tau_k^\ast}} + \sigma \epsilon_i,\, i=1,\ldots,n,
\]
where $\set{\tau_k^\ast}/n = \set{0.1, 0.13, 0.15, 0.23, 0.25, 0.40, 0.44, 0.65, 0.76, 0.78, 0.81}$, $\set{h_k} = \{2.01,\\ -2.51, 1.51, -2.01, 2.51, -2.11, 1.05, 2.16, -1.56, 2.56, -2.11\}$, $\set{\epsilon_i} \overset{i.i.d}{\sim} t(3)$, and $\sigma = 0.5$. This configuration is designated as Model~(B), in contrast to the three-changepoint setting, which is referred to as Model~(A). Table~\ref{tab:nmcd_op_pelt} summarises the results for $n = 1000$. The findings highlight that while PELT significantly reduces computational time compared to OP, Reliever can further decrease this burden without compromising the detection accuracy, remaining nearly identical to those obtained via the original OP algorithm.

\begin{table}[htb]
\setlength\tabcolsep{0pt}
\centering
\begin{threeparttable}
\caption{\small Average absolute errors of changepoint number estimates, detection error, and computational time of OP and PELT, with and without Reliever, under the univariate nonparametric model with $n=1000$ and $K^\ast = 11$.}
\label{tab:nmcd_op_pelt}
    \tabcolsep=0.45em
    \begin{tabular*}{.97\linewidth}{@{\extracolsep{\fill}}cc*{8}{r}}
        \toprule
        Model & Coverage ratio & \multicolumn{2}{c}{$|\widehat{K} - K^\ast|$} & \multicolumn{2}{c}{OE} & \multicolumn{2}{c}{UE} & \multicolumn{2}{c}{Time (Second)} \\
        \midrule
         & & OP & PELT & OP & PELT & OP & PELT & OP & PELT \\
        \multirow{6}{*}{(A)}
        & (Original) 1.0 & 0.13 & 0.13 & 14.30 & 14.30 & 48.18 & 48.18 & 2804.17 & 1098.88 \\
        &0.9 & 0.13 & 0.13 & 14.01 & 14.01 & 49.16 & 49.16 & 22.65 & 19.22   \\
        &0.8 & 0.14 & 0.14 & 14.05 & 14.10 & 49.05 & 49.11 & 7.71  & 6.60   \\
        &0.7 & 0.15 & 0.15 & 13.95 & 14.12 & 54.36 & 55.83 & 3.56  & 3.01   \\
        &0.6 & 0.17 & 0.17 & 14.87 & 15.12 & 57.39 & 58.21 & 2.07  & 1.72   \\
        &0.5 & 0.21 & 0.21 & 18.67 & 19.72 & 69.16 & 72.34 & 1.34  & 1.07   \\
        \midrule
        \multirow{6}{*}{(B)}
        & (Original) 1.0 & 0.01 & 0.01 & 2.32 & 2.32 & 2.54 & 2.54 & 3010.80 & 65.85 \\
        & 0.9 & 0.01 & 0.01 & 2.29 & 2.29 & 2.51 & 2.51 & 24.26   & 8.56  \\
        & 0.8 & 0.01 & 0.01 & 2.25 & 2.25 & 2.47 & 2.47 & 8.22    & 3.51  \\
        & 0.7 & 0.01 & 0.01 & 2.37 & 2.37 & 2.41 & 2.41 & 3.80    & 1.80  \\
        & 0.6 & 0.00 & 0.00 & 2.18 & 2.18 & 2.18 & 2.18 & 2.19    & 1.13  \\
        & 0.5 & 0.01 & 0.01 & 2.45 & 2.45 & 2.56 & 2.56 & 1.41    & 0.75  \\
        \bottomrule
    \end{tabular*}
    \vspace{-1em}
\end{threeparttable}
\end{table}

\subsection{Integration of Reliever with DCDP}\label{subsec:dcdp}

\cite{li2023divide} propose the two-stage Divide and Conquer Dynamic Programming (DCDP) algorithm. Stage I runs dynamic programming (DP, that is, SN) on a coarse grid of candidate points, while Stage II locally refines each preliminary changepoint to improve accuracy.
This design reduces computation by largely eliminating the grid-search space. Because DCDP and Reliever address different bottlenecks, the two can be combined. We therefore evaluate (i) DP (Original versus Reliever); (ii) DCDP Stage I (Original versus Reliever); and (iii) DCDP Stages I--II (Original versus Reliever). Here, ``Original'' corresponds to Reliever with \(r=1\) (that is, full model fits). We set $r=0.9$ for Reliever and use a coarse-grid step of $20$ for DCDP. Results for the high-dimensional linear setting of Section~\ref{subsec:hdlinear} ($n = 1200$) are summarized in Table~\ref{tab:DCDP}.
Across all three schemes, Reliever cuts fitting time while keeping error roughly at the same level---so the combined strategy ``Reliever + DCDP'' is promising for very large-scale problems.
\begin{table}[htb]
\setlength\tabcolsep{0pt}
\centering
\begin{threeparttable}
\caption{\small Detection error and runtime for DP and DCDP with/without Reliever in a high-dimensional linear model (coarse step $=20$; coverage ratio $r=0.9$).}
\label{tab:DCDP}
\tabcolsep=0.44em
\begin{tabular*}{.97\linewidth}{c@{\extracolsep{\fill}}*{6}{r}}
\toprule
Grid search & \multicolumn{2}{c}{DP} & \multicolumn{2}{c}{DCDP I} & \multicolumn{2}{c}{DCDP I--II} \\
\cline{1-1} \cline{2-3}  \cline{4-5} \cline{6-7}
Model fitting & Original & Reliever & Original & Reliever & Original & Reliever \\
\midrule
Error & 12.6 & 13.4 & 13.9 & 14.3 & 13.6 & 13.9 \\
Time & 7,700.0 & 82.0 & 32.9 & 10.7 & 33.4 & 11.2 \\
\bottomrule
\end{tabular*}
\end{threeparttable}
\end{table}

\subsection{Comparison with Two-Step Methods}\label{sec:comp_ts}

We present a comparative analysis between the Reliever method and the two-step approach proposed by \cite{kaul_efficient_2019-1}. The two-step method is specifically designed to detect a single changepoint in a high-dimensional linear model. It involves an initial guess of the changepoint, which divides the data into two intervals. Proxy models are then fitted within these intervals. Consequently, both methods expedite the process of change detection by reducing extensive model fits. For mitigating the uncertainty in the initialization, multiple guesses are considered, and a changepoint estimator that minimizes the total loss on both segments is reported. In our study, we consider the high-dimensional linear model discussed in Section 5.1 of \citet{kaul_efficient_2019-1}, with $n = 1200$ and $\tau^\ast = 120$. We consider multiple initial guesses, specifically ${0.25n, 0.5n, 0.75n}$. The results presented in Table \ref{tab:scp_ts_vs_rf} indicate that although the two-step method may offer faster computation due to fewer model fits, it also exhibits larger changepoint detection error. This can be attributed to its performance being heavily reliant on the accuracy of the initial changepoint estimate (or the quality of the corresponding intervals). In contrast, the Reliever method demonstrates stability across a range of choices for the parameter $r$, varying from $0.9$ to $0.3$.

\begin{table}[htb]
        \setlength\tabcolsep{1em}
        \centering
        \begin{threeparttable}
        \caption{\small Comparison of average changepoint detection error and computational time (\textit{in centiseconds}) between the Reliever method and the two-step method under the high-dimensional linear model with single changepoint, and $(n,\delta_{\mathsf{m}})=(1200,30)$. The numbers in parentheses represent the corresponding standard errors.}
        \label{tab:scp_ts_vs_rf}
        \begin{tabular*}{.97\linewidth}{c@{\extracolsep{\fill}}*{5}{r}}
        \toprule
        & Two-step & $r=0.9\,$ & $r=0.7\,$ & $r=0.5\,$ & $r=0.3\,$ \\
        \midrule
        Error & $18.6(3.2)$ & $9.2(1.0)$ & $9.4(1.1)$ & $7.6(0.8)$ & $8.7(1.4)$  \\
        Time ($10$ms) & $60.7(0.6)$ & $480.5(1.1)$ & $141.0(0.4)$ & $89.0(0.3)$ & $64.1(0.3)$ \\
        \bottomrule
        \end{tabular*}
        \vspace{-1em}
        \end{threeparttable}
\end{table}

Though without theoretical guarantees, the two-step method can be extended for multiple changepoint detection by incorporating BS along with the multiple guess scheme, as suggested by \citet{londschien_random_2022}. This extension can also be applied to WBS and SeedBS in a similar manner. In our study, we examine the examples presented in Sections \ref{subsec:hdlinear} and \ref{subsec:nmcd} with $n=1200$. Multiple initial guesses are selected as $m$-equally spaced quantiles within a search interval, following the recommendation by \citet{londschien_random_2022}. The results depicted in Table \ref{tab:mcp_ts_vs_rf} reveal that the two-step approach is less efficient for multiple changepoint detection, and increasing the number of multiple initial guesses can even have a detrimental impact on its performance. In contrast, the Reliever method (with $r=0.9$) exhibits performances that are almost comparable to the original implementation.

We also report the corresponding average computational time in Table~\ref{tab:mcp_ts_vs_rf_time}. It is noteworthy to emphasize that, for multiple changepoint detection tasks, even with $r=0.9$, the Reliever method shows comparable computational time to the two-step approach. When $r=0.8$, the Reliever method becomes more efficient. The result is slightly different from the single changepoint case. The reason is that for multiple changepoint detection algorithms like WBS and SeedBS, the two-step method should repeatedly fit the models for every wild/seeded interval. In contrast, the relief models are shared with the global system.

Furthermore, we posit that Reliever serves as a complementary tool rather than a rival to the two-step approach in the domain of multiple changepoint detection. The Reliever method can be combined with the two-step method for multiple changepoint detection tasks. This integration is beneficial because the two-step method still involves a significant number of model fits within the seeded/wild intervals. The Reliever method can further enhance computational efficiency by reducing the time required for these model fits.

\begin{table}[htb]
        \setlength\tabcolsep{1em}
        \centering
        \begin{threeparttable}
        \caption{Comparison of average changepoint detection error between the Reliever method and the two-step method under the multiple changepoint setting in Section~\ref{subsec:hdlinear} and  Section~\ref{subsec:nmcd}. The numbers in parentheses represent the corresponding standard errors}
        \label{tab:mcp_ts_vs_rf}
        \tabcolsep=0.44em
        \begin{tabular*}{.97\linewidth}{@{\extracolsep{\fill}}cc*{6}{r}}
        \toprule
        Example & Algorithm & $m=1\,$ & $m=3\,$ & $m=5\,$ & $r=0.9\,$ & $r=0.8\,$ & Original\, \\
        \midrule
        \multirow{2}{*}{HD} & WBS & $19.7(0.9)$ & $14.7(0.6)$ & $16.8(0.8)$ & $13.3(0.7)$ & $15.1(0.5)$ & $12.1(0.6)$ \\
        & SeedBS & $21.4(1.0)$ & $17.3(0.8)$ & $17.8(0.8)$ & $13.9(0.7)$ & $15.2(0.5)$ & $12.0(0.6)$ \\
         \multirow{2}{*}{NP} & WBS & $85.5(3.2)$ & $17.4(1.2)$ & $17.5(1.2)$ & $11.1(0.5)$ & $12.3(0.6)$ & $13.6(1.0)$ \\
        & SeedBS & $87.5(3.2)$ & $18.7(1.5)$ & $17.6(1.2)$ & $11.4(0.5)$ & $11.9(0.6)$ & $14.1(1.0)$ \\
        \bottomrule
        \end{tabular*}
        \vspace{-1em}
        \end{threeparttable}
\end{table}

\begin{table}[htb]
        \setlength\tabcolsep{1em}
        \centering
        \begin{threeparttable}
        \caption{Comparison of average computational time (\textit{in seconds}) between the Reliever method and the two-step method under the multiple changepoint setting in Section~\ref{subsec:hdlinear} and  Section~\ref{subsec:nmcd}.}
        \label{tab:mcp_ts_vs_rf_time}
        \tabcolsep=0.44em
        \begin{tabular*}{.97\linewidth}{@{\extracolsep{\fill}}cc*{5}{r}}
        \toprule
        Example & Algorithm & $m=1\,$ & $m=3\,$ & $m=5\,$ & $r=0.9\,$  & $r=0.8\,$ \\
        \midrule
        \multirow{2}{*}{HD} & WBS & $20.8 (0.2)$ & $37.0 (0.6)$ & $52.4 (1.2)$ & $49.1 (0.8)$ & $16.3 (0.3)$ \\
        & SeedBS & $17.2 (0.1)$ & $30.2 (0.3)$ & $42.6 (0.8)$ & $24.0 (0.4)$ & $10.0 (0.1)$ \\
         \multirow{2}{*}{NP} & WBS & $2.6 (0.1)$ & $5.7 (0.2)$ & $8.7 (0.2)$ & $20.4 (0.1)$ & $6.3 (0.1)$ \\
        & SeedBS & $1.2 (0.1)$ & $2.5 (0.1)$ & $3.9 (0.1)$ & $7.0 (0.1)$ & $3.2 (0.1)$ \\
        \bottomrule
        \end{tabular*}
        \vspace{-1em}
        \end{threeparttable}
\end{table}

\section{Concluding Remarks}\label{sec:conclusion}

Searching for multiple changepoints in complex models with large datasets poses significant computational challenges.
Current algorithms involve fitting a sequence of models and evaluating losses within numerous intervals during the search process. Existing approaches, such as PELT, WBS, SeedBS, and optimistic search algorithms, aim to reduce the number of (search) intervals. In this paper, we introduce Reliever which specifically relieves the computational burden by reducing the number of fitted models, as they are the primary contributors to computational costs. Our method associates each search interval with a deterministic (relief) interval from a pre-defined pool, enabling the fitting of models only within (or partially within) these selected intervals. The simplicity of the Reliever approach allows for seamless integration with various grid-search algorithms and accommodates different models, providing tremendous potential for leveraging modern machine learning tools \citep{londschien_random_2022,liu2021score,li2022automatic}.

Reliever incorporates a coverage ratio parameter, which balances computational efficiency and estimation accuracy. For high-dimensional regression models with changepoints, by employing an OP algorithm, we characterize requirements on the search path to ensure consistent and nearly rate-optimal estimators for changepoints; see Lemma \ref{lem:loc_err_g}. Our analysis demonstrates that the Reliever method satisfies these properties for any fixed coverage ratio parameter. Further investigation is warranted to characterize the search path for other algorithms and broader model classes. Additionally, our theoretical analysis highlights the importance of adaptively selecting the nuisance parameter based on the underlying change magnitude. Future research should focus on extending the Reliever to enable data-driven selection of nuisance parameters. While the Reliever focuses on changepoint estimation, it is worth exploring the generalization of these concepts to quantify uncertainty in changepoint detection \citep{frick2014multiscale,chen_data-driven_2021} and perform post-change-estimation inference \citep{jewell_testing_2022}.

\section*{Acknowledgments}

All authors contributed equally to this work and are listed in alphabetical order. We would like to acknowledge the action editor, Ji Zhu, and anonymous referees for their valuable comments and suggestions, which have improved the manuscript greatly. Qian's research was supported by the National Natural Science Foundation of China (No. 12501410) and the China Postdoctoral Science Foundation (No. 2024M761934, GZB20250716). Wang was supported by the National Natural Science Foundation of China (No. 12471255), the Natural Science Foundation of Shanghai (No. 23ZR1419400), and the Fundamental Research Funds for the Central Universities (No. 63253110). Zou was supported by the National Key R\&D Program of China (No. 2022YFA1003703) and the National Natural Science Foundation of China (No. 12231011).

\clearpage
\appendix

\section*{Appendix}

The appendix provides proofs of all theoretical results in this article and offers additional numerical analyses.

\section{Proof of Proposition \ref{thm:mest}}\label{sec:proof_mest}

For a fixed $\balpha$, denote the random vectors $\xbf_{i}$ by
\begin{equation*}
    \xbf_{i} = g\Bigl(\zbf_i, \btheta_{I}^\circ + \frac{\balpha}{\size{I}^{\frac{1}{2}}}\Bigr) - g(\zbf_i, \btheta_{I}^\circ).
\end{equation*}
Denote $v_I = (\log n)^{\frac{1}{2}}$. By (g), uniformly for all $\norm{\balpha}_2 \le M v_I$ (with some constant $M > 0$), $\norm{\xbf_i}_{\Psi_1} \le C_{\ref*{sec:proof_mest}.3} M v_I \size{I}^{-\frac{1}{2}}$. Therefore, by applying an exponential inequality,
\begin{equation*}
    \sup_{\norm{\balpha}_2 \le M v_I} \Pbb\biggl[\Bigabs{ G_{I}(\balpha) - G_{I}(\zero) - \overline{G}_{I}(\balpha)} \ge \frac{C_u C_{\ref*{sec:proof_mest}.3} M}{c_b} \size{I}^{-1} v_I (\log n)^{\frac{1}{2}}\biggr] \le 2 \exp(-C_u \log n).
\end{equation*}
By (f),
\begin{equation*}
    \sup_{\norm{\balpha}_2 \le M v_I} \Bigabs{\frac{\Hbf_{I} \balpha}{\size{I}^{\frac{1}{2}}} - \overline{G}_{I}(\balpha)} \le C_{\ref*{sec:proof_mest}.2} M^2 v_I^2 \size{I}^{-1}.
\end{equation*}
The above two inequalities imply that
\begin{equation*}
    \sup_{\norm{\balpha}_2 \le M v_I} \Pbb\biggl[\Bigabs{ G_{I}(\balpha) - G_{I}(\zero) - \frac{\Hbf_{I} \balpha}{\size{I}^{\frac{1}{2}}}} \ge \frac{C_u C_{\ref*{sec:proof_mest}.3} M}{c_b} \size{I}^{-1} v_I (\log n)^{\frac{1}{2}}\biggr] \le 2 \exp(-C_u \log n).
\end{equation*}
By the chaining technique for convex function, that is, the $\delta$-triangulation argument used in \citet{niemiro_asymptotics_1992},
\begin{equation*}
    \Pbb\biggl[\sup_{\norm{\balpha}_2 \le M v_I} \Bigabs{G_{I}(\balpha) - G_{I}(\zero) - \frac{\Hbf_{I} \balpha}{\size{I}^{\frac{1}{2}}}} \ge C_{\ref*{sec:proof_mest}.6} \size{I}^{-1} v_I (\log n)^{\frac{1}{2}} \biggr] \le 2 \size{I}^{\frac{p}{2}} \exp(- C_u \log n).
\end{equation*}
By the sub-exponential assumption, we choose $M > 0$ such that $\Pbb[\norm{\size{I}^{\frac{1}{2}}\Hbf_{I}^{-1} G_{I}(\zero)}_2 \ge (M - 1) (\log n)^{\frac{1}{2}}] \le 2\exp(-C_u \log n)$. It implies that with high probability, $\size{I}^{\frac{1}{2}}\Hbf_{I}^{-1} G_{I}(\zero)$ is in the ball $\set{\xbf \in \Rbb^p: \norm{\xbf}_2 < (M - 1) (\log n)^{\frac{1}{2}}}$. For all $\ebf \in \Rbb^{p}$ with $\norm{\ebf}_2 = 1$, let $\balpha = -\size{I}^{\frac{1}{2}} \{\Hbf_{I}^{-1} G_{I}(\zero) + (K \log n) \size{I}^{-1} \ebf \}$ with $K = 2 C_{\ref*{sec:proof_mest}.6}/\lambda_{\min}(\Hbf_{I})$. With probability at least $1 - 2 (1 + \size{I}^{p/2}) \exp(-C_u \log n)$,
\begin{align*}
    & \ebf^\top G_{I}\bigl(\size{I}^{\frac{1}{2}}\Hbf_{I}^{-1} G_{I}(\zero) + (K \log n) \size{I}^{-\frac{1}{2}} \ebf\bigr) \\
    \ge & (K \size{I}^{-1} \log n) \cdot \ebf^\top \Hbf_{I} \ebf - C_{\ref*{sec:proof_mest}.6} \size{I}^{-1} \log n > 0.
\end{align*}
It means that $\hat{\btheta}_{I}$ is in the open ball $\set{\btheta_{I}^\circ - \Hbf_{I}^{-1} G_{I}(\zero) + (K \log n) \size{I}^{-1} \ebf: \norm{\ebf}_2 < 1}$. By taking the union bounds over the intervals $I \subset (0, n]$, uniformly with probability at least $1 - \exp(- C_{\ref*{sec:proof_mest}.7} \log n)$,
\begin{equation}\label{equ:bahadur}
    (\hat{\btheta}_{I} - \btheta_{I}^\circ) = - \Hbf_{I}^{-1} G_{I}(\zero) + \rbf_{I},
\end{equation}
where $\max_{\size{I} \subset (0, n]} \rbf_{I} \size{I} / \log n = O(1)$.

We have now obtained the uniform Bahadur representation, which holds over $I \subset (0, n]$ with high probability. To measure the difference between $\hat{\btheta}_{I}$ and $\hat{\btheta}_{R}$, we first consider the population one. Recall that $R \in \R$ is the relief interval of $I$. First of all, we study the population minimizers. By the $\rho$-strong convexity and the definition of $\btheta_{I}^\circ$ and $\btheta_{R}^\circ$,
\begin{equation*}
    0 \le \overline{\Lcal}(I, \btheta_{R}^\circ) - \overline{\Lcal}(I, \btheta_{I}^\circ) \le \nabla_{\theta} \overline{\Lcal}(I, \btheta_{R}^\circ)^\top (\btheta_{R}^\circ - \btheta_{I}^\circ) - \frac{\rho \size{I}}{2} \norm{\btheta_{R}^\circ - \btheta_{I}^\circ}_2^2,
\end{equation*}
which implies that
\begin{equation*}
    \norm{\btheta_{R}^\circ - \btheta_{I}^\circ}_2 \le \frac{2}{\rho \size{I}} \Bignorm{\sum_{i \in I \setminus R} \Ebb g(\zbf_i, \btheta_{R}^\circ)}_2 \le \frac{2 \zeta}{\rho \size{I}} \sum_{i \in I \setminus R} \norm{\btheta_{R}^\circ - \btheta_{i}^\circ}_2 = O(1-r).
\end{equation*}

Assume that the Bahadur representation Eq. (\ref{equ:bahadur}) holds thereafter. We have the following identity of the difference between $\htheta_I$ and $\htheta_R$,
\begin{equation*}
    \hat{\btheta}_{I} - \hat{\btheta}_{R} = \btheta_{I}^\circ - \btheta_{R}^\circ + \Hbf_{R}^{-1} G_{R}(\zero) - \Hbf_{I} G_{I}(\zero) + \rbf_{I} - \rbf_{R}.
\end{equation*}
For $\Hbf_{R}^{-1} G_{R}(\zero) - \Hbf_{I} G_{I}(\zero)$, further consider the following decomposition,
\begin{equation*}
    \Hbf_{R}^{-1} G_{R}(\zero) - \Hbf_{I} G_{I}(\zero) = (\Hbf_{R}^{-1} - \Hbf_{I}^{-1}) G_{R}(\zero) + \Hbf_{I}^{-1} \{G_{R}(\zero) - G_{I}(\zero)\}.
\end{equation*}
For the first part, by the sub-exponential assumption (d), with probability at least $1 - \exp(-C_{u} \log n)$,
\begin{equation*}
    \norm{(\Hbf_{R}^{-1} - \Hbf_{I}^{-1}) G_{R}(\zero)}_2 \le C_{\ref*{sec:proof_mest}.8} \norm{\btheta_{I}^\circ - \btheta_{R}^\circ}_2  \biggl[\Bigl(\frac{\log n}{\size{R}}\Bigr)^{\frac{1}{2}} + \frac{\log n}{\size{R}}\biggr].
\end{equation*}
For the second part,
\begin{equation*}
    G_{R}(\zero) - G_{I}(\zero) = \sum_{i \in R} \Bigl[\frac{1}{\size{R}} g(\zbf_i, \btheta_{R}^\circ) - \frac{1}{\size{I}} g(\zbf_i, \btheta_{I}^\circ)\Bigr] - \sum_{i \in I \setminus R} \frac{1}{\size{I}} g(\zbf_i, \btheta_{I}^\circ) \triangleq \frac{1}{\size{I}} \sum_{i \in I} \xbf_{i},
\end{equation*}
where $\xbf_{i} = [g(\zbf_i, \btheta_{R}^\circ) \size{I} / \size{R}] - g(\zbf_i, \btheta_{I}^\circ)$ for $i \in R$ and $\xbf_i = - g(\zbf_i, \btheta_{I}^\circ)$ for $i \in I \setminus R$. For any individual $i \in R$, by assumptions (d) and (g),
\begin{align*}
    \norm{\xbf_i}_{\Psi_1} &= \Bignorm{ \{g(\zbf_i, \btheta_{R}^\circ) - g(\zbf_i, \btheta_{I}^\circ)\} + \frac{(1-r)}{r} g(\zbf_i, \btheta_{R}^\circ)}_{\Psi_1} \\
    & \le \norm{\btheta_{I}^\circ - \btheta_{R}^\circ}_2 + \frac{1-r}{r} (C_{\ref*{sec:proof_mest}.3} \norm{\btheta_{R}^\circ - \btheta_{i}^\circ}_2 + C_{\ref*{sec:proof_mest}.1}) \le \norm{\btheta_{I}^\circ - \btheta_{R}^\circ}_2 + \frac{1-r}{r}C_{\ref*{sec:proof_mest}.9}
\end{align*}
For $i \in I \setminus R$,
\begin{equation*}
    \norm{\xbf_i}_{\Psi_1} \le (C_{\ref*{sec:proof_mest}.3} \norm{\btheta_{I}^\circ - \btheta_{i}^\circ}_2 + C_{\ref*{sec:proof_mest}.1}) \le C_{\ref*{sec:proof_mest}.9}.
\end{equation*}
In the above two inequalities, we make use of Condition (j), the boundness of parameters. By Bernstein's inequality (Lemma \ref{lem:bernstein}), with probability at least $1 - \exp(-C_{u} \log n)$,
\begin{equation*}
    \norm{G_{R}(\zero) - G_{I}(\zero)}_2 \le C_{\ref*{sec:proof_mest}.10}\biggl[\bigl(\norm{\btheta_{I}^\circ - \btheta_{R}^\circ}_2 + r^{-\frac{1}{2}}(1 - r)^{\frac{1}{2}}\bigr) \Bigl(\frac{\log n}{\size{I}}\Bigr)^{\frac{1}{2}} + \frac{\log n}{r \size{I}}\biggr].
\end{equation*}
Overall we obtain,
\begin{equation*}
    \norm{\hat{\btheta}_{I} - \hat{\btheta}_{R}}_2 \le O\biggl(\norm{\btheta_{I}^\circ - \btheta_{R}^\circ}_2 + (1 - r)^{\frac{1}{2}} \Bigl(\frac{\log n}{r\size{I}}\Bigr)^{\frac{1}{2}} + \frac{\log n}{r\size{I}}\biggr).
\end{equation*}

By the definition of $\hat{\btheta}_{R}$, one obtains $\nabla_{\btheta}\Lcal(I, \hat{\btheta}_{R}) = \sum_{i \in I \setminus R} g(\zbf_i, \hat{\btheta}_{R})$. Similarly, by the $\delta$-triangulation argument used in the proof of the Bahadur representation, with probability at least $1 - \exp(-C_u \log n)$, uniformly for all intervals $I$,
\begin{equation*}
    \Bignorm{\sum_{i \in I \setminus R} \Bigl\{g(\zbf_i, \hat{\btheta}_{R}) - g(\zbf_i, \btheta_{R}^\circ) - \Ebb\bigl[g(\zbf_i, \hat{\btheta}_{R}) - g(\zbf_i, \btheta_{R}^\circ)\bigr]\Bigr\}}_2 = O(\log n),
\end{equation*}
\begin{equation*}
    \Bignorm{\sum_{i \in I \setminus R} \Ebb\{g(\zbf_i, \hat{\btheta}_{R}) - g(\zbf_i, \btheta_{R}^\circ)\}}_2 \le \zeta (1 - r) \size{I} \norm{\hat{\btheta}_{R} - \btheta_{R}^\circ}_2 = O\!\biggl(\!(1 - r)\biggl\{\!\!\Bigl(\frac{\size{I} \log n}{r}\Bigr)^{\frac{1}{2}} + \frac{\log n}{r}\!\biggr\}\!\biggr),
\end{equation*}

\begin{equation*}
    \Bignorm{\sum_{i \in I \setminus R} g(\zbf_i, \btheta_{R}^\circ)}_2 = \Bignorm{\sum_{i \in I \setminus R} \Ebb g(\zbf_i, \btheta_{R}^\circ)}_2 + O(\{(1-r)\size{I} \log n\}^{\frac{1}{2}} + \log n).
\end{equation*}
Combining the above three upper bounds,
\begin{equation*}
    \nabla_{\btheta}\Lcal(I, \hat{\btheta}_{R}) = \Bignorm{\sum_{i \in I \setminus R} \Ebb g(\zbf_i, \btheta_{R}^\circ)}_2 + O\biggl( (1 - r)^{\frac{1}{2}} \Bigl(\frac{\size{I}\log n}{r}\Bigr)^{\frac{1}{2}} + \frac{\log n}{r}\biggr).
\end{equation*}
By the convexity condition (h),
\begin{align}
    & \frac{1}{\size{I}}\{\Lcal(I, \hat{\btheta}_{R}) - \Lcal(I, \hat{\btheta}_{I})\} \le \frac{1}{\size{I}}\nabla_{\btheta}\Lcal(I, \hat{\btheta}_{R})^\top (\hat{\btheta}_{R} - \hat{\btheta}_{I}) \le \frac{1}{\size{I}}\norm{\nabla_{\btheta}\Lcal(I, \hat{\btheta}_{R})}_2 \norm{\hat{\btheta}_{R} - \hat{\btheta}_{I}}_2 \nonumber\\
    = & O\biggl(\frac{1}{\rho \size{I}^2} \Bignorm{\sum_{i \in I \setminus R} \Ebb g(\zbf_i, \btheta_{R}^\circ)}_2^2 + \frac{(1 - r)\log n}{r\size{I}} + \frac{(\log n)^2}{r^2\size{I}^2}\biggr).
\end{align}
When $I = (s, e]$ contains no changepoint, or it is nearly homogeneous such that if a true changepoint $\tau \in I$, then $\min(\tau - s, e - \tau) = O(\log n)$, we have $\sum_{i \in I \setminus R} \Ebb g(\zbf_i, \btheta_R^\circ) = O(\min(\tau - s, e - \tau)) = O(\log n)$. Therefore,
\begin{equation*}
    \frac{1}{\size{I}}\{\Lcal(I, \hat{\btheta}_{R}) - \Lcal(I, \hat{\btheta}_{I})\} = O\biggl(\frac{(1 - r)\log n}{r\size{I}} + \frac{(\log n)^2}{r^2\size{I}^2}\biggr).
\end{equation*}

\section{Proof of Lemma \ref{lem:loc_err_g}}

We first introduce some notations. For a given changepoint estimation $\tau \in [n]$ and a changepoint set $\cpsset = \set{0 = \tau_0 < \tau_1 < \dots < \tau_K < \tau_{K+1} = n}$, denote
$$\L(\Tcal) \triangleq \sum_{k=1}^{K+1} \L((\tau_{k-1},\tau_{k}];\widetilde{\btheta}((\tau_{k-1},\tau_{k}]))$$
as the loss function in Eq.~\eqref{object}, $\Kcal_{+}(\tau,\cpsset) \triangleq \min_{k}\set{k : \tau_k > \tau}$ and $\Kcal_{-}(\tau,\cpsset) \triangleq \max_{k} \set{k: \tau_k < \tau}$. For simplicity, further denote $k_{\tau,+}^\ast = \Kcal_{+}(\tau, \truecps)$, $\hat{k}_{\tau,+} = \Kcal_{+}(\tau, \estcps)$, $k_{\tau,-}^{\ast} = \Kcal_{-}(\tau, \truecps)$ and $\hat{k}_{\tau,-} = \Kcal_{-}(\tau, \estcps)$. Let $\estcps = \set{\hat{\tau}_1,\dots,\hat{\tau}_{\widehat{K}}}$ be the minimizer of Eq. (\ref{object}). Denote $\delta_{\mathsf{m}} = C_{\mathsf{m}} s \log (p \vee n)$ and $\delta_k = \widetilde{C} s \log(p \vee n) \Delta_{k}^{-2}$ where $\Delta_k = \norm{\btheta_{k+1}^\ast - \btheta_k^\ast}_{\Sigma}$, and $\Hcal = \set{(\hat{\tau}_a, \hat{\tau}_{a+1}]: \exists k \in [K^\ast], \min(\tau_k^\ast - \hat{\tau}_a,\hat{\tau}_{a+1} - \tau_k^\ast) > \delta_k}$.

Assume that $\Hcal \neq \varnothing$, that is, $\exists k \in [K^\ast]$ such that $\estcps \cap [\tau_k^\ast - \delta_k, \tau_k^\ast + \delta_k] = \varnothing$. For such $h$ and $a$, without loss of generality assume that $\tau_k^\ast - \hat{\tau}_a > \delta_k$, it can be observed that $(\tau_k^\ast - \delta_k, \tau_k^\ast + \delta_k] \subset (\hat{\tau}_a, \hat{\tau}_{a+1}]$ and $\Delta_{(\hat{\tau}_a, \hat{\tau}_{a+1}]}^2 (\hat{\tau}_{a+1} - \hat{\tau}_a) \ge 2 \delta_k \Delta_{(\tau_k^\ast - \delta_k, \tau_k^\ast + \delta_k]}^2 = \delta_k \Delta_k^2 / 2 = 2^{-1}\widetilde{C} s \log(p \vee n)$.

To move further, we need the following definitions to divide $\Hcal$ into four groups.
\begin{definition}[Separability of a point]
    For a changepoint estimation $\tau$ and the true changepoint set $\truecps$, let $u = k_{\tau,-}^\ast$ and $v = k_{\tau,+}^\ast$.
    We say that $\tau$ is separable from the left if $\tau - \tau_u^\ast > \delta_u \vee \delta_{\mathsf{m}}$ and separable from the right if $\tau_v^\ast - \tau > \delta_v \vee \delta_{\mathsf{m}}$.
    Otherwise, $\tau$ is inseparable from the left (right).
\end{definition}
\begin{definition}[Separability of an interval]
    For the intervals $(\tau_l,\tau_r] \in \Hcal$, we make the following definitions,
    \begin{enumerate}[label=$\Hcal_\arabic*:\,$]
        \item $(\tau_l,\tau_r] \in (0, n]$ is separable if $\tau_l$ is separable from the right and $\tau_r$ is separable from the left.
        \item $(\tau_l,\tau_r] \in (0, n]$ is left-separable if $\tau_l$ is separable from the right and $\tau_r$ is inseparable from the left.
        \item $(\tau_l,\tau_r] \in (0, n]$ is right-separable if $\tau_l$ is inseparable from the right and $\tau_r$ is separable from the left.
        \item $(\tau_l,\tau_r] \in (0, n]$ is inseparable if $\tau_l$ is inseparable from the right and $\tau_r$ is inseparable from the left.
    \end{enumerate}
\end{definition}
Now the sub-intervals in $\Hcal$ have been classified into four groups $\Hcal = \Hcal_1 \cup \Hcal_2 \cup \Hcal_3 \cup \Hcal_4$. We will show that $\Hcal = \varnothing$ by emptying these groups.

\subsubsection*{Case 1: $\Hcal_1 = \varnothing$}
For $(\hat{\tau}_a, \hat{\tau}_{a+1}] \in \Hcal_1$, let $h = k_{\hat{\tau}_a,+}^\ast$. Denote $\cpsset_a = \set{\tau_{h}^\ast,\dots,\tau_{h+t}^\ast}=\truecps \cap (\hat{\tau}_a, \hat{\tau}_{a+1})$. Let $\widetilde{\cpsset} = \estcps \cup \cpsset_a$. Since $\gamma = C_{\gamma} s \log(p \vee n)$,
\begin{align*}
    &\L(\estcps) - \L(\widetilde{\cpsset}) = \L_{(\hat{\tau}_a, \hat{\tau}_{a+1}]} - \Bigl[\L_{(\hat{\tau}_a, \tau_h^\ast]} + \L_{(\tau_{h+t}^\ast, \hat{\tau}_{a+1}]} + \sum_{j=h}^{h+t-1} \L_{(\tau_j^\ast, \tau_{j+1}^\ast]} + (t+1)\gamma\Bigr]\\
    > & (1 - C_{\ref*{lem:loc_err_g}.3}) \Delta_{(\hat{\tau}_a, \hat{\tau}_{a+1}]}^2 (\hat{\tau}_{a+1} - \hat{\tau}_a) - (t+2)C_{\ref*{lem:loc_err_g}.1} s \log(p \vee n) - (t+1) \gamma\\
    = & (1 - C_{\ref*{lem:loc_err_g}.3}) \sum_{i \in (\hat{\tau}_a, \hat{\tau}_{a+1}]} \norm{\btheta_i^\circ - \btheta_{(\hat{\tau}_a, \hat{\tau}_{a+1}]}^\circ}_{\Sigma}^2 - [(t+2) C_{\ref*{lem:loc_err_g}.1} + (t+1) C_{\gamma}] s \log(p \vee n) \\
    \ge & \Bigl[(1 - C_{\ref*{lem:loc_err_g}.3}) (t + 1) 2^{-1} \widetilde{C} - (t+2) C_{\ref*{lem:loc_err_g}.1} - (t+1) C_{\gamma}\Bigr] s \log(p \vee n) > 0,
\end{align*}
provided that $\widetilde{C} \ge 2 (1 - C_{\ref*{lem:loc_err_g}.3})^{-1} (2 C_{\ref*{lem:loc_err_g}.1} + C_{\gamma})$. Therefore $\Hcal_1 = \varnothing$.

\subsubsection*{Case 2: $\Hcal_2 = \Hcal_3 = \varnothing$}

Without loss of generality, by the symmetry of $\Hcal_2$ and $\Hcal_3$, we only show that $\Hcal_3 = \varnothing$. If the claim does not hold, one can choose $(\hat{\tau}_a, \hat{\tau}_{a+1}] \in \Hcal_3$ to be the leftmost one. Hence $\hat{\tau}_a$ must be separable from the left by Condition~\ref{cond:change}. Since $\Hcal_1 = \varnothing$ and $(\hat{\tau}_a, \hat{\tau}_{a+1}]$ is the leftmost interval in $\Hcal_3$, one obtains $(\hat{\tau}_{a-1}, \hat{\tau}_a] \not\in \Hcal$. Denote $h = k_{\hat{\tau}_a,+}^\ast$ and $\cpsset_a = \truecps \cap (\hat{\tau}_a + \delta_{\mathsf{m}}, \hat{\tau}_{a+1}-\delta_{\mathsf{m}}) = \set{\tau_{h+1}^\ast,\dots,\tau_{h+t}^\ast}$ ($t=0$ if $\cpsset_a = \varnothing$). Let $\widetilde{\cpsset} = (\estcps \setminus {\hat{\tau}_a}) \cup {\tau_h^\ast} \cup \cpsset_a = (\estcps \setminus {\hat{\tau}_a}) \cup \set{\tau_j^\ast}_{j=h}^{h+t}$.
\begin{align}\label{equ:case2_1}
    \L(\estcps) - \L(\widetilde{\cpsset}) = & \L_{(\hat{\tau}_a, \hat{\tau}_{a+1}]} + \bigl(\L_{(\hat{\tau}_{a-1}, \hat{\tau}_a]} - \L_{(\hat{\tau}_{a-1}, \tau_h^\ast]}\bigr) - \Bigl[\sum_{j=h}^{h+t-1} \L_{(\tau_j^\ast, \tau_{j+1}^\ast]} + \L_{(\tau_{h+t}^\ast, \hat{\tau}_{a+1}]} + t \gamma \Bigr]\nonumber\\
    > & (1 - C_{\ref*{lem:loc_err_g}.3})\Delta_{(\hat{\tau}_a, \hat{\tau}_{a+1}]}^2 (\hat{\tau}_{a+1} - \hat{\tau}_a) - [(t+1) C_{\ref*{lem:loc_err_g}.1} + t C_{\gamma}] s \log(p \vee n) \nonumber\\
    + & \Bigl(\sum_{i \in (\hat{\tau}_{a},\tau_h^\ast]} \epsilon_i^2 + \L_{(\hat{\tau}_{a-1}, \hat{\tau}_a]} - \L_{(\hat{\tau}_{a-1}, \tau_h^\ast]}\Bigr).
\end{align}
Since $(\hat{\tau}_{a-1}, \hat{\tau}_a] \not\in \Hcal$ and $0 < \hat{\tau}_a - \tau_h^\ast < \delta_{\mathsf{m}}$, one must obtain that either $(\hat{\tau}_{a-1}, \hat{\tau}_a) \cap \cpsset^\ast = \varnothing$ or $0 < \tau_{h-1}^\ast - \hat{\tau}_{a-1} < \delta_{h-1} = \widetilde{C} \Delta_{h-1}^{-2} s \log(p \vee n)$.

For the first scenario, under $\mathbb{G}_1$,
\begin{equation*}
    \Bigabs{\sum_{i \in (\hat{\tau}_{a},\tau_h^\ast]} \epsilon_i^2 + \L_{(\hat{\tau}_{a-1}, \hat{\tau}_a]} - \L_{(\hat{\tau}_{a-1}, \tau_h^\ast]}} \le 2 C_{\ref*{lem:loc_err_g}.1} s \log(p \vee n).
\end{equation*}
Hence,
\begin{align*}
    \L(\estcps) - \L(\widetilde{\cpsset}) & > (1 - C_{\ref*{lem:loc_err_g}.3})\Delta_{(\hat{\tau}_{a}, \hat{\tau}_{a+1}]}^2 (\hat{\tau}_{a+1} - \hat{\tau}_{a}) - [(t+3) C_{\ref*{lem:loc_err_g}.1} + t C_{\gamma}] s \log(p \vee n)\\
    & \ge \Bigl\{(1 - C_{\ref*{lem:loc_err_g}.3})(t \vee 1)2^{-1}\widetilde{C} - (t + 3) C_{\ref*{lem:loc_err_g}.1} - t C_{\gamma} \Bigr\} s \log(p \vee n) > 0,
\end{align*}
provided that $\widetilde{C} \ge 2 (1 - C_{\ref*{lem:loc_err_g}.3})^{-1} (4 C_{\ref*{lem:loc_err_g}.1} + C_{\gamma})$.

For the second scenario, let $I_1 = (\hat{\tau}_{a-1}, \hat{\tau}_a]$ and $I_2 = (\hat{\tau}_{a-1}, \tau_h^\ast]$. Firstly, we will bound the gap $\Delta_{I_2}^2 \size{I_2} - \Delta_{I_1}^2 \size{I_1}$. Since $I_1 \subset I_2$, we have $\Delta_{I_2}^2 \size{I_2} - \Delta_{I_1}^2 \size{I_1} \ge 0$.

Denote $d_1 = \tau_{h-1}^\ast - \hat{\tau}_{a-1}$, $d_2 = \hat{\tau}_a - \tau_{h-1}^\ast$ and $d_3 = \tau_h^\ast - \hat{\tau}_a$. Recall that $\Delta_{h-1} = \norm{\btheta_h^\ast - \btheta_{h-1}^\ast}_{\Sigma}$ and the definition of $\Delta_I^2$, we have
\begin{equation*}
    \Delta_{I_2}^2 \size{I_2} = \frac{d_1 (d_2 + d_3)}{d_1 + d_2 + d_3} \Delta_{h-1}^2,\, \Delta_{I_1}^2 \size{I_1} = \frac{d_1 d_2}{d_1 + d_2} \Delta_{h-1}^2.
\end{equation*}
It follows that
\begin{equation*}
    \Delta_{I_2}^2 \size{I_2} - \Delta_{I_1}^2 \size{I_1} = \frac{d_1^2 d_3 \Delta_{h-1}^2}{(d_1 + d_2)(d_1 + d_2 + d_3)} \le \frac{\widetilde{C}^2 ( \widetilde{C} \vee C_{\mathsf{m}})}{C_{\mathsf{snr}}(C_{\mathsf{snr}} - \widetilde{C} \vee C_{\mathsf{m}})} s \log(p \vee n).
\end{equation*}
where the last inequality is from the conditions $d_1 \le \widetilde{C} \Delta_{h-1}^{-2} s \log(p \vee n)$, $d_3 \le \delta_h \vee \delta_{\mathsf{m}}$ and $d_1 + d_2 + d_3 \ge C_{\mathsf{snr}} s \log(p \vee n) [1 + \Delta_{h-1}^{-2} + \Delta_{h}^{-2}]$. Denote $C_{m,1} = \widetilde{C}^2 ( \widetilde{C} \vee C_{\mathsf{m}}) / \{ C_{\mathsf{snr}}(C_{\mathsf{snr}} - \widetilde{C} \vee C_{\mathsf{m}})\}$.

By $0 < \tau_{h-1}^\ast - \hat{\tau}_{a-1} < \delta_{h-1} = \widetilde{C} \Delta_{h-1}^{-2} s \log(p \vee n)$, $\Delta_{I_1}^2 \size{I_1} \le \Delta_{I_2}^2 \size{I_2} \le \widetilde{C} s \log(p \vee n)$. It means that $I_1 \subset I_2 \in E_2^+ \subseteq E_2^-$.
Hence by $\mathbb{G}_2^- \cap \mathbb{G}_2^+$,
\begin{equation}\label{equ:case2_low_diff}
    \sum_{i \in (\hat{\tau}_{a},\tau_h^\ast]} \epsilon_i^2 + \L_{(\hat{\tau}_{a-1}, \hat{\tau}_a]} - \L_{(\hat{\tau}_{a-1}, \tau_h^\ast]} > -(2 C_{\ref*{lem:loc_err_g}.2} + C_{m,1}) s \log(p \vee n).
\end{equation}
By Eq. (\ref{equ:case2_1}) and Eq. (\ref{equ:case2_low_diff}),
\begin{align*}
    & \L(\estcps) - \L(\widetilde{\cpsset}) \\
    >& (1 - C_{\ref*{lem:loc_err_g}.3})\Delta_{(\hat{\tau}_{a}, \hat{\tau}_{a+1}]}^2 (\hat{\tau}_{a+1} - \hat{\tau}_{a}) - [(t + 1) C_{\ref*{lem:loc_err_g}.1} + t C_{\gamma} + 2 C_{\ref*{lem:loc_err_g}.2} + C_{m,1}] s \log(p \vee n)\\
    \ge & [(1 - C_{\ref*{lem:loc_err_g}.3})(t \vee 1) 2^{-1} \widetilde{C} - (t + 1) C_{\ref*{lem:loc_err_g}.1} - t C_{\gamma} - 2 C_{\ref*{lem:loc_err_g}.2} - C_{m,1}] s \log(p \vee n) > 0,
\end{align*}
provided that $\widetilde{C} \ge 2 (1 - C_{\ref*{lem:loc_err_g}.3})^{-1} (2 C_{\ref*{lem:loc_err_g}.1} + C_{\gamma} + 2 C_{\ref*{lem:loc_err_g}.2} + C_{m,1})$. Hence $\Hcal_2 \cup \Hcal_3 = \varnothing$.

\subsubsection*{Case 3: $\Hcal_4 = \varnothing$}

Similar to Case 2, let $(\hat{\tau}_a, \hat{\tau}_{a+1}] \in \Hcal_4$, then $\hat{\tau}_a$ is separable from the left and $\hat{\tau}_{a+1}$ is separable from the right. By the fact that $\Hcal_1 \cup \Hcal_2 \cup \Hcal_3 = \varnothing$, we also obtain $(\hat{\tau}_{a-1}, \hat{\tau}_a] \not\in \Hcal$ and $(\hat{\tau}_{a+1}, \hat{\tau}_{a+2}] \not\in \Hcal$. Let $h = k_{\hat{\tau}_a,+}^\ast$ and $h + t = k_{\hat{\tau}_{a+1},-}^\ast$. Denote $\cpsset_a = \set{\tau_h^\ast,\dots,\tau_{h+t}^\ast}$ and $\widetilde{\cpsset} = (\estcps \setminus \set{\hat{\tau}_a, \hat{\tau}_{a+1}} \cup \cpsset_a$. We have
\begin{align*}
    \L(\estcps) - \L(\widetilde{\cpsset}) =& \L_{(\hat{\tau}_a, \hat{\tau}_{a+1}]} +[\L_{(\hat{\tau}_{a-1}, \hat{\tau}_a]} + \L_{(\hat{\tau}_{a+1}, \hat{\tau}_{a+2}]} - \L_{(\hat{\tau}_{a-1}, \tau_{h}^\ast]} - \L_{(\tau_{h+t}^\ast, \hat{\tau}_{a+2}]}]\\
    - & \sum_{j=h}^{h+t-1} \L_{(\tau_j^\ast, \tau_{j+1}^\ast]} - (t-1)\gamma\\
    > & (1 - C_{\ref*{lem:loc_err_g}.3})\Delta_{(\hat{\tau}_a, \hat{\tau}_{a+1}]}^2 (\hat{\tau}_{a+1} - \hat{\tau}_a) - ( t C_{\ref*{lem:loc_err_g}.1} + (t-1) C_{\gamma}) s \log(p \vee n) \\
    + & [ \sum_{i \in (\hat{\tau}_a, \tau_h^\ast] \cup (\tau_{h+1}^\ast, \hat{\tau}_{a+1}]} \epsilon_i^2 + \L_{(\hat{\tau}_{a-1}, \hat{\tau}_a]} + \L_{(\hat{\tau}_{a+1}, \hat{\tau}_{a+2}]} - \L_{(\hat{\tau}_{a-1}, \tau_{h}^\ast]} - \L_{(\tau_{h+t}^\ast, \hat{\tau}_{a+2}]}].
\end{align*}

Following the same discussion in Case 2, that is, Eq. (\ref{equ:case2_low_diff}), we have
\begin{align*}
    &\sum_{i \in (\hat{\tau}_a, \tau_h^\ast] \cup (\tau_{h+1}^\ast, \hat{\tau}_{a+1}]} \epsilon_i^2 +  \L_{(\hat{\tau}_{a-1}, \hat{\tau}_a]} + \L_{(\hat{\tau}_{a+1}, \hat{\tau}_{a+2}]} - \L_{(\hat{\tau}_{a-1}, \tau_{h}^\ast]} - \L_{(\tau_{h+t}^\ast, \hat{\tau}_{a+2}]}\\
    >& -(4 C_{\ref*{lem:loc_err_g}.2} + 2 C_{m,1}) s \log(p \vee n).
\end{align*}
Hence,
\begin{align*}
    & \L(\estcps) - \L(\widetilde{\cpsset})\\
    > & (1 - C_{\ref*{lem:loc_err_g}.3})\Delta_{(\hat{\tau}_a, \hat{\tau}_{a+1}]}^2 (\hat{\tau}_{a+1} - \hat{\tau}_a) - [t C_{\ref*{lem:loc_err_g}.1} + (t-1) C_{\gamma} + 4 C_{\ref*{lem:loc_err_g}.2} + 2 C_{m,1}] s \log(p \vee n) \\
    \ge & \Bigl\{(1 - C_{\ref*{lem:loc_err_g}.3})[(t-1) \vee 1] 2^{-1}\widetilde{C} - t C_{\ref*{lem:loc_err_g}.1} - (t-1) C_{\gamma} - 4 C_{\ref*{lem:loc_err_g}.2} - 2 C_{m,1}\Bigr\} s \log(p \vee n) \ge 0,
\end{align*}
provided that $\widetilde{C} \ge 2 (1 - C_{\ref*{lem:loc_err_g}.3})^{-1}(2 C_{\ref*{lem:loc_err_g}.1} + C_{\gamma} + 4 C_{\ref*{lem:loc_err_g}.2} + 2 C_{m,1})$.

In summary, we obtain $\Hcal = \varnothing$ provided that $\widetilde{C} \ge 2 (1 - C_{\ref*{lem:loc_err_g}.3})^{-1}(2 C_{\ref*{lem:loc_err_g}.1} + C_{\gamma} + 4 C_{\ref*{lem:loc_err_g}.2} + 2 C_{m,1})$. Hence $\max_{1 \le j \le K^\ast} \min_{1 \le k \le \widehat{K}} \Delta_j^2 |\tau_j^\ast - \hat{\tau}_k| \le \widetilde{C} s \log(p \vee n)$. It also implies that $\widehat{K} \ge K^\ast$.

It remains to show that $\hat{K} \le K^\ast$. Otherwise, assume that $\hat{K} > K^\ast$. Then there must be $j \in [0, K^\ast]$ and $k \in [1, \hat{K}]$ such that $\tau_j^\ast - \delta_j \le \hat{\tau}_{k-1} < \hat{\tau}_k < \hat{\tau}_{k + 1} \le \tau_{j+1}^\ast + \delta_{j+1}$. Similar to the decomposition of $\Hcal$, we can also divide it into four groups.
\begin{enumerate}[label=$\Gcal_\arabic*:\,$]
    \item $\tau_j^\ast \le \hat{\tau}_{k-1} < \hat{\tau}_k < \hat{\tau}_{k + 1} \le \tau_{j+1}^\ast$.
    \item $\tau_j^\ast - \delta_j \le \hat{\tau}_{k-1} < \tau_j^\ast$ and $\tau_j^\ast \le \hat{\tau}_k < \hat{\tau}_{k+1} \le \tau_{j+1}^\ast$.
    \item $\tau_j^\ast \le \hat{\tau}_{k-1} < \hat{\tau}_k \le \tau_{j+1}^\ast$ and $\tau_{j+1}^\ast < \hat{\tau}_{k+1} \le \tau_{j+1}^\ast + \delta_{j+1}$.
    \item $\tau_j^\ast - \delta_j \le \hat{\tau}_{k-1} < \tau_j^\ast \le \hat{\tau}_k \le \tau_{j+1}^\ast < \hat{\tau}_{k+1} \le \tau_{j+1}^\ast + \delta_{j+1}$.
\end{enumerate}

\subsubsection*{Case 1: $\Gcal_1 = \varnothing$}
Let $\widetilde{\Tcal} = \estcps \setminus \set{\hat{\tau}_k}$.
We have
\begin{align*}
    \L(\widetilde{\Tcal}) - \L(\estcps) &= \L_{(\hat{\tau}_{k-1}, \hat{\tau}_{k+1}]} - \L_{(\hat{\tau}_{k-1}, \hat{\tau}_k]} -  \L_{(\hat{\tau}_k, \hat{\tau}_{k+1}]} - \gamma \\
    &< (3 C_{\ref*{lem:loc_err_g}.1} - C_{\gamma}) s \log(p \vee n) \le 0,
\end{align*}
provided that $C_{\gamma} \ge 3 C_{\ref*{lem:loc_err_g}.1}$.

\subsubsection*{Case 2: $\Gcal_2 \cup \Gcal_3 = \varnothing$}
We will show that $\Gcal_2 = \varnothing$ because the proof for $\Gcal_3 = \varnothing$ is the same by symmetry. Assume that the pair $(j, k)$ is the leftmost one that satisfies $\Gcal_2$. It implies that $\hat{\tau}_{k-2} \in [\tau_{j-1}^\ast - \delta_{j-1}, \tau_{j-1}^\ast + \delta_{j-1}]$. Otherwise assume $\hat{\tau}_{k-2} > \tau_{j-1}^\ast + \delta_{j-1}$. Since $\max_{1 \le j \le K^\ast} \min_{1 \le k \le \widehat{K}} \Delta_{j}^2 \abs{\tau_j^\ast - \hat{\tau}_k} \le \widetilde{C} s \log(p \vee n)$, there must be $\hat{\tau}_{k - h} \in [\tau_{j-1}^\ast - \delta_{j-1}, \tau_{j-1}^\ast + \delta_{j-1}]$ for some $h > 2$. It contradicts the fact that $\Gcal_1 = \varnothing$ and the choice of $k$.

Let $\widetilde{\Tcal} = \set{\tau_j^\ast} \cup \estcps \setminus \set{\hat{\tau}_{k-1}, \hat{\tau}_k}$. Note that $(\hat{\tau}_{k-2}, \hat{\tau}_{k-1}], (\hat{\tau}_{k-1}, \hat{\tau}_{k}] \in E_2^-$ and $(\hat{\tau}_{k-2}, \tau_j^\ast] \in E_2^+$, we have
\begin{align*}
    \L(\widetilde{\Tcal}) - \L(\estcps) &= \L_{(\hat{\tau}_{k-2}, \tau_j^\ast]} + \L_{(\tau_j^\ast, \hat{\tau}_{k+1}]} - \Bigl[\sum_{t = k-2}^{k} \L_{(\hat{\tau}_t, \hat{\tau}_{t+1}]} + \gamma \Bigr] \\
    & = [\L_{(\hat{\tau}_{k-2}, \tau_j^\ast]} - \L_{(\hat{\tau}_{k-2}, \hat{\tau}_{k-1}]}] + \L_{(\tau_j^\ast, \hat{\tau}_{k+1}]} - \Bigl[\sum_{t = k-1}^{k} \L_{(\hat{\tau}_t, \hat{\tau}_{t+1}]} + \gamma \Bigr]\\
    & < \Bigl[ \L_{(\hat{\tau}_{k-2}, \tau_j^\ast]} - \L_{(\hat{\tau}_{k-2}, \hat{\tau}_{k-1}]} - \sum_{i \in (\hat{\tau}_{k-1}, \tau_j^\ast]} \epsilon_i^2 \Bigr] + (2 C_{\ref*{lem:loc_err_g}.1} + C_{\ref*{lem:loc_err_g}.2} - C_{\gamma}) s \log(p \vee n) \\
    &\le (2 C_{\ref*{lem:loc_err_g}.1} + 3 C_{\ref*{lem:loc_err_g}.2} + C_{m,1} - C_{\gamma}) s \log(p \vee n) \le 0,
\end{align*}
provided $C_{\gamma} \ge 2 C_{\ref*{lem:loc_err_g}.1} + 3 C_{\ref*{lem:loc_err_g}.2} + C_{m,1}$. The second last inequality is from Eq. (\ref{equ:case2_low_diff}).

\subsubsection*{Case 3: $\Gcal_4 = \varnothing$}
Now $\Gcal_1 \cup \Gcal_2 \cup \Gcal_3 = \varnothing$. Assume that $\set{\hat{\tau}_{k-1}, \hat{\tau}_k, \hat{\tau}_{k+1}}$ satisfies $\Gcal_4$. Similar to the analysis of $\Gcal_2 = \varnothing$, we have $\hat{\tau}_{k-2} \in [\tau_{j-1}^\ast - \delta_{j-1}, \tau_{j-1}^\ast + \delta_{j-1}]$ and $\hat{\tau}_{k+2} \in [\tau_{j+1}^\ast + \delta_{j+1}, \tau_{j+1}^\ast + \delta_{j+1}]$. Follow the same arguments in the proof for $\Hcal_4 = \varnothing$, we can set $\widetilde{\Tcal} = \set{\tau_j^\ast, \tau_{j+1}^\ast} \cup \estcps \setminus \set{\hat{\tau}_{k-1}, \hat{\tau}_k, \hat{\tau}_{k+1}}$.
\begin{align*}
    &\L(\widetilde{\Tcal}) - \L(\estcps) \\
    =& \L_{(\hat{\tau}_{k-2}, \tau_j^\ast]} + \L_{(\tau_j^\ast, \tau_{j+1}^\ast]} + \L_{(\tau_{j+1}^\ast, \hat{\tau}_{k+2}]} - \Bigl[\sum_{t = k-2}^{k+1} \L_{(\hat{\tau}_t, \hat{\tau}_{t+1}]} + \gamma \Bigr] \\
    =& [\L_{(\hat{\tau}_{k-2}, \tau_j^\ast]} - \L_{(\hat{\tau}_{k-2}, \hat{\tau}_{k-1}]} + \L_{(\tau_{j+1}^\ast, \hat{\tau}_{k+2}]} - \L_{(\hat{\tau}_{k+1}, \hat{\tau}_{k+2}]}] \\
    &+ \L_{(\tau_j^\ast, \tau_{j+1}^\ast]} - \Bigl[\sum_{t = k-1}^{k} \L_{(\hat{\tau}_t, \hat{\tau}_{t+1}]} + \gamma \Bigr]\\
    <& \Bigl[ \L_{(\hat{\tau}_{k-2}, \tau_j^\ast]} - \L_{(\hat{\tau}_{k-2}, \hat{\tau}_{k-1}]} + \L_{(\tau_{j+1}^\ast, \hat{\tau}_{k+2}]} - \L_{(\hat{\tau}_{k+1}, \hat{\tau}_{k+2}]} - \sum_{i \in (\hat{\tau}_{k-1}, \tau_j^\ast] \cup (\tau_{j+1}^\ast, \hat{\tau}_{k+1}]} \epsilon_i^2 \Bigr] \\
    &+ (C_{\ref*{lem:loc_err_g}.1} + 2 C_{\ref*{lem:loc_err_g}.2} - C_{\gamma}) s \log(p \vee n) \\
    \le& (C_{\ref*{lem:loc_err_g}.1} + 6 C_{\ref*{lem:loc_err_g}.2} + 2 C_{m,1} - C_{\gamma}) s \log(p \vee n) \le 0,
\end{align*}
provided $C_{\gamma} \ge C_{\ref*{lem:loc_err_g}.1} + 6 C_{\ref*{lem:loc_err_g}.2} + 2 C_{m,1}$. The second last inequality is from Eq. (\ref{equ:case2_low_diff}).

Combining the proof in the $\Hcal$ and $\Gcal$ parts, we can determine the two constants by solving the following inequalities,
\begin{equation}\label{equ:ieqs_to_solve}
    \left\{
    \begin{aligned}
        C_{\gamma} &\ge C_{\ref*{lem:loc_err_g}.1} + 6 C_{\ref*{lem:loc_err_g}.2} + 2 C_{m,1} \\
        \widetilde{C} &\ge 2 (1 - C_{\ref*{lem:loc_err_g}.3})^{-1}(2 C_{\ref*{lem:loc_err_g}.1} + C_{\gamma} + 4 C_{\ref*{lem:loc_err_g}.2} + 2 C_{m,1})
    \end{aligned}
    \right.
\end{equation}
Since $C_{\mathsf{snr}}$ and $C_{\mathsf{m}}$ are sufficiently large, we have $C_{m,1} = \widetilde{C}^2 ( \widetilde{C} \vee C_{\mathsf{m}}) / \{C_{\mathsf{snr}}(C_{\mathsf{snr}} - \widetilde{C} \vee C_{\mathsf{m}})\} = \widetilde{C}^2 C_{\mathsf{m}} / \{C_{\mathsf{snr}}(C_{\mathsf{snr}} - C_{\mathsf{m}})\}$. Let $C_{\gamma} = C_{\ref*{lem:loc_err_g}.1} + 6 C_{\ref*{lem:loc_err_g}.2} + 2 C_{m,1}$, we obtain the following inequality w.r.t. $\widetilde{C}$,
\begin{equation}\label{equ:ieq_to_solve_Ctilde}
    \frac{C_{\mathsf{m}} \widetilde{C}^2}{C_{\mathsf{snr}} (C_{\mathsf{snr}} - C_{\mathsf{m}})} - \frac{(1 - C_{\ref*{lem:loc_err_g}.3}) \widetilde{C}}{2} + 3 C_{\ref*{lem:loc_err_g}.1} + 10 C_{\ref*{lem:loc_err_g}.2} \ge 0.
\end{equation}
Treat it as a quadratic inequality w.r.t. $\widetilde{C}$, we can figure out that there exist solutions if and only if $C_{\mathsf{snr}} (C_{\mathsf{snr}} - C_{\mathsf{m}}) \ge 16 (1 - C_{\ref*{lem:loc_err_g}.3})^{-2} C_{\mathsf{m}} (3 C_{\ref*{lem:loc_err_g}.1} + 10 C_{\ref*{lem:loc_err_g}.2})$. And by solving Eq.~(\ref{equ:ieq_to_solve_Ctilde}), we have
\begin{equation}\label{equ:bound_Ctilde}
    \widetilde{C} = 2 \{a - (a^2 - b)^{\frac{1}{2}}\} \le \frac{b}{(a^2 - b)^{\frac{1}{2}}} \le \frac{2 b}{a} = 4 (1 - C_{\ref*{lem:loc_err_g}.3})^{-1} (3 C_{\ref*{lem:loc_err_g}.1} + 10 C_{\ref*{lem:loc_err_g}.2}),
\end{equation}
satisfies Eq. (\ref{equ:ieqs_to_solve}) with $a = (1 - C_{\ref*{lem:loc_err_g}.3}) C_{\mathsf{snr}} (C_{\mathsf{snr}} - C_{\mathsf{m}})/(8 C_{\mathsf{m}})$, $b = (3 C_{\ref*{lem:loc_err_g}.1} + 10 C_{\ref*{lem:loc_err_g}.2}) C_{\mathsf{snr}} (C_{\mathsf{snr}} - C_{\mathsf{m}}) / (4 C_{\mathsf{m}})$. The last inequality in Eq. (\ref{equ:bound_Ctilde}) holds provided that $C_{\mathsf{snr}}$ is sufficiently large so that $b \le 3 a^2 / 4$.

It follows that $\estcps = \widetilde{\cpsset}$ provided that Eq. (\ref{equ:ieqs_to_solve}) holds. Finally, we obtain
\begin{equation*}
    \widehat{K} = K^\ast; \max_{1 \le k \le K^\ast} \min_{1 \le j \le \widehat{K}} \Delta_k^2 |\tau_k^\ast - \hat{\tau}_j| \le \widetilde{C} s \log(p \vee n),
\end{equation*}
with any $\widetilde{C} \ge 4 (1 - C_{\ref*{lem:loc_err_g}.3})^{-1} (3 C_{\ref*{lem:loc_err_g}.1} + 10 C_{\ref*{lem:loc_err_g}.2})$.

\section{Proof of Theorem \ref{thm:localization_error}}\label{sec:proof_main_thm}

For a interval $I$, denote the sparsity constant $s_I = s \vee \size{\set{1 \le j \le p: \exists i \in I, \btheta_{i, j}^\circ \neq 0}}$. Observe that $s_I \le (1 \vee \size{\truecps \cap I}) \times s$. Define $\Delta_{I,q,\btheta} = ({\size{I}}^{-1}\sum_{i \in I} \norm{\btheta_i^\circ - \btheta}_{\Sigma}^q)^{{1}/{q}}$ and let $\Delta_{I, \btheta} = \Delta_{I,2,\btheta}$ be the root average square variation of $I$ and $\Delta_{I,\infty, \btheta} = \max_{i \in I} \norm{\btheta_i^\circ - \btheta}_{\Sigma}$ be the maximum variation of $I$.
For simplicity, denote $\Delta_{I, q} = \Delta_{I, q,\btheta_I^\circ}$, $\Delta_{I} = \Delta_{I, 2}$ and $\Delta_{I, \infty} = \Delta_{I, \infty, \btheta_I^{\circ}}$.

As stated in Lemma \ref{lem:loc_err_g}, to show that the bound of localization error in Theorem \ref{thm:localization_error} holds, we only need to certify that the event $\mathbb{G}$ holds with high probability for both the original full model-fitting approach and the Reliever approach with suitable constants. These two claims are shown in Corollary \ref{cor:in_err} and Corollary \ref{cor:mix_err}, respectively. Finally, the $L_2$ error bound of the parameter estimation follows the oracle inequality of lasso.

This section is organized as follows. In Section \ref{subsec:pre_required}, we introduce several useful non-asymptotic probability bounds, including the oracle inequality of lasso with heterogeneous data. In Section \ref{subsec:cert_full} and \ref{subsec:cert_reliever}, we show that $\mathbb{G}$ holds with high probability for the two approaches correspondingly. All the proofs are relegated to the last part.

\subsection{Deviation Bounds via the Bernstein's Inequality}\label{subsec:pre_required}

The deviation bounds in this subsection will rely on the following Bernstein's inequality.

\begin{lemma}[Bernstein's inequality]\label{lem:bernstein}
    Let $\set{X_i}_{i=1}^{n}$ be independent, mean-zero random variables with sub-exponential tails.
    For every $t > 0$, we have
    \begin{equation*}
        \Pbb\Bigset{\Bigabs{\sum_{i=1}^n X_i} > t} \le 2 \exp\Bigl[-c_b \Bigl(
            \frac{t^2}{\sum_{i=1}^n \norm{X_i}_{\Psi_1}^2} \wedge \frac{t}{\max_i \norm{X_i}_{\Psi_1}}\Bigr)\Bigr],
    \end{equation*}
    where $c_b > 0$ is an absolute constant.
    Choose
    $$t = \frac{C_{u}}{c_b} \biggl[\Bigl({\sum_{i \in [n]} \norm{X_i}_{\Psi_1}^2 A_{n, p, s}}\Bigr)^{\frac{1}{2}} \vee (\max_i \norm{X_i}_{\Psi_1} A_{n, p, s}) \biggr] $$
    with $C_{u} \ge c_b$, we have
    \begin{equation*}
        \Pbb\Bigset{\Bigabs{\sum_{i=1}^n X_i} > t} \le  2 \exp\{-C_{u} A_{n, p, s}\}.
    \end{equation*}
    Here $A_{n, p, s}$ is a diverging sequence. For instance, $A_{n, p, s} = \log(p \vee n)$ and $A_{n, p, s} = s \log(p \vee n)$.
\end{lemma}

\begin{lemma}[Uniform restricted eigenvalue condition]\label{lem:re}
    Assume that Condition~\ref{cond:dist}(a) holds.
    For any interval $I \subset (0,n]$, denote $\hat{\Sigma}_I = \size{I}^{-1} \sum_{i \in I} \xbf_i \xbf_i^\top$.
    Uniformly for all intervals $I \subset (0, n]$ such that $\size{I} \ge s_I \log(p \vee n)$, with probability at least $1 - \exp\{-C_{u,1} \log(p \vee n)\}$,
    \begin{equation*}
        \vbf^\top \hat{\Sigma}_{I} \vbf \ge \norm{\vbf}_{\Sigma}^2 - C_{u,2} C_x^2 \sigma_x^2 \Bigl\{\frac{s_I \log(p \vee n)}{\size{I}}\Bigr\}^{\frac{1}{2}} \Bigl(\norm{\vbf}_2^2 + \frac{1}{s_I}\norm{\vbf}_1^2\Bigr),\, \forall \vbf \in \Rbb^p,
    \end{equation*}
    and
    \begin{equation*}
        \vbf^\top \hat{\Sigma}_{I} \vbf \le \norm{\vbf}_{\Sigma}^2 + C_{u,2} C_x^2 \sigma_x^2 \Bigl\{\frac{s_I \log(p \vee n)}{\size{I}}\Bigr\}^{\frac{1}{2}} \Bigl(\norm{\vbf}_2^2 + \frac{1}{s_I}\norm{\vbf}_1^2\Bigr),\, \forall \vbf \in \Rbb^p,
    \end{equation*}
    where $C_{u,1}$ and $C_{u,2}$ are two universal constants.
    Furthermore, let $\size{I} \ge C_{\mathsf{re}} s_I \log(p \vee n)$ with a sufficiently large constant $C_{\mathsf{re}} \ge 1 \vee ({34 C_{u,2}C_x^2 \sigma_x^2}/{\kmin})^2$. For any support set $\Scal \in [p]$ with $\size{\Scal} \le s_I$ and $\vbf \in \Rbb^p$ such that $\norm{\vbf_{\Scal^\complement}}_1 \le 3 \norm{\vbf_{\Scal}}_1$, under the same event above,
    \begin{equation*}
        \frac{1}{2} \norm{\vbf}_{\Sigma}^2 \le (1 - C_{\ref*{lem:re}}) \norm{\vbf}_{\Sigma}^2 \le \vbf^\top \hat{\Sigma}_{I} \vbf \le (1 + C_{\ref*{lem:re}}) \norm{\vbf}_{\Sigma}^2 \le \frac{3}{2} \norm{\vbf}_{\Sigma}^2,
    \end{equation*}
    \begin{equation*}
        \frac{\kmin}{2} \norm{\vbf}_2^2 \le (1 - C_{\ref*{lem:re}}) \kmin \norm{\vbf}_{2}^2 \le \vbf^\top \hat{\Sigma}_{I} \vbf \le  (\sigma_x^2 + C_{\ref*{lem:re}} \kmin) \norm{\vbf}_{2}^2 \le \Bigl(\sigma_x^2 + \frac{\kmin}{2}\Bigr) \norm{\vbf}_2^2,
    \end{equation*}
    where $C_{\ref*{lem:re}} = 17 C_{u,2} C_x^2 \sigma_x^2 \kappa^{-1} C_{\mathsf{re}}^{-\frac{1}{2}}$.
\end{lemma}

\begin{proof}[Proof of Lemma \ref{lem:re}.]
    For sparsity level $s$, denote $\A(s) = \set{\vbf \in \Rbb^p: \norm{\vbf}_2 = 1, \size{\supp(\vbf)} \le s}$. We will first show that with high probability, $\sup_{\vbf \in \A(2s_I)} |\vbf^\top (\hat{\Sigma}_I - \Sigma) \vbf| = O(C_x^2 \{\size{I}^{-1} s_I \log(p \vee n)\}^{\frac{1}{2}})$ uniformly for all intervals $\{I\}$ such that $\sup_{\size{I} \ge s_I \log(p \vee n)}$. Then the result follows from Lemma 12 in \citet{loh_high-dimensional_2012}.

    For a fixed interval $I$, let $\Dbf = \hat{\Sigma}_I - \Sigma$.
    For any $\Ucal \subset [p]$ and $\size{\Ucal} = 2s_I$, let $\Dbf_{\Ucal} \in \Rbb^{2 s_I \times 2 s_I}$ be the sub-matrix of $\Dbf$ with $\Ucal$ being the set of row and column indices. Let $\Bcal_{\Ucal} = \set{\vbf \in \Rbb^p: \norm{\vbf}_2 = 1, \supp(\vbf) = \Ucal}$. There is a $4^{-1}$-net $\Ncal_{\Ucal} \subseteq \Bcal_{\Ucal}$ of $\Bcal_{\Ucal}$ with cardinality $\size{\Ncal_{\Ucal}} \le 9^{2s_I}$. For any $\vbf \in \Bcal_{\Ucal} - \Ncal_{\Ucal}$, there is $\ubf \in \Ncal_{\Ucal}$ such that $\norm{\vbf - \ubf}_2 \le 4^{-1}$ and $\norm{\vbf - \ubf}_2^{-1}(\vbf - \ubf) \in \Scal_{\Ucal}$. Therefore,
    \begin{equation*}
        \abs{\vbf^\top \Dbf \vbf - \ubf^\top \Dbf \ubf} = \abs{\vbf^\top \Dbf (\vbf - \ubf) + \ubf^\top \Dbf (\vbf - \ubf)} \le 2 \norm{\Dbf_{\Ucal}}_{\mathsf{op}} \norm{\vbf - \ubf}_2 \le \frac{1}{2} \norm{\Dbf_{\Ucal}}_{\mathsf{op}}.
    \end{equation*}
    By the definition of $\Dbf_{\Ucal}$, we have $\norm{\Dbf_{\Ucal}}_{\mathsf{op}} = \sup_{\vbf \in \Bcal_{\Ucal}} \abs{\vbf^\top \Dbf \vbf}$. Hence
    \begin{equation*}
        \sup_{\vbf \in \Bcal_{\Ucal}} \abs{\vbf^\top \Dbf \vbf} \le 2 \sup_{\vbf \in \Ncal_{\Ucal}} \abs{\vbf^\top \Dbf \vbf}.
    \end{equation*}
    Let $\Ncal = \cup_{\size{\Ucal} = 2s_I} \Ncal_{\Ucal}$. We have $\size{\Ncal} \le \binom{p}{2s_I} 9^{2s_I} \le (9p)^{2s_I}$ and $\Ncal$ is the $4^{-1}$-net of $\A(2s_I)$ because $\A(2s_I) = \cup_{\size{\Ucal}=2s_I} \Bcal_{\Ucal}$. Also,
    \begin{equation*}
        \sup_{\vbf \in \A(2s_I)} \abs{\vbf^\top \Dbf \vbf} \le 2 \sup_{\vbf \in \Ncal} \abs{\vbf^\top \Dbf \vbf}.
    \end{equation*}

    For a fixed $\vbf \in \A(2s_I)$, by the Bernstein's inequality (Lemma \ref{lem:bernstein}),
    \begin{equation*}
        \Pbb\biggl[ \abs{\vbf^\top \Dbf \vbf} > \frac{t}{\size{I}} \biggr] \le 2 \exp\biggl[-c_b \Bigl(\frac{t^2}{C_x^4 \sigma_x^4 \size{I}} \wedge \frac{t}{C_x^2 \sigma_x^2}\Bigr)\biggr].
    \end{equation*}
    Set $t = c_b^{-1} C_u C_x^2 \sigma_x^2 \{\size{I} s_I \log(p \vee n)\}^{\frac{1}{2}}$ with $C_u \ge c_b$ be a sufficiently large constant. With probability at least $1 - \exp\{- C_u s_I \log(p \vee n)\}$,
    \begin{equation*}
        \abs{\vbf^\top \Dbf \vbf} \le c_b^{-1} C_u C_x^2 \sigma_x^2 \Bigl\{\frac{s_I\log(p \vee n)}{\size{I}}\Bigr\}^{\frac{1}{2}}.
    \end{equation*}
    Note that $s_I \ge s$ by its definition. By taking the union bound over $\vbf \in \Ncal$ and $\set{I:\size{I} \ge s_I \log(p \vee n)}$, with probability at least $1 - n^2 (9p)^{2s} \exp\{-C_u s \log(p \vee n)\} \ge 1 - \exp\{- C_{u,1} \log(p \vee n)\}$ for some $C_{u,1} > 0$, uniformly for all $I$ such that $\size{I} \ge s_I \log(p \vee n)$,
    \begin{equation*}
        \sup_{\vbf \in \A(2s_I)} \abs{\vbf^\top (\hat{\Sigma}_I - \Sigma) \vbf} \le 2\sup_{\vbf \in \Ncal} \abs{\vbf^\top (\hat{\Sigma}_I - \Sigma) \vbf} \le 2 c_b^{-1} C_u C_x^2 \sigma_x^2 \Bigl\{\frac{s_I\log(p \vee n)}{\size{I}}\Bigr\}^{\frac{1}{2}}.
    \end{equation*}
    By Lemma 12 in \citet{loh_high-dimensional_2012}, under the above event,
    \begin{equation*}
        \abs{\vbf^\top (\hat{\Sigma}_I - \Sigma) \vbf} \le 54 c_b^{-1} C_u C_x^2 \sigma_x^2 \Bigl\{\frac{s_I\log(p \vee n)}{\size{I}}\Bigr\}^{\frac{1}{2}} (\norm{\vbf}_2^2 + \frac{1}{s_I} \norm{\vbf}_1^2),
    \end{equation*}
    for all $\vbf \in \Rbb^p$ and all intervals in $\set{I: \size{I} \ge s_I \log(p \vee n)}$. Let $C_{u,2}=54 c_b^{-1} C_u$. With probability at least $1 - \exp\{-C_{u,1} \log(p \vee n)\}$, for all $\vbf \in \Rbb^p$ and all $I$ such that $\size{I} \ge s_I \log(p \vee n)$,
    \begin{equation}\label{equ:rec_1}
        \vbf^\top \hat{\Sigma}_I \vbf \ge \norm{\vbf}_{\Sigma}^2 - C_{u,2} C_x^2 \sigma_x^2 \Bigl\{\frac{s_I\log(p \vee n)}{\size{I}}\Bigr\}^{\frac{1}{2}} \Bigl(\norm{\vbf}_2^2 + \frac{1}{s_I} \norm{\vbf}_1^2\Bigr),
    \end{equation}
    and
    \begin{equation}\label{equ:rec_1_upper}
        \vbf^\top \hat{\Sigma}_{I} \vbf \le \norm{\vbf}_{\Sigma}^2 + C_{u,2} C_x^2 \sigma_x^2 \Bigl\{\frac{s_I \log(p \vee n)}{\size{I}}\Bigr\}^{\frac{1}{2}} \Bigl(\norm{\vbf}_2^2 + \frac{1}{s_I}\norm{\vbf}_1^2\Bigr).
    \end{equation}
    If there exists a support set $\Scal \in [p]$ with $\size{\Scal} \le s_I$ so that $\norm{\vbf_{\Scal^\complement}}_1 \le 3 \norm{\vbf_{\Scal}}_1$, we have $\norm{\vbf}_1 \le 4 \norm{\vbf_{\Scal}}_1 \le 4 s_I^{\frac{1}{2}} \norm{\vbf_{\Scal}}_2 \le 4 s_I^{\frac{1}{2}} \norm{\vbf}_2$. By Eq. (\ref{equ:rec_1})--(\ref{equ:rec_1_upper}) and the inequality that $\frac{1}{s_I}\norm{\vbf}_1^2 \le 16 \norm{\vbf}_2^2$, we have
    \begin{equation*}
        \vbf^\top \hat{\Sigma}_I \vbf \ge \norm{\vbf}_\Sigma^2 - 17 C_{u,2} C_x^2 \sigma_x^2 \Bigl\{\frac{s_I\log(p \vee n)}{\size{I}}\Bigr\}^{\frac{1}{2}} \norm{\vbf}_2^2,
    \end{equation*}
    and
    \begin{equation*}
        \vbf^\top \hat{\Sigma}_I \vbf \le \norm{\vbf}_\Sigma^2 + 17 C_{u,2} C_x^2 \sigma_x^2 \Bigl\{\frac{s_I\log(p \vee n)}{\size{I}}\Bigr\}^{\frac{1}{2}} \norm{\vbf}_2^2.
    \end{equation*}
    Because $\size{I} \ge C_{\mathsf{re}} s_I \log(p \vee n)$ and $C_{\mathsf{re}} \ge 1 \vee (34 \kmin^{-1} C_{u,2}C_x^2 \sigma_x^2)^2$, we have
    $$17 C_{u,2} C_x^2 \sigma_x^2 \size{I}^{-\frac{1}{2}} \{s_I\log(p \vee n)\}^{\frac{1}{2}} \le \frac{\kmin}{2}.$$
    The last two results in the lemma are due to the inequality $\norm{\vbf}_2^2 \le \kmin^{-1} \norm{\vbf}_{\Sigma}^2$.
    \end{proof}

\begin{lemma}\label{lem:xxbeinf}
    Assume that Condition~\ref{cond:dist} holds.
    For interval $I \subset (0, n]$ and a fixed $\btheta$, with probability at least $1 - n^{-2}\exp\{-C_{u,1} \log(p \vee n)\}$,
    \begin{align*}
        & \Bignorm{\sum_{i \in I}\{ (\xbf_i \xbf_i^\top - \Sigma) (\btheta_i^\circ - \btheta) + \epsilon_i \xbf_i\}}_{\infty}\\
        \le & C_{u,2} C_x \sigma_x \biggl\{(C_x^2\Delta_{I,\btheta}^2 + C_{\epsilon}^2) \vee \frac{(C_x^2\Delta_{I, \infty,\btheta}^2 + C_{\epsilon}^2) \log(p \vee n)}{\size{I}}\biggr\}^{\frac{1}{2}} \{\size{I} \log(p \vee n)\}^{\frac{1}{2}}
    \end{align*}
    where $C_{u,1}$ and $C_{u,2}$ are two universal constants.
\end{lemma}

\begin{proof}[Proof of Lemma \ref{lem:xxbeinf}.]
    By condition~\ref{cond:dist}, $\Ebb[\sum_{i \in I}\{ (\xbf_i \xbf_i^\top - \Sigma) (\btheta_i^\circ - \btheta) + \epsilon_i \xbf_i\}] = \zero$. By Condition~\ref{cond:dist}, $\xbf_i^\top(\btheta_i^\circ - \btheta)$ is sub-Gaussian with mean zero and $\Psi_2$-norm $C_x \norm{\btheta_i^\circ - \btheta}_{\Sigma}$ and $\epsilon_i$ is sub-Gaussian with mean zero and $\Psi_2$-norm $\norm{\epsilon}_{\Psi_2} = C_{\epsilon}$. Hence $\xbf_i^\top(\btheta_i^\circ - \btheta) + \epsilon_i$ is sub-Gaussian with mean zero and
    \[\norm{\xbf_i^\top(\btheta_i^\circ - \btheta) + \epsilon_i}_{\Psi_2} \le (C_x^2 \norm{\btheta_i^\circ - \btheta}_{\Sigma}^2 + C_{\epsilon}^2)^{\frac{1}{2}}.\]
    Then $\xbf_i \{\xbf_i^\top (\btheta_i^\circ - \btheta) + \epsilon_i\} - \Sigma (\btheta_i^\circ - \btheta)$ is sub-exponential with $\Psi_1$-norm:
    \[\norm{\xbf_i \{\xbf_i^\top (\btheta_i^\circ - \btheta) + \epsilon_i\} - \Sigma (\btheta_i^\circ - \btheta)}_{\Psi_1} \le \norm{\xbf_i \{\xbf_i^\top (\btheta_i^\circ - \btheta) + \epsilon_i\}}_{\Psi_1} \le C_x \sigma_x (C_x^2 \norm{\btheta_i^\circ - \btheta}_{\Sigma}^2 + C_{\epsilon}^2)^{\frac{1}{2}}.\]
    By the Bernstein's inequality (Lemma \ref{lem:bernstein}), for any given $\vbf \in \Sbb^{p-1}$,
    \begin{align*}
        &\Pbb\Bigset{\Bigabs{\sum_{i \in I} \vbf^\top \{ (\xbf_i \xbf_i^\top - \Sigma) (\btheta_i^\circ - \btheta) + \epsilon_i \xbf_i\} } > t} \\
        \le & 2 \exp\biggl(-\frac{c_b t^2}{C_x^2 \sigma_x^2 (C_x^2\Delta_{I,\btheta}^2 + C_{\epsilon}^2)\size{I}} \wedge \frac{c_b t}{C_x \sigma_x (C_x^2\Delta_{I,\infty,\btheta}^2 + C_{\epsilon}^2)^{\frac{1}{2}}}\biggr).
    \end{align*}
    Set $t = c_b^{-1} (C_{u,1}+3) C_x \sigma_x [\{(C_x^2\Delta_{I, \btheta}^2 + C_{\epsilon}^2) \size{I} \log(p \vee n)\}^{\frac{1}{2}} \vee \{(C_x^2\Delta_{I, \infty, \btheta}^2 + C_{\epsilon}^2) \log^2(p \vee n)\}^{\frac{1}{2}}]$ with any constant $C_{u,1} \ge c_b$.
    With probability at least $1 - p \exp\{-(C_{u,1} + 3) \log(p \vee n)\} \ge 1 - n^{-2}\exp\{-C_{u,1} \log(p \vee n)\}$,
    \begin{align*}
        & \Bignorm{\sum_{i \in I}\{ (\xbf_i \xbf_i^\top - \Sigma) (\btheta_i^\circ - \btheta) + \epsilon_i \xbf_i\}}_{\infty}\\
        \le & C_{u,2} C_x \sigma_x \Bigl[\{(C_x^2\Delta_{I, \btheta}^2 + C_{\epsilon}^2) \size{I} \log(p \vee n)\}^{\frac{1}{2}} \vee \{(C_x^2\Delta_{I, \infty, \btheta}^2 + C_{\epsilon}^2) \log^2(p \vee n)\}^{\frac{1}{2}}\Bigr],
    \end{align*}
    where $C_{u,2} = c_b^{-1} (C_{u,1}+3)$.
    \end{proof}

\begin{lemma}[Oracle inequalities for the parametric estimates]\label{lem:oracle}
    Assume that Condition~\ref{cond:parameter}(a) and Condition~\ref{cond:dist} hold.
    For any interval $I \subset (0, n]$, recall that $D_{I} = [(C_x^2\Delta_I^2 + C_{\epsilon}^2) \vee \{\size{I}^{-1}(C_x^2\Delta_{I, \infty}^2 + C_{\epsilon}^2) \log(p \vee n)\}]^{\frac{1}{2}}$.
    We have with probability at least $1 - 2 \exp\{-C_{u,1}\log(p \vee n)\}$, uniformly for any interval $I \subset (0, n]$ with $\size{I} \ge C_{\mathsf{re}} s_I \log(p \vee n)$, provided that $\lambda_I = 4 C_{u,2} C_x \sigma_x D_{I} \{\size{I} \log(p \vee n)\}^{\frac{1}{2}}$, the solution $\htheta_I$ satisfies that
    \begin{equation*}
        \norm{\htheta_I - \btheta_I^\circ}_2 \le \kmin^{-\frac{1}{2}}\norm{\htheta_I - \btheta_I^\circ}_\Sigma \le C_{\ref*{lem:oracle}} D_{I} \Bigl\{\frac{s_I \log (p \vee n)}{\size{I}}\Bigr\}^{\frac{1}{2}},
    \end{equation*}
    \begin{equation*}
        \norm{\htheta_I - \btheta_I^\circ}_1 \le C_{\ref*{lem:oracle}} D_{I} s_I \Bigl\{\frac{ \log (p \vee n)}{\size{I}}\Bigr\}^{\frac{1}{2}},
    \end{equation*}
    where the model-based constant $C_{\ref*{lem:oracle}} = 12 \kmin^{-1} C_{u,2} C_x \sigma_x$. Furthermore, let $\Scal_I$ be the support set of $\btheta_I^\circ$, we have $\norm{\htheta_{I,\Scal_I^\complement} - \btheta_{I,\Scal_I^\complement}^\circ}_1 \le 3 \norm{\htheta_{I,\Scal_I} - \btheta_{I,\Scal_I}^\circ}_1$.
\end{lemma}

\begin{proof}[Proof of Lemma \ref{lem:oracle}.]
    (Oracle inequality for the mixture of distributions.)

    In the following proof, we assume that the inequalities in Lemmas \ref{lem:re}--\ref{lem:xxbeinf} hold for all intervals $I$ such that $\size{I} \ge C_{\mathsf{re}} s_I \log(p \vee n)$ and $\btheta = \btheta_I^\circ$. The claim holds with a probability lower bound $1 - 2 \exp\{-C_{u,1} \log(p \vee n)\}$.

    By the definition of $\htheta_I$,
    \begin{align*}
        &\sum_{i \in I} (y_i - \xbf_i^\top \htheta_I)^2 + \lambda_I \norm{\htheta_I}_1 = \sum_{i \in I} \{y_i - \xbf_i^\top \btheta_i^\circ + \xbf^\top(\btheta_i^\circ - \btheta_{I}^\circ) + \xbf_i^\top (\btheta_I^\circ - \htheta_I)\}^2 + \lambda_I \norm{\htheta_I}_1\\
        =& \sum_{i \in I} \epsilon_i^2 + \{\xbf_i^\top (\btheta_i^\circ - \btheta_I^\circ)\}^2 + \{\xbf_i^\top (\btheta_I^\circ - \htheta_I)\}^2 + \lambda_I \norm{\htheta_I}_1 \\
        +& 2 \sum_{i \in I} \{\epsilon_i \xbf_i^\top(\btheta_i^\circ - \btheta_I^\circ) + \epsilon_i \xbf_i^\top(\btheta_I^\circ - \htheta_I)\} + 2(\btheta_I^\circ - \htheta_I)^\top \sum_{i \in I} \xbf_i \xbf_i^\top (\btheta_i^\circ - \btheta_I^\circ)\\
        \le& \sum_{i \in I} (y_i - \xbf_i^\top\btheta_I^\circ)^2 + \lambda_I \norm{\btheta_I^\circ}_1 = \sum_{i \in I} \{y_i - \xbf_i^\top \btheta_i^\circ + \xbf^\top(\btheta_i^\circ - \btheta_{I}^\circ)\}^2 + \lambda_I \norm{\btheta_I^\circ}_1\\
        =& \sum_{i \in I} \epsilon_i^2 + \{\xbf_i^\top(\btheta_i^\circ - \btheta_I^\circ)\}^2 + 2 \sum_{i \in I}\epsilon_i \xbf_i^\top(\btheta_i^\circ - \btheta_I^\circ) + \lambda_I \norm{\btheta_I^\circ}_1.
    \end{align*}
    Hence,
    \begin{align*}
        &\sum_{i \in I} \{\xbf_i^\top(\htheta_I - \btheta_I^\circ)\}^2 + \lambda_I \norm{\htheta_I}_1\\
        \le& 2 (\htheta_I - \btheta_I^\circ)^\top \sum_{i \in I}\{ \epsilon_i \xbf_i + \xbf_i \xbf_i^\top (\btheta_i^\circ - \btheta_I^\circ)\} + \lambda_I \norm{\btheta_I^\circ}_1 \le \lambda_{I,1} \norm{\htheta_I - \btheta_I^\circ}_1 + \lambda_I \norm{\btheta_I^\circ}_1,
    \end{align*}
    where {$\lambda_{I,1} = 2\norm{\sum_{i \in I}\{ \epsilon_i \xbf_i + \xbf_i \xbf_i^\top (\btheta_i^\circ - \btheta_I^\circ)\}}_{\infty}$}.
    By Lemma \ref{lem:xxbeinf},
    \begin{equation*}
        \lambda_{I,1} \le 2 C_{u,2} C_x \sigma_x D_I \{\size{I} \log(p \vee n)\}^{\frac{1}{2}}.
    \end{equation*}
    where $D_I = [(C_x^2\Delta_I^2 + C_{\epsilon}^2) \vee \{ \size{I}^{-1}(C_x^2\Delta_{I, \infty}^2 + C_{\epsilon}^2) \log(p \vee n)\}]^{\frac{1}{2}}$ for easing the notation.

    Since $\sum_{i \in I}\{\xbf_i^\top(\htheta_I - \btheta_I^\circ)\}^2 \ge 0$, $(\lambda_I - \lambda_{I,1}) \norm{\htheta_{I,\Scal_I^\complement} - \btheta_{I,\Scal_I^\complement}^\circ}_1 \le (\lambda_I + \lambda_{I,1}) \norm{\htheta_{I,\Scal_I} - \btheta_{I,\Scal_I}^\circ}_1$. Choosing $\lambda_I = 2 \lambda_{I,1}$, we have $\norm{\htheta_{I,\Scal_I^\complement} - \btheta_{I,\Scal_I^\complement}^\circ}_1 \le 3 \norm{\htheta_{I,\Scal_I} - \btheta_{I,\Scal_I}^\circ}_1$.

    Apply Lemma \ref{lem:re}, the uniform restricted eigenvalue condition holds for any interval $I$ with $\size{I} \ge C_{\mathsf{re}} s_I \log(p \vee n)$.
    Hence $2^{-1} \size{I} \norm{\htheta_I - \btheta_I^\circ}_{\Sigma}^2 \le \sum_{i \in I} \{\xbf_i^\top (\htheta_I - \btheta_I^\circ)\}^2 \le \lambda_{I,1}\norm{\htheta_I - \btheta_I^\circ}_1 + \lambda_I \norm{\btheta_I^\circ}_1 - \lambda_I \norm{\htheta_I}_1 \le (\lambda_I + \lambda_{I,1}) \norm{\htheta_{I,\Scal_I} - \btheta_{I,\Scal_I}^\circ}_1 - \lambda_{I,1} \norm{\htheta_{I,\Scal_I^{\complement}}}_1 \le 3 \lambda_{I,1} s_I^{\frac{1}{2}}\norm{\htheta_I - \btheta_I^\circ}_2 \le 3 \lambda_{I,1} s_I^{\frac{1}{2}} \kmin^{-\frac{1}{2}}\norm{\htheta_I - \btheta_I^\circ}_\Sigma$. Therefore,
    \begin{equation*}
        \norm{\htheta_I - \btheta_I^\circ}_\Sigma \le \frac{3 \lambda_{I, 1} \kmin^{-\frac{1}{2}} s_I^\frac{1}{2} }{2^{-1} \size{I}} \le \frac{12 C_{u,2} C_x \sigma_x D_I}{\kmin^{\frac{1}{2}}} \Bigl\{\frac{s_I\log(p \vee n)}{\size{I}}\Bigr\}^{\frac{1}{2}}.
    \end{equation*}
    Similarly, we have
    \begin{equation*}
        \norm{\htheta_I - \btheta_I^\circ}_2 \le \frac{3\lambda_{I,1}s_I^{\frac{1}{2}}}{2^{-1} \kmin \size{I}} \le \frac{12 C_{u,2} C_x \sigma_x D_{I}}{\kmin} \Bigl\{\frac{s_I\log(p \vee n)}{\size{I}}\Bigr\}^{\frac{1}{2}},
    \end{equation*}
    and
    \begin{equation*}
        \norm{\htheta_I - \btheta_I^\circ}_1 \le \frac{3\lambda_{I,1} s_I}{2^{-1} \kmin \size{I}} \le \frac{12 C_{u,2} C_x \sigma_x D_{I} s_I}{\kmin} \Bigl\{\frac{\log(p \vee n)}{\size{I}}\Bigr\}^{\frac{1}{2}}.
    \end{equation*}
\end{proof}

\begin{lemma}\label{lem:xb2}
    Assume that Condition~\ref{cond:dist} holds.
    For interval $I \subset (0, n]$ and a fixed $\btheta$,
    with probability at least $1 - n^{-2}\exp\{-C_{u,1} \log(p \vee n)\}$, uniformly for any sub-interval $I \subset (0, n]$,
    \begin{equation*}
        \Bigabs{\sum_{i \in I} \{\xbf_i^\top(\btheta_i^\circ - \btheta)\}^2 - \Delta_{I, \btheta}^2 \size{I}} \le C_{u,2} C_x^2 \biggl\{\Delta_{I,4,\btheta}^4 \vee \frac{\Delta_{I,\infty,\btheta}^4\log(p \vee n)}{\size{I}}\biggr\}^{\frac{1}{2}} \{\size{I} \log(p \vee n)\}^{\frac{1}{2}},
    \end{equation*}
    where $C_{u,1}$ and $C_{u,2}$ are two universal constants.
\end{lemma}

\begin{lemma}\label{lem:xbe}
    Assume that Condition~\ref{cond:dist} holds.
    For interval $I \subset (0, n]$ and a fixed $\btheta$,
    with probability at least $1 - n^{-2}\exp\{-C_{u,1} \log(p \vee n)\}$,
    \begin{equation*}
        \Bigabs{\sum_{i \in I} \xbf_i^\top(\btheta_i^\circ - \btheta) \epsilon_i} \le C_{u,2} C_x C_{\epsilon} \biggl\{\Delta_{I, \btheta}^2 \vee \frac{\Delta_{I,\infty, \btheta}^2 \log(p \vee n)}{\size{I}}\biggr\}^{\frac{1}{2}} \{\size{I} \log(p \vee n)\}^{\frac{1}{2}},
    \end{equation*}
    where $C_{u,1}$ and $C_{u,2}$ are two universal constants.
\end{lemma}

\begin{proof}[Proof of Lemma \ref{lem:xb2}--\ref{lem:xbe}.]
    They both follow from Bernstein's inequality with similar arguments as in the proof of Lemma \ref{lem:xxbeinf}. \end{proof}

\subsection{Certifying $\mathbb{G}$ for the Full Model-fitting}\label{subsec:cert_full}

\begin{lemma}[In-sample error]\label{lem:in_err}
    Assume Condition~\ref{cond:parameter}\,(a) and Condition~\ref{cond:dist} hold.
    Under the setting in Lemma \ref{lem:oracle},
    with probability at least $1 - 4 \exp\{-C_{u,1} \log(p \vee n)\}$, for any interval $I = (\tau_l, \tau_r]$ such that $\size{I} \ge C_{\mathsf{re}} s_I \log(p \vee n)$,
    \begin{align*}
        & \Bigabs{\L_I - \sum_{i \in I} \epsilon_i^2 - \Delta_I^2 \size{I}} \le \frac{48 C_{u,2}^2 C_x^2 \sigma_x^2 D_{I}^2 s_I \log(p \vee n)}{\kmin}\nonumber\\
        + & C_{u,2} C_x (C_x \Delta_{I, \infty} + 2 C_{\epsilon}) [(\Delta_{I}^2\size{I}) \vee \{\Delta_{I, \infty}^2\log(p \vee n)\}]^{\frac{1}{2}} \{\log(p \vee n)\}^{\frac{1}{2}}.
    \end{align*}
    Additionally if Condition~\ref{cond:parameter}\,(b) holds,
    \begin{align*}
        & \Bigabs{\L_I - \sum_{i \in I} \epsilon_i^2 - \Delta_I^2 \size{I}} \le \frac{48 C_{u,2}^2 C_x^2 \sigma_x^2 D_{I}^2 s_I \log(p \vee n)}{\kmin}\nonumber\\
        + & C_{u,2} C_x (C_x \sigma_x C_{\theta} + 2 s_I^{-\frac{1}{2}} C_{\epsilon}) [(\Delta_{I}^2\size{I}) \vee \{\sigma_x^2 C_{\theta}^2 s_I\log(p \vee n)\}]^{\frac{1}{2}} \{s_I \log(p \vee n)\}^{\frac{1}{2}}.
    \end{align*}
\end{lemma}

\begin{proof}[Proof of Lemma \ref{lem:in_err}.]
    Assume that the inequalities in Lemmas \ref{lem:re}--\ref{lem:xbe} hold for all interval $I$ such that $\size{I} \ge C_{\mathsf{re}} s_I \log(p \vee n)$ and $\btheta = \btheta_I^\circ$. It holds with a probability lower bound $1 - 4 \exp\{-C_{u,1} \log(p \vee n)\}$.

    For any interval $I = (c,d]$, we will analyze the cost $\L_I$. By the definition of the cost $\L_I$,
    \begin{align*}
        \L_I =& \sum_{i \in I} (y_i - \xbf_i^\top \htheta_I)^2 = \sum_{i \in I} \{y_i - \xbf_i^\top \btheta_I^\circ + \xbf_i^\top(\btheta_I^\circ - \htheta_I)\}^2 \\
        =& \sum_{i \in I} \{(y_i - \xbf_i^\top \btheta_I^\circ)^2 + \{\xbf_i^\top(\btheta_I^\circ - \htheta_I)\}^2\} + 2 \sum_{i \in I} \{\xbf_i \epsilon_i + \xbf_i \xbf_i^\top (\btheta_i^\circ - \btheta_I^\circ)\}^\top (\btheta_I^\circ - \htheta_I) \\
        \ge & \sum_{i \in I} (y_i - \xbf_i^\top \btheta_I^\circ)^2 - \lambda_{I,1} \norm{\btheta_I^\circ - \htheta_I}_1 \ge \sum_{i \in I} (y_i - \xbf_i^\top \btheta_I^\circ)^2 - \lambda_I \norm{\btheta_I^\circ - \htheta_I}_1,
    \end{align*}
    where the second last inequality follows from Lemma \ref{lem:xxbeinf} and the last one is from $\lambda_I = 2 \lambda_{I,1} > 0$. By the definition of $\htheta_I$,
    \begin{equation*}
        \L_I - \sum_{i \in I}(y_i - \xbf_i^\top \btheta_I^\circ)^2 \le \lambda_I (\norm{\btheta_I^\circ}_1 - \norm{\htheta_I}_1) \le \lambda_I \norm{\btheta_I^\circ - \htheta_I}_1.
    \end{equation*}
    By combining the result in Lemma \ref{lem:oracle},
    \begin{equation}\label{equ:in_err_ieq1}
        \Bigabs{\L_I - \sum_{i \in I} (y_i - \xbf_i^\top \btheta_I^\circ)^2} \le \lambda_I \norm{\btheta_I^\circ - \htheta_I}_1 \le \frac{12 \lambda_{I,1}^2 s_I}{\kmin\size{I}} \le \frac{48 C_{u,2}^2 C_x^2 \sigma_x^2 D_I^2 s_I \log(p \vee n)}{\kmin}.
    \end{equation}
    By Lemma \ref{lem:xb2} and Lemma \ref{lem:xbe},

    \begin{equation}\label{equ:in_err_ieq2}
    \begin{aligned}
        &\Bigabs{\sum_{i \in I} [(y_i - \xbf_i^\top \btheta_I^\circ)^2 - \epsilon_i^2] - \Delta_I^2 \size{I}} \\
        = & \Bigabs{\sum_{i \in I}\{\xbf_i^\top(\btheta_i^\circ - \btheta_I^\circ)\}^2 - \Delta_I^2 \size{I} + \sum_{i \in I} 2 \epsilon_i \xbf_i^\top (\btheta_i^\circ - \btheta_I^\circ) } \\
        \le & C_{u,2} C_x^2 \Bigl\{\Delta_{I,4}^4 \vee \frac{\Delta_{I,\infty
        }^4 \log(p \vee n)}{\size{I}}\Bigr\}^{\frac{1}{2}} \{\size{I} \log(p \vee n)\}^{\frac{1}{2}} \\
        +& 2 C_{u,2} C_x C_{\epsilon} \Bigl\{\Delta_I^2 \vee \frac{\Delta_{I,\infty
        }^2 \log(p \vee n)}{\size{I}}\Bigr\}^{\frac{1}{2}} \{\size{I} \log(p \vee n)\}^{\frac{1}{2}}.
    \end{aligned}
    \end{equation}

    From Eq. (\ref{equ:in_err_ieq1}) and Eq. (\ref{equ:in_err_ieq2}),
    \begin{align*}
        & \Bigabs{\L_I - \sum_{i \in I} \epsilon_i^2 - \Delta_I^2 \size{I}} \le \frac{48 C_{u,2}^2 C_x^2 \sigma_x^2 D_I^2 s_I \log(p \vee n)}{\kmin}\nonumber\\
        + & C_{u,2} C_x^2 \Bigl\{\Delta_{I,4}^4 \vee \frac{\Delta_{I,\infty
        }^4 \log(p \vee n)}{\size{I}}\Bigr\}^{\frac{1}{2}} \{\size{I} \log(p \vee n)\}^{\frac{1}{2}} \nonumber\\
        +& 2 C_{u,2} C_x C_{\epsilon} \Bigl\{\Delta_I^2 \vee \frac{\Delta_{I,\infty
        }^2 \log(p \vee n)}{\size{I}}\Bigr\}^{\frac{1}{2}} \{\size{I} \log(p \vee n)\}^{\frac{1}{2}}.
    \end{align*}

    By Condition~\ref{cond:parameter} (b), $\Delta_{I, \infty} \le C_{\theta} \sigma_x s_I^{\frac{1}{2}}$. Hence we have
    \[
        \Bigl\{\Delta_I^2 \vee \frac{\Delta_{I,\infty
        }^2 \log(p \vee n)}{\size{I}}\Bigr\}^{\frac{1}{2}} \{\size{I} \log(p \vee n)\}^{\frac{1}{2}} \le s_I^{-\frac{1}{2}} \Bigl[\{C_{\theta} \sigma_x s_I \log(p \vee n)\} \vee \{\Delta_I^2 \size{I} s_I \log(p \vee n)\}^{\frac{1}{2}}\Bigr].
    \]
    Since $\Delta_{I,4}^2 \le \Delta_{I,\infty} \Delta_{I,2}$, it also holds that
    \begin{align*}
        & \Bigl\{\Delta_{I,4}^4 \vee \frac{\Delta_{I,\infty
        }^4 \log(p \vee n)}{\size{I}}\Bigr\}^{\frac{1}{2}} \{\size{I} \log(p \vee n)\}^{\frac{1}{2}} \\
        \le & C_{\theta} \sigma_x \Bigl[\{C_{\theta} \sigma_x s_I \log(p \vee n)\} \vee \{\Delta_I^2 \size{I} s_I \log(p \vee n)\}^{\frac{1}{2}}\Bigr].
    \end{align*}
    Finally, we have
    \begin{align*}
        & \Bigabs{\L_I - \sum_{i \in I} \epsilon_i^2 - \Delta_I^2 \size{I}} \le \frac{48 C_{u,2}^2 C_x^2 \sigma_x^2 D_{I}^2 s_I \log(p \vee n)}{\kmin}\nonumber\\
        + & C_{u,2} C_x (C_x \sigma_x C_{\theta} + 2 C_{\epsilon}s_I^{-\frac{1}{2}}) [(\Delta_{I}^2\size{I}) \vee \{\sigma_x^2 C_{\theta}^2 s_I\log(p \vee n)\}]^{\frac{1}{2}} \{s_I\log(p \vee n)\}^{\frac{1}{2}}.
    \end{align*}
    \end{proof}

\begin{corollary}\label{cor:in_err}
    Assume Condition~\ref{cond:change}, Condition~\ref{cond:parameter}, and Condition~\ref{cond:dist} hold.
    Under the same probability event in Lemma \ref{lem:in_err} and with sufficiently large $C_{\mathsf{m}} \ge C_{\mathsf{re}}$, we have the following conclusions.
    \begin{enumerate}[label=(\alph*)]
        \item For $I$ such that $\Delta_I = 0$ and $\size{I} \ge C_{\mathsf{m}} s \log(p \vee n)$,
        \begin{equation*}
            \Bigabs{\L_I - \sum_{i \in I} \epsilon_i^2} \le \frac{48 C_{u,2}^2 C_x^2 \sigma_{x}^2 C_{\epsilon
            }^2 s \log(p \vee n)}{\kmin} \triangleq C_{\ref*{cor:in_err}.1} s \log(p \vee n).
        \end{equation*}
        \item For $I = (a, b]$ such that $I \cap \truecps = \{\tau_k^\ast\}$, $\min(\tau_k^\ast - a, b - \tau_k^\ast) \le \widetilde{C} \Delta_k^{-2} s \log(p \vee n)$ and $\size{I} \ge C_{\mathsf{m}} s \log(p \vee n)$,
        \begin{equation*}
            \Bigabs{\L_I - \sum_{i \in I} \epsilon_i^2 - \Delta_I^2 \size{I}} \le C_{\ref*{cor:in_err}.2} s \log(p \vee n),
        \end{equation*}
        where $C_{\ref*{cor:in_err}.2} = 2 C_{\ref*{cor:in_err}.1} + (C_{\mathsf{m}} \kmin)^{-1}{48 C_{u,2}^2 C_x^4 \sigma_x^2 \widetilde{C}} + C_{u,2} C_x (C_x \sigma_x C_{\theta} + 2  C_{\epsilon}) \{(2C_{\theta}\sigma_x) \vee (\widetilde{C}^{\frac{1}{2}})\}$.
        \item For $I$ such that $\size{I} \ge C_{\mathsf{m}} s \log(p \vee n)$ and $\Delta_I^2 \size{I} \ge 2^{-1} \widetilde{C} s \log(p \vee n)$ for some sufficiently large $\widetilde{C} \ge 6 C_{\theta}^2 \sigma_x^2$,
        \begin{equation*}
            \L_I - \sum_{i \in I} \epsilon_i^2 \ge (1 - C_{\ref*{cor:in_err}.3})\Delta_I^2 \size{I},
        \end{equation*}
        where $
            C_{\ref*{cor:in_err}.3} = (C_{\mathsf{m}} \kmin)^{-1}(96 C_{u,2}^2 C_x^4 \sigma_x^2) + 6^{\frac{1}{2}} C_{u,2} C_x \Bigl(C_x \sigma_x C_{\theta} + 2 C_{\epsilon}\Bigr) \widetilde{C}^{-\frac{1}{2}} + 6 \widetilde{C}^{-1} C_{\ref*{cor:in_err}.1}
        $.
    \end{enumerate}
\end{corollary}

\begin{proof}[Proof of Corollary \ref{cor:in_err}.]
    All of the results in the three parts of Corollary~\ref{cor:in_err} follow from the proof of Lemma \ref{lem:in_err} and the conditions about the change signals, that is, the variations $\Delta_{I}^2 \size{I}$.

    \noindent \textit{Part\,(a)}. If $\Delta_I = 0$, we have $D_{I}^2 = C_{\epsilon}^2$ and $\Delta_{I, \infty} = 0$.
        Lemma \ref{lem:in_err} reduces to
        \begin{equation*}
            \Bigabs{\L_I - \sum_{i \in I} \epsilon_i^2 - \Delta_I^2 \size{I}} \le \frac{48 C_{u,2}^2 C_x^2 \sigma_x^2 C_{\epsilon}^2}{\kmin} s \log(p \vee n) = C_{\ref*{cor:in_err}.1} s \log(p \vee n).
        \end{equation*}

    \noindent \textit{Part\,(b)}. Since $\size{I \cap \Tcal^\ast} = 1$, we have $s_I \le 2 s$ and Lemma \ref{lem:in_err} holds for $\size{I} \ge C_{\mathsf{m}} s \log(p \vee n)\ge 2^{-1} C_{\mathsf{m}} s_I \log(p \vee n) \ge C_{\mathsf{re}} s_I \log(p \vee n)$ with sufficiently large $C_{\mathsf{m}} \ge 2 C_{\mathsf{re}}$. By $\size{I} \ge C_{\mathsf{m}} s \log(p \vee n)$,
        we have $\size{I}^{-1}(C_x^2 \Delta_{I, \infty}^2 + C_{\epsilon}^2) \log(p \vee n) \le 2 (C_{\mathsf{m}}s)^{-1} C_x^2 \sigma_x^2 C_{\theta}^2 s + C_{\epsilon}^2 \le C_{\epsilon}^2$ provided that $C_{\mathsf{m}}$ is sufficiently large.
        Hence $D_{I}^2 = C_x^2 \Delta_I^2 + C_{\epsilon}^2$.
        Combining $\Delta_I^2 \size{I} \le 2^{-1} \widetilde{C} s \log(p \vee n)$ and $\size{I} \ge C_{\mathsf{m}} s \log(p \vee n)$, $\Delta_I^2 \le 2^{-1} C_{\mathsf{m}}^{-1} \widetilde{C}$.
        Therefore by Lemma \ref{lem:in_err},
        \begin{equation*}
            \Bigabs{\L_I - \sum_{i \in I} \epsilon_i^2 - \Delta_I^2 \size{I}} \le C_{\ref*{cor:in_err}.2} s \log(p \vee n),
        \end{equation*}
        where $$C_{\ref*{cor:in_err}.2} = \frac{96 C_{u,2}^2 C_x^2 \sigma_x^2 C_{\epsilon}^2}{\kmin} + \frac{48 C_{u,2}^2 C_x^4 \sigma_x^2 \widetilde{C}}{C_{\mathsf{m}} \kmin} + C_{u,2} C_x (C_x \sigma_x C_{\theta} + 2 C_{\epsilon}) \{(2C_{\theta}\sigma_x) \vee (\widetilde{C}^{\frac{1}{2}})\}.$$

    \noindent \textit{Part\,(c)}. When $\size{I \cap \Tcal^\ast} \le 1$, the discussion in (b) is still valid. We have $\size{I} \ge C_{\mathsf{m}} s \log(p \vee n) \ge 2^{-1} C_{\mathsf{m}} s_I \log(p \vee n)$ and $\Delta_I^2 \size{I} \ge 4^{-1} \widetilde{C} s_I \log(p \vee n)$. Otherwise when $\size{I \cap \Tcal^\ast} \ge 2$, by Condition~\ref{cond:change}, we can obtain $\size{I} \ge 3^{-1} C_{\mathsf{snr}} s_I \log(p \vee n) \ge 2^{-1} C_{\mathsf{m}} s_I \log(p \vee n)$ and $\Delta_I^2 \size{I} \ge 6^{-1}\widetilde{C} s_I \log(p \vee n)$ since $C_{\mathsf{snr}}$ is sufficiently large.
        The above claim becomes trivial when $\size{I \cap \truecps} \ge 3$. When $\size{I \cap \truecps} = 2$, the result follows from the fact that $s \ge 3^{-1} s_I$.
        Recall that $\widetilde{C} \ge 6 C_{\theta}^2 \sigma_x^2$.
        We have $\Delta_I^2 \ge (6 \size{I})^{-1} \widetilde{C} s_I \log(p \vee n) \ge \size{I}^{-1} C_{\theta}^2 \sigma_x^2 s_I \log(p \vee n) \ge \size{I}^{-1} \Delta_{I,\infty}^2 \log(p \vee n)$.
        It implies that $D_{I}^2 = C_x^2 \Delta_{I}^2 + C_{\epsilon}^2$.
        By Lemma \ref{lem:in_err},
        \begin{align}
            & \L_I - \sum_{i \in I} \epsilon_i^2 \ge \Delta_{I}^2 \size{I} - \frac{48 C_{u,2}^2 C_x^2 \sigma_x^2 (C_x^2 \Delta_I^2 + C_{\epsilon}^2) s_I \log(p \vee n)}{\kmin} \nonumber\\
            -& C_{u,2} C_x (C_x \sigma_x C_{\theta} + 2 s_I^{-\frac{1}{2}} C_{\epsilon}) \{\Delta_{I}^2 \size{I} s_I \log(p \vee n)\}^{\frac{1}{2}} \ge (1 - C_{\ref*{cor:in_err}.3})\Delta_I^2 \size{I},
        \end{align}
        where
        \begin{equation*}
            C_{\ref*{cor:in_err}.3} = \frac{96 C_{u,2}^2 C_x^4 \sigma_x^2}{\kmin C_{\mathsf{m}}} + \frac{6 C_{\ref*{cor:in_err}.1}}{\widetilde{C}} + \frac{6^{\frac{1}{2}} C_{u,2} C_x \bigl(C_x \sigma_x C_{\theta} + 2 C_{\epsilon} \bigr)}{\widetilde{C}^{\frac{1}{2}}}.
        \end{equation*}
    \end{proof}

\subsection{Certifying $\mathbb{G}$ for Reliever}\label{subsec:cert_reliever}

To proceed, we first introduce some notations used in the theoretical justification for Reliever.
Let $R$ be the surrogate interval w.r.t. $I$ and $J = I \setminus R$ be the complement. Denote oracle change variation for Reliever by $\overline{\Delta}_{I}^2 = \Delta_{I, \btheta_R^\circ}^2$.
To distinguish the losses for Reliever from these for full model-fitting, we denote the losses of interval $I$ under the relief model estimate $\hat{\btheta}_R$, named relief losses, by $\widetilde{\L}_{I} = \sum_{i \in I} (y_i - \xbf_i^\top \htheta_{R})^2$.

\begin{lemma}[Relief error]\label{lem:rf_err}
    Assume Condition~\ref{cond:parameter} and Condition~\ref{cond:dist} hold.
    Under the setting in Lemma \ref{lem:oracle},
        with probability at least $1 - 4 \exp\{-C_{u,1} \log(p \vee n)\}$, for any interval $I = (\tau_l, \tau_r]$ such that $\size{I} \ge r^{-1} C_{\mathsf{re}} s_I \log(p \vee n)$ so that its relief interval $R$ satisfies $\size{R} \ge C_{\mathsf{re}} s_I \log(p \vee n)$, the relief loss $\widetilde{\L}_I$ satisfies that:
\begin{align*}
    & \bigabs{\widetilde{\L}_I - \sum_{i\in I} \epsilon_i^2 - \Delta_{I,\btheta_R^\circ}^2 \size{I}} \\
    \le & (1 + C_{\ref*{lem:re}}) r^{-1} \kmin C_{\ref*{lem:oracle}}^2 D_{R}^2 s_R \log (p \vee n) \\
    &+ 6^{-1} r^{-\frac{1}{2}} \kmin C_{\ref*{lem:oracle}}^2 D_{I, \btheta_{R}^\circ} D_{R} s_R \log(p \vee n) \\
    &+ 2 r^{-\frac{1}{2}} (1-r)^{\frac{1}{2}} \kmin^{\frac{1}{2}} C_{\ref*{lem:oracle}} D_{R} \{\Delta_{J,\btheta_R^\circ}^2 \size{J}\}^{\frac{1}{2}} \{s_R \log(p \vee n)\}^{\frac{1}{2}}.\\
    &+ C_{u,2} C_x (C_x \sigma_x C_{\theta} + 2 C_{\epsilon}s_I^{-\frac{1}{2}}) \Bigl[\{C_{\theta} \sigma_x s_I \log(p \vee n)\} \vee \{\Delta_{I,\btheta_R^\circ}^2 \size{I} s_I \log(p \vee n)\}^{\frac{1}{2}}\Bigr].
\end{align*}
\end{lemma}

\begin{proof}[Proof of Lemma~\ref{lem:rf_err}]
Similarly, we begin with the control of the difference $\widetilde{\L}_I - \widetilde{\L}_I^\circ$.
Observe that,
\begin{align*}
    &\widetilde{\L}_I - \widetilde{\L}_I^\circ \\
    =& \size{I} \norm{\btheta_R^\circ - \htheta_R}_{\widehat{\Sigma}_I}^2 + 2 \sum_{i \in I} (y_i - \xbf_i^\top \btheta_R^\circ) \xbf_i^\top (\btheta_R^\circ - \htheta_R) \\
    =& \size{I} \norm{\btheta_R^\circ - \htheta_R}_{\widehat{\Sigma}_I}^2 + 2 \sum_{i \in I} \{ (\xbf_i \xbf_i^\top - \Sigma) (\btheta_i^\circ - \btheta_R^\circ) + \epsilon_i \xbf_i\}^\top (\btheta_R^\circ - \htheta_R) \\
    &+ 2 \sum_{i \in I \setminus R} (\btheta_i^\circ - \btheta_R^\circ)^\top \Sigma (\btheta_R^\circ - \htheta_R).
\end{align*}
We will derive the upper bounds for the absolute values of the three terms in the above decomposition.

For the first term, assuming the probability events in Lemma~\ref{lem:re} and Lemma~\ref{lem:oracle} hold, we have:
\begin{equation}\label{equ:rf_part1}
    \size{I} \norm{\btheta_R^\circ - \htheta_R}_{\widehat{\Sigma}_I}^2 \le (1 + C_{\ref*{lem:re}}) r^{-1} \kmin C_{\ref*{lem:oracle}}^2 D_{R}^2 s_R \log (p \vee n).
\end{equation}

For the second term, assuming the probability events in Lemma~\ref{lem:xxbeinf} and Lemma~\ref{lem:oracle} hold, we have:
\begin{align*}
    & \Bigabs{\sum_{i \in I} \{ (\xbf_i \xbf_i^\top - \Sigma) (\btheta_i^\circ - \btheta_R^\circ) + \epsilon_i \xbf_i\}^\top (\btheta_R^\circ - \htheta_R)} \\
    \le & \Bignorm{\sum_{i \in I} \{ (\xbf_i \xbf_i^\top - \Sigma) (\btheta_i^\circ - \btheta_R^\circ) + \epsilon_i \xbf_i\}}_\infty \bignorm{\btheta_R^\circ - \htheta_R}_1 \\
    \le & C_{u,2} C_x \sigma_x D_{I, \btheta_{R}^\circ} \{\size{I} \log(p \vee n)\}^{\frac{1}{2}} \times C_{\ref*{lem:oracle}} D_{R} s_R \Bigl\{\frac{ \log (p \vee n)}{\size{I} r}\Bigr\}^{\frac{1}{2}} \\
    = & 12^{-1} r^{-\frac{1}{2}} \kmin C_{\ref*{lem:oracle}}^2 D_{I, \btheta_{R}^\circ} D_{R} s_R \log(p \vee n).
\end{align*}
where $D_{I, \btheta} = [(C_x^2\Delta_{I,\btheta}^2 + C_{\epsilon}^2) \vee \{ \size{I}^{-1}(C_x^2\Delta_{I, \infty, \btheta}^2 + C_{\epsilon}^2) \log(p \vee n)\}]^{\frac{1}{2}}$.

For the third term, assuming the probability event in Lemma~\ref{lem:oracle} holds, we have:
\begin{equation}\label{equ:diff_Delta_cross_1}
    \begin{aligned}
    & \Bigabs{\sum_{i \in J} (\btheta_i^\circ - \btheta_{R}^\circ)^\top \Sigma (\btheta_{R}^\circ - \htheta_{R})} \\
    \le & \sum_{i \in J} \norm{\btheta_i^\circ - \btheta_{R}^\circ}_{\Sigma} \norm{\btheta_{R}^\circ - \htheta_{R}}_{\Sigma} \\
    \le & \Delta_{J,\btheta_R^\circ} \size{J} \times \kmin^{\frac{1}{2}} C_{\ref*{lem:oracle}} D_{R} \Bigl\{\frac{s_R \log(p \vee n)}{\size{R}}\Bigr\}^{\frac{1}{2}} \\
    \le & \kmin^{\frac{1}{2}} C_{\ref*{lem:oracle}} D_{R} \{\Delta_{J,\btheta_R^\circ}^2 \size{J}\}^{\frac{1}{2}} \Bigl\{\frac{(1-r) s_R \log(p \vee n)}{r}\Bigr\}^{\frac{1}{2}}.
\end{aligned}
\end{equation}

Therefore,
\begin{equation}\label{equ:rf_err_ieq1}
    \begin{aligned}
    \abs{\widetilde{\L}_I - \widetilde{\L}_I^\circ} \le & (1 + C_{\ref*{lem:re}}) r^{-1} \kmin C_{\ref*{lem:oracle}}^2 D_{R}^2 s_R \log (p \vee n) \\
    &+ 6^{-1} r^{-\frac{1}{2}} \kmin C_{\ref*{lem:oracle}}^2 D_{I, \btheta_{R}^\circ} D_{R} s_R \log(p \vee n) \\
    &+ 2 r^{-\frac{1}{2}} (1-r)^{\frac{1}{2}} \kmin^{\frac{1}{2}} C_{\ref*{lem:oracle}} D_{R} \{\Delta_{J,\btheta_R^\circ}^2 \size{J}\}^{\frac{1}{2}} \{s_R \log(p \vee n)\}^{\frac{1}{2}}.
\end{aligned}
\end{equation}

We now turn to deriving the upper bound of $\abs{\widetilde{\L}_I^\circ - \sum_{i\in I} \epsilon_i^2 - \Delta_{I,\btheta_R^\circ}^2 \size{I}}$.
By Lemma \ref{lem:xb2} and Lemma \ref{lem:xbe},
\begin{align*}
    &\bigabs{\widetilde{\L}_I^\circ - \sum_{i\in I} \epsilon_i^2 - \Delta_{I,\btheta_R^\circ}^2 \size{I}}\\
    = & \Bigabs{\Bigl[\sum_{i \in I}\{\xbf_i^\top(\btheta_i^\circ - \btheta_R^\circ)\}^2\Bigr] - \Delta_{I, \btheta_R^\circ}^2 \size{I} + \sum_{i \in I} 2 \epsilon_i \xbf_i^\top (\btheta_i^\circ - \btheta_I^\circ) }\\
    \le & C_{u,2} C_x^2 \Bigl\{\Delta_{I,4, \btheta_R^\circ}^4 \vee \frac{\Delta_{I,\infty, \btheta_R^\circ}^4 \log(p \vee n)}{\size{I}}\Bigr\}^{\frac{1}{2}} \{\size{I} \log(p \vee n)\}^{\frac{1}{2}}\\
    +& 2 C_{u,2} C_x C_{\epsilon} \Bigl\{\Delta_{I, \btheta_R^\circ}^2 \vee \frac{\Delta_{I,\infty, \btheta_R^\circ}^2 \log(p \vee n)}{\size{I}}\Bigr\}^{\frac{1}{2}} \{\size{I} \log(p \vee n)\}^{\frac{1}{2}}.
\end{align*}
By Condition~\ref{cond:parameter} (b), $\Delta_{I, \infty,\btheta_R^\circ} \le C_{\theta} \sigma_x s_I^{\frac{1}{2}}$. Hence we have
\begin{align*}
    & \Bigl\{\Delta_{I, \btheta_R^\circ}^2 \vee \frac{\Delta_{I,\infty, \btheta_R^\circ}^2 \log(p \vee n)}{\size{I}}\Bigr\}^{\frac{1}{2}} \{\size{I} \log(p \vee n)\}^{\frac{1}{2}} \\
    \le & s_I^{-\frac{1}{2}} \Bigl[\{C_{\theta} \sigma_x s_I \log(p \vee n)\} \vee \{\Delta_{I, \btheta_R^\circ}^2 \size{I} s_I \log(p \vee n)\}^{\frac{1}{2}}\Bigr].
\end{align*}
Since $\Delta_{I,4,\btheta_R^\circ}^2 \le \Delta_{I,\infty,\btheta_R^\circ} \Delta_{I,\btheta_{R}^\circ}$, it also holds that
\begin{align*}
    & \Bigl\{\Delta_{I,4,\btheta_R^\circ}^4 \vee \frac{\Delta_{I,\infty
    ,\btheta_R^\circ}^4 \log(p \vee n)}{\size{I}}\Bigr\}^{\frac{1}{2}} \{\size{I} \log(p \vee n)\}^{\frac{1}{2}} \\
    \le & C_{\theta} \sigma_x \Bigl[\{C_{\theta} \sigma_x s_I \log(p \vee n)\} \vee \{\Delta_{I,\btheta_R^\circ}^2 \size{I} s_I \log(p \vee n)\}^{\frac{1}{2}}\Bigr].
\end{align*}
Therefore,
\begin{align}\label{equ:rf_err_ieq2}
    &\bigabs{\widetilde{\L}_I^\circ - \sum_{i\in I} \epsilon_i^2 - \Delta_{I,\btheta_R^\circ}^2 \size{I}} \nonumber\\
    \le & C_{u,2} C_x (C_x \sigma_x C_{\theta} + 2 C_{\epsilon}s_I^{-\frac{1}{2}}) \Bigl[\{C_{\theta} \sigma_x s_I \log(p \vee n)\} \vee \{\Delta_{I,\btheta_R^\circ}^2 \size{I} s_I \log(p \vee n)\}^{\frac{1}{2}}\Bigr].
\end{align}

By Eq.(\ref{equ:rf_err_ieq1}) and Eq.(\ref{equ:rf_err_ieq2}),
\begin{align*}
    & \bigabs{\widetilde{\L}_I - \sum_{i\in I} \epsilon_i^2 - \Delta_{I,\btheta_R^\circ}^2 \size{I}} \\
    \le & (1 + C_{\ref*{lem:re}}) r^{-1} \kmin C_{\ref*{lem:oracle}}^2 D_{R}^2 s_R \log (p \vee n) \\
    &+ 6^{-1} r^{-\frac{1}{2}} \kmin C_{\ref*{lem:oracle}}^2 D_{I, \btheta_{R}^\circ} D_{R} s_R \log(p \vee n) \\
    &+ 2 r^{-\frac{1}{2}} (1-r)^{\frac{1}{2}} \kmin^{\frac{1}{2}} C_{\ref*{lem:oracle}} D_{R} \{\Delta_{J,\btheta_R^\circ}^2 \size{J}\}^{\frac{1}{2}} \{s_R \log(p \vee n)\}^{\frac{1}{2}}.\\
    &+ C_{u,2} C_x (C_x \sigma_x C_{\theta} + 2 C_{\epsilon}s_I^{-\frac{1}{2}}) \Bigl[\{C_{\theta} \sigma_x s_I \log(p \vee n)\} \vee \{\Delta_{I,\btheta_R^\circ}^2 \size{I} s_I \log(p \vee n)\}^{\frac{1}{2}}\Bigr].
\end{align*}

\end{proof}

We now proceed to certify that $\mathbb{G} = \mathbb{G}_1 \cap \mathbb{G}_2^- \cap \mathbb{G}_2^+ \cap \mathbb{G}_3$ holds with high probability.

\begin{corollary}\label{cor:mix_err}
    Assume Condition~\ref{cond:change}, Condition~\ref{cond:parameter}, and Condition~\ref{cond:dist} hold.
    Under the same probability event in Lemma \ref{lem:rf_err} and with sufficiently large $C_{\mathsf{m}} \ge C_{\mathsf{re}}$, we have the following conclusions.
        \begin{enumerate}[label=(\alph*)]
        \item For $I$ such that $\Delta_I = 0$ and $\size{I} \ge C_{\mathsf{m}} s \log(p \vee n)$,
        \begin{equation*}
            \Bigabs{\widetilde{\L}_I - \sum_{i \in I} \epsilon_i^2} \le \frac{48 C_{u,2}^2 C_x^2 \sigma_{x}^2 C_{\epsilon
            }^2 s \log(p \vee n)}{\kmin} \triangleq C_{\ref*{cor:mix_err}.1} s \log(p \vee n),
        \end{equation*}
        where $C_{\ref*{cor:mix_err}.1} = \frac{6 + 6 C_{\ref*{lem:re}} + r^{\frac{1}{2}}}{6r} \kmin C_{\ref*{lem:oracle}}^2 C_\epsilon^2 + C_{u,2} C_x \sigma_x (C_x \sigma_x C_{\theta} + 2 C_{\epsilon}s^{-\frac{1}{2}}) C_{\theta}$.
        \item For $I = (a, b] \in E_{2}^{-}$ such that $I \cap \truecps = \{\tau_k^\ast\}$, $\min(\tau_k^\ast - a, b - \tau_k^\ast) \le \widetilde{C} \Delta_k^{-2} s \log(p \vee n)$ and $\size{I} \ge C_{\mathsf{m}} s \log(p \vee n)$,
        \begin{equation*}
            \widetilde{\L}_I - \sum_{i \in I} \epsilon_i^2 - \Delta_{I}^2 \size{I} \ge - C_{\ref*{cor:mix_err}.2} s \log(p \vee n),
        \end{equation*}
        and for $I \in E_2^{+} \subset E_2^{-}$ such that $I \cap \truecps = \{\tau_k^\ast\}$, $\min(\tau_k^\ast - a, b - \tau_k^\ast) \le \widetilde{C} \Delta_k^{-2} s \log(p \vee n)$ and $\size{I} \ge (\Delta_k^{2} \vee 1) C_{\mathsf{snr}} s \log(p \vee n)$,
        \begin{equation*}
            \widetilde{\L}_I - \sum_{i \in I} \epsilon_i^2 - \Delta_{I}^2 \size{I} \le C_{\ref*{cor:mix_err}.2} s \log(p \vee n),
        \end{equation*}
        where
        \begin{align*}
            C_{\ref*{cor:mix_err}.2} =& 2 \{(1 + C_{\ref*{lem:re}}) (r^{-1} C_{\epsilon}^2 + r^{-2} C_x^2 C_{m}^{-1} \widetilde{C}) + 6^{-1} (r^{-\frac{1}{2}} C_{\epsilon}^2 + r^{-\frac{3}{2}} C_x^2 C_{m}^{-1} \widetilde{C})\} \kmin C_{\ref*{lem:oracle}}^2 \\
            &+ 2^{\frac{3}{2}} (r^{-\frac{1}{2}} C_{\epsilon} \widetilde{C}^{\frac{1}{2}} + r^{-1} C_x C_{\mathsf{m}}^{-1} \widetilde{C}) (1 - r)^{\frac{1}{2}} \kmin^{\frac{1}{2}} C_{\ref*{lem:oracle}}\\
            &+ 2 C_{u,2} C_x (C_x \sigma_x C_{\theta} + 2 C_{\epsilon}s_I^{-\frac{1}{2}}) (C_{\theta} \sigma_x + 2^{-\frac{1}{2}} r^{-\frac{1}{2}} \widetilde{C}^{\frac{1}{2}} ) + r^{-1} C_{\mathsf{snr}}^{-1} \widetilde{C}.
        \end{align*}

        \item For $I$ such that $\size{I} \ge C_{\mathsf{m}} s \log(p \vee n)$ and $\Delta_I^2 \size{I} \ge 2^{-1} \widetilde{C} s \log(p \vee n)$ for some sufficiently large $\widetilde{C} \ge 6 C_{\theta}^2 \sigma_x^2$,
        \begin{equation*}
            \widetilde{\L}_I - \sum_{i \in I} \epsilon_i^2 \ge (1 - C_{\ref*{cor:in_err}.3})\Delta_I^2 \size{I},
        \end{equation*}
        where
        \begin{align*}
            C_{\ref*{cor:mix_err}.3} =& (1 + C_{\ref*{lem:re}}) \kmin C_{\ref*{lem:oracle}}^2 \{ 6 C_{\epsilon}^2 \widetilde{C}^{-1} r^{-1} + 2 (C_x^2 + C_{\ref*{cor:mix_err}.3}' \widetilde{C}^{-1}) C_{\mathsf{m}}^{-1} r^{-2} \} \\
            & + \kmin C_{\ref*{lem:oracle}}^2 \{ C_{\epsilon}^2 \widetilde{C}^{-1} r^{-\frac{1}{2}} + 6^{-1} (C_x^2 + C_{\ref*{cor:mix_err}.3}' \widetilde{C}^{-1}) C_{\mathsf{m}}^{-1} (r^{-\frac{3}{2}} + r^{-\frac{1}{2}}) \} \\
&+ (1-r)^{\frac{1}{2}} \kmin^{\frac{1}{2}} C_{\ref*{lem:oracle}} \{12^{\frac{1}{2}} C_{\epsilon} \widetilde{C}^{-\frac{1}{2}} r^{-\frac{1}{2}} + 2 (C_x^2 + C_{\ref*{cor:mix_err}.3}' \widetilde{C}^{-1})^{\frac{1}{2}} C_{\mathsf{m}}^{-\frac{1}{2}} r^{-1} \} \\
&+ C_{u,2} C_x (C_x \sigma_x C_{\theta} + 2 C_{\epsilon}) (6 C_\theta \sigma_x \widetilde{C}^{-1} + 6^{\frac{1}{2}} \widetilde{C}^{-\frac{1}{2}}) .
        \end{align*}
    \end{enumerate}
\end{corollary}

\begin{proof}[Proof of Corollary~\ref{cor:mix_err}]
    All of the results in the three parts of Corollary~\ref{cor:mix_err} follow from the proof of Lemma \ref{lem:rf_err} and the conditions about the change signals, that is, the variations $\Delta_{I}^2 \size{I}$. In the proof, we will also need to analyze the relationship of $\Delta_{I}^2 \size{I}$ and $\Delta_{I, \btheta_R^\circ}^2 \size{I}$.

    \noindent\textit{Part\,(a).} For $I$ such that $\Delta_{I} = 0$ and $\size{I} \ge \delta_{\mathsf{m}}$, there is not changepoint in $I$. So $s_I = s_R = s$, $\btheta_R^\circ = \btheta_I^\circ$, $\Delta_{I, \btheta_R^\circ} = \Delta_{J, \btheta_R^\circ} = 0$. Since $C_{\mathsf{m}}$ is sufficiently large and $r \in (0, 1]$ is a fixed constant, $D_R = D_{I, \btheta_R^\circ} = C_\epsilon$. Therefore,
\begin{align*}
    & \bigabs{\widetilde{\L} - \sum_{i \in I} \epsilon_i^2} \\
    \le & \Bigl\{\frac{6 + 6 C_{\ref*{lem:re}} + r^{\frac{1}{2}}}{6r} \kmin C_{\ref*{lem:oracle}}^2 C_\epsilon^2 + C_{u,2} C_x \sigma_x (C_x \sigma_x C_{\theta} + 2 C_{\epsilon}s^{-\frac{1}{2}}) C_{\theta}\Bigr\} s \log(p \vee n)\\
    = & C_{\ref*{cor:mix_err}.1} s \log(p \vee n),
\end{align*}

\noindent\textit{Part\,(b).} We study the interval $I$ from the set
$$
E_2^-=\{I = (a, b]: \size{I} \ge \delta_{\mathsf{m}}, I \cap \truecps = \{\tau_k^\ast\}, \min(\tau_k^\ast - a, b - \tau_k^\ast) \le \widetilde{C} \Delta_k^{-2} s \log(p \vee n)\}.
$$

Since $I$ contains only one changepoint, we have $s_R \le s_I \le 2 s$, $\Delta_{R, \infty}^2 \le \Delta_{I, \infty, \btheta_R^\circ}^2 \le 2 \sigma_x^2 C_{\theta}^2 s$.

For the average variation terms, by definition, we have $\Delta_{R}^2 \size{R} \le \Delta_{I}^2 \size{I} \le \widetilde{C} s \log(p \vee n)$. It means that $\Delta_R^2 \le r^{-1} \Delta_I^2 \le r^{-1} C_{\mathsf{m}}^{-1} \widetilde{C}$. Similarly, with some calculation, we can also have $\Delta_{J, \btheta_R^\circ}^2 \size{J} \le \Delta_{I, \btheta_R^\circ}^2 \size{I} \le r^{-1} \Delta_I^2 \size{I} \le r^{-1} \widetilde{C} s \log(p \vee n)$ and $\Delta_{I, \btheta_R^\circ}^2 \le r^{-1} C_{\mathsf{m}}^{-1} \widetilde{C}$.

We now discuss the terms $D_R^2$ and $D_{I,\btheta_R^\circ}^2$.
Since $\size{I} \ge \delta_{\mathsf{m}} \ge C_{\mathsf{m}} s \log(p \vee n)$ and $C_{\mathsf{m}}$ is sufficiently large, we have
\begin{align*}
& \size{I}^{-1} (C_x^2 \Delta_{I, \infty}^2 + C_{\epsilon}^2) \log(p \vee n) \\
\le & \size{I}^{-1} (2 C_x^2 \sigma_x^2 C_{\theta}^2 s + C_{\epsilon}^2) \log(p \vee n) \\
\le & (2 C_x^2 \sigma_x^2 C_{\theta}^2 + C_{\epsilon}^2 s^{-1}) C_{\mathsf{m}}^{-1} \le C_{\epsilon}^2.
\end{align*}
Therefore $D_{I, \btheta_R^\circ}^2 = C_x^2 \Delta_{I, \btheta_R^\circ}^2 + C_{\epsilon}^2$ and $D_{R}^2 = C_x^2 \Delta_{R}^2 + C_{\epsilon}^2$.

Combining Lemma \ref{lem:rf_err}, we have
\begin{align*}
    & \bigabs{\widetilde{\L}_I - \sum_{i\in I} \epsilon_i^2 - \Delta_{I,\btheta_R^\circ}^2 \size{I}} \\
    \le & (1 + C_{\ref*{lem:re}}) r^{-1} \kmin C_{\ref*{lem:oracle}}^2 C_{\epsilon}^2 s_R \log (p \vee n) +  (1 + C_{\ref*{lem:re}}) r^{-1} \kmin C_{\ref*{lem:oracle}}^2 C_x^2 \Delta_{R}^2 s_R \log (p \vee n) \\
    &+ 6^{-1} r^{-\frac{1}{2}} \kmin C_{\ref*{lem:oracle}}^2 C_{\epsilon}^2 s_R \log(p \vee n) + 12^{-1} r^{-\frac{1}{2}} \kmin C_{\ref*{lem:oracle}}^2 C_x^2 (\Delta_{R}^2 + \Delta_{I, \btheta_R^\circ}^2) s_R \log(p \vee n)  \\
    &+ 2 r^{-\frac{1}{2}} (1-r)^{\frac{1}{2}} \kmin^{\frac{1}{2}} C_{\ref*{lem:oracle}} (C_x \Delta_{R} + C_{\epsilon}) \{\Delta_{J,\btheta_R^\circ}^2 \size{J}\}^{\frac{1}{2}} \{s_R \log(p \vee n)\}^{\frac{1}{2}}.\\
    &+ C_{u,2} C_x (C_x \sigma_x C_{\theta} + 2 C_{\epsilon}s_I^{-\frac{1}{2}}) \Bigl[\{C_{\theta} \sigma_x s_I \log(p \vee n)\} \vee \{\Delta_{I,\btheta_R^\circ}^2 \size{I} s_I \log(p \vee n)\}^{\frac{1}{2}}\Bigr] \\
    \le & 2 \{(1 + C_{\ref*{lem:re}}) (r^{-1} C_{\epsilon}^2 + r^{-2} C_x^2 C_{m}^{-1} \widetilde{C}) + 6^{-1} (r^{-\frac{1}{2}} C_{\epsilon}^2 + r^{-\frac{3}{2}} C_x^2 C_{m}^{-1} \widetilde{C})\} \kmin C_{\ref*{lem:oracle}}^2 s \log (p \vee n) \\
    & + 2^{\frac{3}{2}} (r^{-\frac{1}{2}} C_{\epsilon} \widetilde{C}^{\frac{1}{2}} + r^{-1} C_x C_{\mathsf{m}}^{-1} \widetilde{C}) (1 - r)^{\frac{1}{2}} \kmin^{\frac{1}{2}} C_{\ref*{lem:oracle}} s \log(p \vee n) \\
    & + 2 C_{u,2} C_x (C_x \sigma_x C_{\theta} + 2 C_{\epsilon}s_I^{-\frac{1}{2}}) (C_{\theta} \sigma_x + 2^{-\frac{1}{2}} r^{-\frac{1}{2}} \widetilde{C}^{\frac{1}{2}} ) s \log(p \vee n) \\
    \triangleq & C_{\ref*{cor:mix_err}.2}' s \log(p \vee n).
\end{align*}

For $I = (a, b] \in E_2^+ \subseteq E_2^-$, since $\size{I} \ge (\Delta_k^2 \vee 1) C_{\mathsf{snr}} s \log(p \vee n)$ and $C_{\mathsf{snr}}$ is sufficiently large, we have $0 \le \Delta_{I, \btheta}^2 \size{I} - \Delta_{I}^2 \size{I} \le r^{-1} C_{\mathsf{snr}}^{-1} \widetilde{C} s \log(p \vee n)$. Therefore, for $I \in E_2^+$,
\[
    \widetilde{\L}_I - \sum_{i \in I} \epsilon_i^2 - \Delta_I^2 \size{I} \le C_{\ref*{cor:mix_err}.2} s \log(p \vee n),
\]
with $C_{\ref*{cor:mix_err}.2} = C_{\ref*{cor:mix_err}.2}' + r^{-1} C_{\mathsf{snr}}^{-1} \widetilde{C}$.

For $I \in E_2^-$, similarly,
\[
    \widetilde{\L}_I - \sum_{i \in I} \epsilon_i^2 - \Delta_I^2 \size{I} \ge \widetilde{\L}_I - \sum_{i \in I} \epsilon_i^2 - \Delta_{I, \btheta_R^\circ}^2 \size{I} \ge - C_{\ref*{cor:mix_err}.2}' s \log(p \vee n) \ge - C_{\ref*{cor:mix_err}.2} s \log(p \vee n).
\]

\noindent\textit{Part\,(c).}
 For $I$ such that $\size{I} \ge C_{\mathsf{m}} s \log(p \vee n)$ and $\Delta_I^2 \size{I} \ge 2^{-1}\widetilde{C} s \log(p \vee n)$, it also holds that $\Delta_{I, \btheta_R^\circ}^2 \size{I} \ge \Delta_{I}^2 \size{I} \ge 2^{-1}\widetilde{C} s \log(p \vee n)$.

 If $I$ contains only one changepoint, we have $s_I \le 2 s$ and $s_R \le 2 s$. Therefore
it holds that $\size{I} \ge 2^{-1} C_{\mathsf{m}} s_{I} \log(p \vee n)$ and $\size{R} \ge 2^{-1} C_{\mathsf{m}} r s_{R} \log(p \vee n)$. If $I$ contains more changepoints, by Conditions \ref{cond:change} and \ref{cond:parameter} and provided that $C_{\mathsf{snr}}$ is sufficiently large, it still holds that $\size{I} \ge 2^{-1} C_{\mathsf{m}} s_{I} \log(p \vee n)$ and $\size{R} \ge 2^{-1} C_{\mathsf{m}} r s_{R} \log(p \vee n)$.
Following similar arguments in the proof of Corollary~\ref{cor:in_err}, we can obtain that $\Delta_{I}^2 \size{I} \ge 6^{-1} \widetilde{C} s_I \log(p \vee n)\ge 6^{-1}\widetilde{C} s_R \log(p \vee n)$.
Therefore we have $s_R \log(p \vee n) \le \frac{2\size{R}}{C_{\mathsf{m}} r}$ and $s_R \log(p \vee n) \le \frac{2\size{R}}{C_{\mathsf{m}} r}$ and $s_R \log(p \vee n) \le s_I \log(p \vee n) \le 6 \widetilde{C}^{-1} \Delta_{I}^2 \size{I} \le 6 \widetilde{C}^{-1} \Delta_{I, \btheta_R^\circ}^2 \size{I}$.

Then we consider the two terms $D_{I, \btheta_R^\circ}^2$ and $D_R^2$.
By the above justifications, we can obtain that
\begin{align*}
    & \size{I}^{-1}(C_x^2\Delta_{I, \infty, \btheta_R^\circ}^2 + C_{\epsilon}^2) \log(p \vee n)\\
    \le & \size{I}^{-1}(C_x^2 \sigma_x^2 C_{\theta}^2 s_I + C_{\epsilon}^2) \log(p \vee n)\\
    \le & 6 \widetilde{C}^{-1} (C_x^2 \sigma_x^2 C_{\theta}^2 + C_{\epsilon}^2 s_I^{-1}) \Delta_{I, \btheta_R^\circ}^2.
\end{align*}
It follows that
$$D_{I, \btheta_R^\circ}^2 \le C_{\epsilon}^2 + \{C_x^2 + 6 \widetilde{C}^{-1} (C_x^2 \sigma_x^2 C_{\theta}^2 + C_{\epsilon}^2 s_I^{-1})\} \Delta_{I, \btheta_R^\circ}^2.$$
Similarly,
$$D_{R}^2 \le C_{\epsilon}^2 + \{C_x^2 + 6 \widetilde{C}^{-1} (C_x^2 \sigma_x^2 C_{\theta}^2 + C_{\epsilon}^2 s_R^{-1})\} \Delta_{R}^2.$$

Denote $C_{\ref*{cor:mix_err}.3}' = 6(C_x^2 \sigma_x^2 C_{\theta}^2 + C_{\epsilon}^2)$. By the above inequalities for the change variation terms, we have:
\begin{align*}
&\bigabs{\widetilde{\L}_I - \sum_{i\in I} \epsilon_i^2 - \Delta_{I,\btheta_R^\circ}^2 \size{I}}\\
\le & \Delta_{I, \btheta_R^\circ}^2 \size{I} \Big[ (1 + C_{\ref*{lem:re}}) \kmin C_{\ref*{lem:oracle}}^2 \{ 6 C_{\epsilon}^2 \widetilde{C}^{-1} r^{-1} + 2 (C_x^2 + C_{\ref*{cor:mix_err}.3}' \widetilde{C}^{-1}) C_{\mathsf{m}}^{-1} r^{-2} \} \\
& ~~~~~~~~~~~~+ \kmin C_{\ref*{lem:oracle}}^2 \{ C_{\epsilon}^2 \widetilde{C}^{-1} r^{-\frac{1}{2}} + 6^{-1} (C_x^2 + C_{\ref*{cor:mix_err}.3}' \widetilde{C}^{-1}) C_{\mathsf{m}}^{-1} (r^{-\frac{3}{2}} + r^{-\frac{1}{2}}) \} \\
& ~~~~~~~~~~~~+ (1-r)^{\frac{1}{2}} \kmin^{\frac{1}{2}} C_{\ref*{lem:oracle}} \{12^{\frac{1}{2}} C_{\epsilon} \widetilde{C}^{-\frac{1}{2}} r^{-\frac{1}{2}} + 2 (C_x^2 + C_{\ref*{cor:mix_err}.3}' \widetilde{C}^{-1})^{\frac{1}{2}} C_{\mathsf{m}}^{-\frac{1}{2}} r^{-1} \} \\
& ~~~~~~~~~~~~+ C_{u,2} C_x (C_x \sigma_x C_{\theta} + 2 C_{\epsilon}) (6 C_\theta \sigma_x \widetilde{C}^{-1} + 6^{\frac{1}{2}} \widetilde{C}^{-\frac{1}{2}}) \Big]\\
= & C_{\ref*{cor:mix_err}.3} \Delta_{I, \btheta_R^\circ}^2 \size{I}.
\end{align*}
Note that $C_{\mathsf{m}}$ is sufficiently large, we can set $\widetilde{C}$ large enough to make $C_{\ref*{cor:mix_err}.3}$ is sufficiently small such that $C_{\ref*{cor:mix_err}.3} < 1$.
Finally, we have
$$\widetilde{\L}_{I} - \sum_{i \in I} \epsilon_i^2 \ge (1 - C_{\ref*{cor:mix_err}.3}) \Delta_{I, \btheta_R^\circ}^2 \size{I} \ge (1 - C_{\ref*{cor:mix_err}.3}) \Delta_{I}^2 \size{I}.$$
\end{proof}

\subsection{Final Result}

Here, we will finally prove Theorem~\ref{thm:localization_error}.
Recall that in the proof of Lemma~\ref{lem:loc_err_g}, we identify the constant $\widetilde{C}$ by solving the inequality (\ref{equ:ieq_to_solve_Ctilde}). It concludes that any $\widetilde{C}$ with $\widetilde{C} \ge 4 (1 - C_{\ref*{lem:loc_err_g}.3})^{-1} (3 C_{\ref*{lem:loc_err_g}.1} + 10 C_{\ref*{lem:loc_err_g}.2})$ suffices for Lemma~\ref{lem:loc_err_g}.
Since $C_{\mathsf{m}}$ can be sufficiently large and $\widetilde{C} C_{\mathsf{m}}^{-1}$ can be sufficiently small, by the results in \ref{cor:in_err} and Corollary \ref{cor:mix_err}, we can verify that the feasible $\widetilde{C}$ exists for the inequality $\widetilde{C} \ge 4 (1 - C_{\ref*{lem:loc_err_g}.3})^{-1} (3 C_{\ref*{lem:loc_err_g}.1} + 10 C_{\ref*{lem:loc_err_g}.2})$, both for the full model-fitting (Corollary~\ref{cor:in_err}) and the Reliever (Corollary~\ref{cor:mix_err}).
Furthermore, the $\widetilde{C}$ will only depends on the constants in the expressions of $\{C_{\ref*{cor:in_err}.j}\}_{j=1}^3$ and $\{C_{\ref*{cor:mix_err}.j}\}_{j=1}^3$.
Thus it is independent of $(n, p, s, K^\ast)$.

In summary, the localization error bound of $\{\hat{\tau}_k\}$ in Theorem \ref{thm:localization_error} follows from Lemma \ref{lem:loc_err_g}, Corollary \ref{cor:in_err} and Corollary \ref{cor:mix_err}. And provided the localization error bound, the error bound of the parameter estimation follows from Lemma \ref{lem:oracle}.

\section{Proof of Corollary~\ref{thm:localization_error_temporal}}\label{sec:proof_under_dependence}

Firstly, note that under Condition~\ref{cond:change_temporal}, Lemma~\ref{lem:loc_err_g} still holds if we replace $\delta_{\mathrm{m}} = C_{\mathrm{m}} s \log(p \vee n)$ by $\delta_{\mathrm{m}} = C_{\mathrm{m}} \{s \log(p \vee n)\}^{2/\gamma - 1}$.
The proof will mainly rely on the replacement of Bernstein's inequality for independent data (Lemma~\ref{lem:bernstein}) by the functionally dependent one. The following lemma is a rearranged result of the functionally dependent Bernstein's inequality (Theorem 33 in \cite{xu2022change}) using the notations of this current manuscript. The proof can be found in the Supplementary Material of \cite{xu2022change}.

\begin{lemma}[Bernstein's inequality under dependence]\label{lem:bernstein_temporal}
    Let $\set{X_i}_{i=1}^n$ be a consecutive subsequence of an infinite functional dependent sequence $\set{X_t}_{t \in \Zbb}$ with dependence function $\set{g_t^X}$. Assume that the sequence is with zero means, exponentially decayed cumulative function dependence measures:
    \begin{equation*}
        \sup_{m \ge 0} \exp(c m^{\zeta_1}) \Delta_{m, 2}^X \le C_1,
    \end{equation*}
    and exponentially decayed tails $\sup_{1 \le i \le m} \norm{X_i}_{\Psi_{\zeta_2}} \le C_2$ for some positive constants $C_1$ and $C_2$ and $\zeta = (\zeta_1^{-1} + \zeta_2^{-1})^{-1} \in (0, 1]$. Then we have for $t > \sqrt{n}$,
    \begin{equation*}
        \Pbb\Bigset{\Bigabs{\sum_{i=1}^n X_i} > t} \le 2 \exp\Bigl\{-c_b \Bigl(t^{\zeta} \wedge \frac{t^2}{n}\Bigr)\Bigr\}.
    \end{equation*}
    By choosing $t = c_b^{-1} C_u \{(n A_{n, p, s})^{\frac{1}{2}} \vee (A_{n, p, s})^{\frac{1}{\zeta}} \}$ with $C_u \ge c_b$, we have
    \begin{equation*}
        \Pbb\Bigset{\Bigabs{\sum_{i=1}^n X_i} > t} \le  2 \exp\{-C_{u} A_{n, p, s}\}.
    \end{equation*}
    When $n \ge A_{n, p, s}^{\frac{2}{\zeta} - 1}$, we have
    \begin{equation*}
        \Pbb\Bigset{\Bigabs{\sum_{i=1}^n X_i} > c_b^{-1} C_u (n A_{n, p, s})^{\frac{1}{2}}} \le  2 \exp\{-C_{u} A_{n, p, s}\}.
    \end{equation*}
    Here $A_{n, p, s}$ is a diverging sequence. For instance, $A_{n, p, s} = \log(p \vee n)$ and $A_{n, p, s} = s \log(p \vee n)$.
\end{lemma}

By applying Lemma~\ref{lem:bernstein_temporal}, we can obtain parallel results to those in Section~\ref{subsec:pre_required}.
For example, under Condition~\ref{cond:dist_temporal}\,(a), we have the restricted eigenvalue conditions in Lemma~\ref{lem:re} hold uniformly for intervals $I$ such that $\size{I} \gtrsim \{s_I \log(p \vee n)\}^{\frac{2}{\zeta} - 1}$.

\begin{lemma}[Uniform restricted eigenvalue condition under dependence]\label{lem:re_temporal}
    Assume Condition~\ref{cond:dist_temporal}\,(a) holds.
    For any interval $I \subset (0,n]$, denote $\hat{\Sigma}_I = \size{I}^{-1} \sum_{i \in I} \xbf_i \xbf_i^\top$.
    Uniformly for all intervals $I \subset (0, n]$ such that $\size{I} \ge \{s_I \log(p \vee n)\}^{\frac{2}{\zeta} - 1}$, with probability at least $1 - \exp\{-C_{u,1} \log(p \vee n)\}$,
    \begin{equation*}
        \vbf^\top \hat{\Sigma}_{I} \vbf \ge \norm{\vbf}_{\Sigma}^2 - C_{u,2} C_x^2 \sigma_x^2 \Bigl\{\frac{s_I \log(p \vee n)}{\size{I}}\Bigr\}^{\frac{1}{2}} \Bigl(\norm{\vbf}_2^2 + \frac{1}{s_I}\norm{\vbf}_1^2\Bigr),\, \forall \vbf \in \Rbb^p,
    \end{equation*}
    and
    \begin{equation*}
        \vbf^\top \hat{\Sigma}_{I} \vbf \le \norm{\vbf}_{\Sigma}^2 + C_{u,2} C_x^2 \sigma_x^2 \Bigl\{\frac{s_I \log(p \vee n)}{\size{I}}\Bigr\}^{\frac{1}{2}} \Bigl(\norm{\vbf}_2^2 + \frac{1}{s_I}\norm{\vbf}_1^2\Bigr),\, \forall \vbf \in \Rbb^p,
    \end{equation*}
    where $C_{u,1}$ and $C_{u,2}$ are two universal constants.
    If additionally $\size{I} \ge C_{\mathsf{re}} s_I \log(p \vee n)$ with a sufficiently large constant $C_{\mathsf{re}} \ge 1 \vee ({34 C_{u,2}C_x^2 \sigma_x^2}/{\kmin})^2$, then for any support set $\Scal \in [p]$ with $\size{\Scal} \le s_I$ and $\vbf \in \Rbb^p$ such that $\norm{\vbf_{\Scal^\complement}}_1 \le 3 \norm{\vbf_{\Scal}}_1$, under the same event above,
    \begin{equation*}
        \frac{1}{2} \norm{\vbf}_{\Sigma}^2 \le (1 - C_{\ref*{lem:re_temporal}}) \norm{\vbf}_{\Sigma}^2 \le \vbf^\top \hat{\Sigma}_{I} \vbf \le (1 + C_{\ref*{lem:re_temporal}}) \norm{\vbf}_{\Sigma}^2 \le \frac{3}{2} \norm{\vbf}_{\Sigma}^2,
    \end{equation*}
    \begin{equation*}
        \frac{\kmin}{2} \norm{\vbf}_2^2 \le (1 - C_{\ref*{lem:re_temporal}}) \kmin \norm{\vbf}_{2}^2 \le \vbf^\top \hat{\Sigma}_{I} \vbf \le  (\sigma_x^2 + C_{\ref*{lem:re_temporal}} \kmin) \norm{\vbf}_{2}^2 \le \Bigl(\sigma_x^2 + \frac{\kmin}{2}\Bigr) \norm{\vbf}_2^2,
    \end{equation*}
    where $C_{\ref*{lem:re_temporal}} = 17 C_{u,2} C_x^2 \sigma_x^2 \kappa^{-1} C_{\mathsf{re}}^{-\frac{1}{2}}$.
\end{lemma}

The proof of Lemma~\ref{lem:re_temporal} is the same as Lemma~\ref{lem:re}'s by replacing the Bernstein's inequality with the functional dependence version.

Note that under Condition~\ref{cond:change_temporal}\,(b), we have for any $I, R \in \I$, $\{\Delta_{I, q}, \Delta_{I, q, \btheta_{R}^\circ}\}_{q=2,4,\infty}$ are all bounded as $O(1)$. The boundness ensures that we can directly apply Lemma~\ref{lem:bernstein_temporal} in the following lemma.

\begin{lemma}\label{lem:xxbeinf_temporal}
    Assume that Conditions~\ref{cond:change_temporal}\,(b), \ref{cond:dist_temporal} hold.
    For intervals $I, R \subset (0, n]$ with $\size{I} \ge \{\log(p \vee n)\}^{\frac{2}{\zeta} - 1}$ and a fixed $\btheta=\btheta_{R}^\circ$, with probability at least $1 - n^{-2}\exp\{-C_{u,1} \log(p \vee n)\}$,
    \begin{align*}
        & \Bignorm{\sum_{i \in I}\{ (\xbf_i \xbf_i^\top - \Sigma) (\btheta_i^\circ - \btheta) + \epsilon_i \xbf_i\}}_{\infty} \le C_{\ref*{lem:xxbeinf_temporal}} \{ \size{I} \log(p \vee n)\}^{\frac{1}{2}},
    \end{align*}
    for some positive constant $C_{\ref*{lem:xxbeinf_temporal}} > 0$.
\end{lemma}

By replacing Lemmas~\ref{lem:re}--\ref{lem:xxbeinf} with Lemmas~\ref{lem:re_temporal}--\ref{lem:xxbeinf_temporal}, we have the following oracle inequalities for lasso with functionally dependent data.

\begin{lemma}[Oracle inequalities for the parametric estimates]\label{lem:oracle_temporal}
    Assume that Conditions~\ref{cond:change_temporal}, \ref{cond:parameter_temporal} and \ref{cond:dist_temporal} hold.
    We have with probability at least $1 - 2 \exp\{-C_{u,1}\log(p \vee n)\}$, uniformly for any interval $I \subset (0, n]$ with $\size{I} \ge \{C_{\mathsf{re}} s_I \log(p \vee n)\} \vee [\{s_I \log(p \vee n)\}^{\frac{2}{\zeta} - 1}]$, provided that $\lambda_I = C_{\lambda} \{\size{I} \log(p \vee n)\}^{\frac{1}{2}}$ for some constant $C_{\lambda} > 0$, the solution $\htheta_I$ satisfies that
    \begin{equation*}
        \norm{\htheta_I - \btheta_I^\circ}_2 \le \kmin^{-\frac{1}{2}}\norm{\htheta_I - \btheta_I^\circ}_\Sigma \le C_{\ref*{lem:oracle_temporal}} \Bigl\{\frac{s_I \log (p \vee n)}{\size{I}}\Bigr\}^{\frac{1}{2}},
    \end{equation*}
    \begin{equation*}
        \norm{\htheta_I - \btheta_I^\circ}_1 \le C_{\ref*{lem:oracle_temporal}} s_I \Bigl\{\frac{ \log (p \vee n)}{\size{I}}\Bigr\}^{\frac{1}{2}},
    \end{equation*}
    for some constant $C_{\ref*{lem:oracle_temporal}} > 0$. Furthermore, let $\Scal_I$ be the support set of $\btheta_I^\circ$, we have $\norm{\htheta_{I,\Scal_I^\complement} - \btheta_{I,\Scal_I^\complement}^\circ}_1 \le 3 \norm{\htheta_{I,\Scal_I} - \btheta_{I,\Scal_I}^\circ}_1$.
\end{lemma}

Lastly, we present lemmas that parallel Lemmas~\ref{lem:xb2}--\ref{lem:xbe}.

\begin{lemma}\label{lem:xb2_temporal}
    Assume that Condition~\ref{cond:dist_temporal} and the condition that $K = O(1)$ in Condition~\ref{cond:change_temporal} hold.
    For interval $I \subset (0, n]$ and a fixed $\btheta$,
    with probability at least $1 - n^{-2}\exp\{-C_{u,1} \log(p \vee n)\}$, uniformly for any sub-interval $I \subset (0, n]$,
    \begin{align*}
        \Bigabs{\sum_{i \in I} \{\xbf_i^\top(\btheta_i^\circ - \btheta)\}^2 - \Delta_{I, \btheta}^2 \size{I}} \le &  C_{u,2} C_x^2 \biggl[\Delta_{I,4,\btheta}^4 \vee \frac{\Delta_{I,\infty,\btheta}^4 \{\log(p \vee n)\}^{\frac{2}{\zeta} - 1}}{\size{I}}\biggr]^{\frac{1}{2}} \{\size{I} \log(p \vee n)\}^{\frac{1}{2}},
    \end{align*}
    where $C_{u,1}$ and $C_{u,2}$ are two universal constants.
\end{lemma}

\begin{lemma}\label{lem:xbe_temporal}
    Assume that Condition~\ref{cond:dist_temporal} and the condition that $K = O(1)$ in Condition~\ref{cond:change_temporal} hold.
    For interval $I \subset (0, n]$ and a fixed $\btheta$,
    with probability at least $1 - n^{-2}\exp\{-C_{u,1} \log(p \vee n)\}$,
    \begin{equation*}
        \Bigabs{\sum_{i \in I} \xbf_i^\top(\btheta_i^\circ - \btheta) \epsilon_i} \le C_{u,2} C_x C_{\epsilon} \biggl[\Delta_{I, \btheta}^2 \vee \frac{\Delta_{I,\infty, \btheta}^2 \{\log(p \vee n)\}^{\frac{2}{\zeta} - 1}}{\size{I}}\biggr]^{\frac{1}{2}} \{\size{I} \log(p \vee n)\}^{\frac{1}{2}},
    \end{equation*}
    where $C_{u,1}$ and $C_{u,2}$ are two universal constants.
\end{lemma}

For Lemmas~\ref{lem:re_temporal}--\ref{lem:oracle_temporal}, only the Bernstein's inequality (Lemma~\ref{lem:bernstein}) needs to be replaced by Lemma~\ref{lem:bernstein_temporal} in the proofs. However, for Lemmas~\ref{lem:xb2}--\ref{lem:xbe}, directly applying Lemma~\ref{lem:bernstein_temporal} will induce upper bound with order $\{\size{I} \log(p \vee n)\}^{\frac{1}{2}}$, which drops the multiplier $\Delta_{I, 4, \btheta}^4$ and $\Delta_{I,\btheta}^2$ in the upper bound. It is because in Lemma~\ref{lem:bernstein_temporal}, the tail bound is scaled with $n^{-1}$ in the term $\frac{t^2}{n}$ rather than $\frac{t^2}{\Var(\sum_{i =1}^n X_i)}$ or $\frac{t^2}{\sum_{i=1}^n \norm{X_i}_{\Psi_{\zeta_2}}^2}$.

To solve this issue, we need the condition that $K = O(1)$ which is also assumed in \cite{xu2022change} for the functional dependence setting. For an interval $I$, let $\set{I_j}_{j=1}^{a}$ be the sub-segments of $I$ divided by the true changepoints $\truecps \cap I$. Then we can apply Lemma~\ref{lem:bernstein_temporal} to $\{ \sum_{i \in I_j} \{\xbf_i^\top(\btheta_{I_j}^\circ - \btheta)\}^2 - \Delta_{I_j, \btheta}^2 \size{I_j} \}_{j=1}^a$ and $\{ \sum_{i \in I_j} \xbf_i^\top(\btheta_{I_j}^\circ - \btheta) \epsilon_i \}_{j=1}^a$ and then taking the union bound will lead to the deviation bounds in Lemmas~\ref{lem:xb2}--\ref{lem:xbe}.

Finally, by replacing the deviation bounds in Section~\ref{subsec:pre_required} with the above functionally dependent bounds, we can obtain the results in Corollary~\ref{thm:localization_error_temporal}.

\section{Additional Numerical Results}\label{sec:extra_simulation}

\subsection{Single-Changepoint Scenario (Section \ref{sec:comp_ts})}

The data in the single changepoint scenario in Section \ref{sec:comp_ts} are generated from the following model,
\begin{equation*}
    y_i = \xbf_i^\top \btheta_1 \id\{i \le \tau^\ast\} + \xbf_i^\top \btheta_2 \id\{i > \tau^\ast\} + \epsilon_i,\, i = 1,\dots, n,
\end{equation*}
where $\{\epsilon_i\}$ and $\{\xbf_i\}$ are drawn independently satisfying $\epsilon_i \sim \Ncal(0, 1)$ and $\xbf_i \sim \Ncal_p(0, \Sigma)$. Here $\Sigma$ is a $p \times p$ matrix with elements $\Sigma_{ij} = 1/2^{|i- j|}$. The regression parameters of the model are set to be
$$\btheta_1 = (1/3,1/3,1/3,1/3,0,\dots,0)_{p \times 1}^\top$$
and
$$\btheta_2 = (\bm{0}_{1 \times 4}, 1/3, 1/3, 1/3, 1/3,0,\dots,0)_{p \times 1}^\top.$$
We set $n = 1200$ and the true changepoint $\tau^\ast = 120$.

\subsection{Extended Comparison with the Two-Step Method}

In Figures~\ref{fig:ts_lasso}--\ref{fig:ts_nmcd}, we provide the complementary numerical results of the multiple changepoint scenarios in Section \ref{sec:comp_ts} with $n$ varying from $n=300$ to $n=1200$ under both the high-dimensional linear model and the univariate nonparametric model. In the \Relief{} method, we set $r=0.9$ as recommended. The \Relief{} provides almost comparable performance with the original algorithm in all the cases.

\begin{figure}[!ht]
    \centering
    \includegraphics[width=0.99\linewidth]{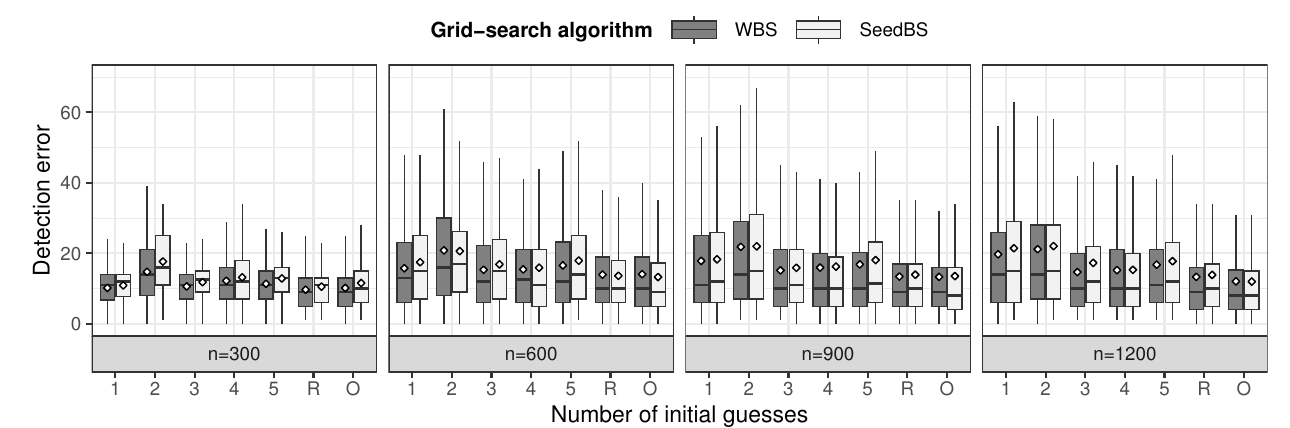}
\caption{Comparison of the modified two-step approach with multiple initial guesses ($1$-$5$), the \Relief{} method (R) and the original full model-fitting (O), under the setting in Section \ref{subsec:hdlinear}.}
\label{fig:ts_lasso}
\end{figure}

\begin{figure}[!ht]
    \centering
    \includegraphics[width=0.99\linewidth]{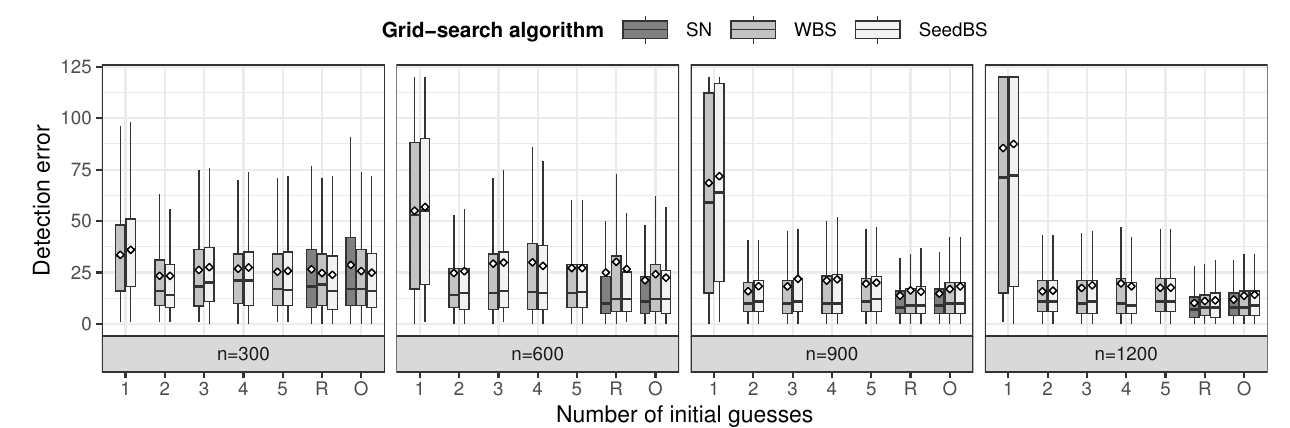}
\caption{Comparison of the modified two-step approach with multiple initial guesses ($1$-$5$), the \Relief{} method (R) and the original full model-fitting (O), under the setting in Section \ref{subsec:nmcd}.}
\label{fig:ts_nmcd}
\end{figure}

\subsection{Impact of a Small Minimal-Spacing Parameter}\label{sec:small_delta_m}

For numerical stability, we choose $\delta_{\mathrm{m}} = 20$ in all simulations in the main text. Here we rerun experiments with much smaller values. In high-dimensional linear models, setting $\delta_{\mathrm{m}} = 3$ (the smallest value accepted by \textsf{glmnet} without runtime errors) still works effectively---see Figure \ref{fig:lasso_small_dm_combined}. The results closely match those with larger $\delta_{\mathrm{m}}$, confirming that Reliever is robust to this choice.

\begin{figure}[!ht]
    \centering
    \begin{subfigure}{1\textwidth}
        \centering
        \includegraphics[width=0.99\linewidth]{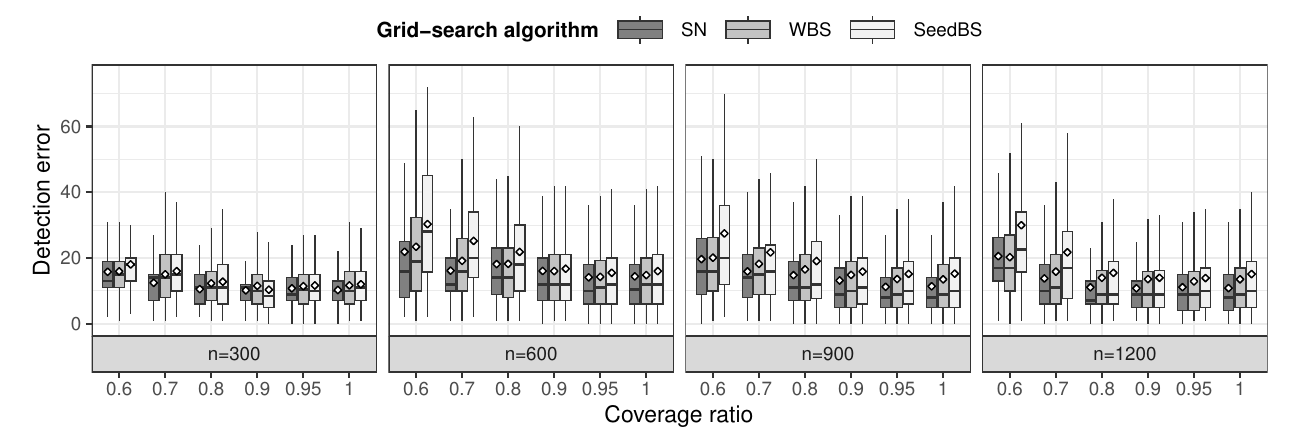}
    \end{subfigure}
    \begin{subfigure}{1\textwidth}
        \centering
        \includegraphics[width=0.99\linewidth]{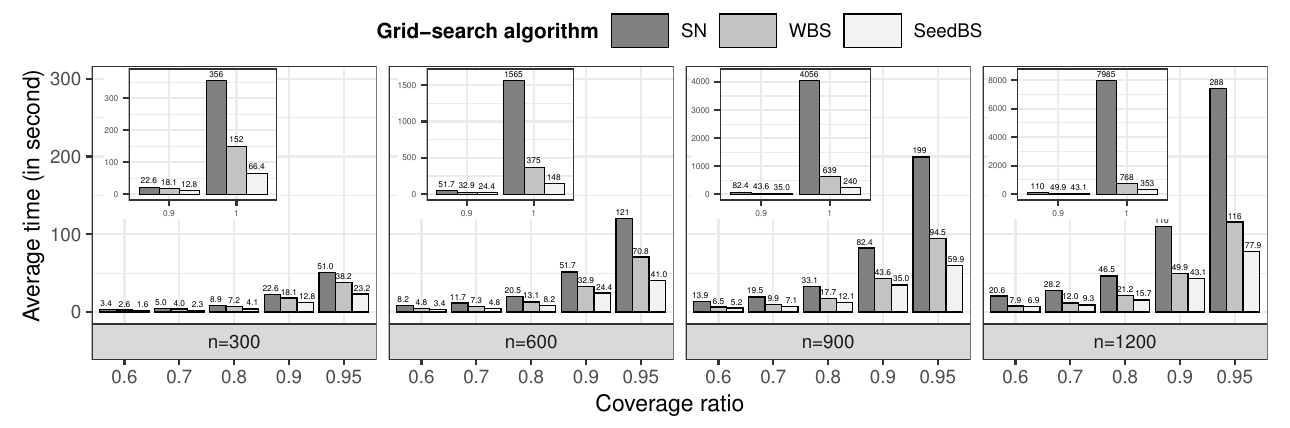}
    \end{subfigure}
    \caption{Performance of grid-search algorithms with Reliever under a high-dimensional linear model using a small minimum spacing $\delta_{\mathrm{m}} = 3$.}
    \label{fig:lasso_small_dm_combined}
\end{figure}

\subsection{High-Dimensional Linear Model with Temporal Dependence}\label{sec:exp_hd_temporal}

We conduct simulations for the high-dimensional linear setting of Section~\ref{subsec:hdlinear} under an AR(1) dependence ($\rho=0.3$) in both covariates and noise, following \citet{xu2022change}: $$\mathbf{x}_i = \rho \mathbf{x}_{i-1} + \sqrt{1 - \rho^2} \widetilde{\mathbf{x}}_{i},\quad\epsilon_i = \rho \epsilon_{i-1} + \sqrt{1 - \rho^2} \widetilde{\epsilon_i},$$ with $\widetilde{\mathbf{x}}_i \stackrel{i.i.d.}{\sim} \mathcal{N}(0, \mathbf{I}_p)$ and $\{\widetilde{\epsilon}_i\} \stackrel{i.i.d.}{\sim} \mathcal{N}(0, 1)$. Figure \ref{fig:lasso_temporal} shows that, even under dependence, Reliever matches the accuracy of the original algorithms while greatly reducing runtime.

\begin{figure}[!ht]
    \centering
    \begin{subfigure}{1\textwidth}
        \centering
        \includegraphics[width=0.99\linewidth]{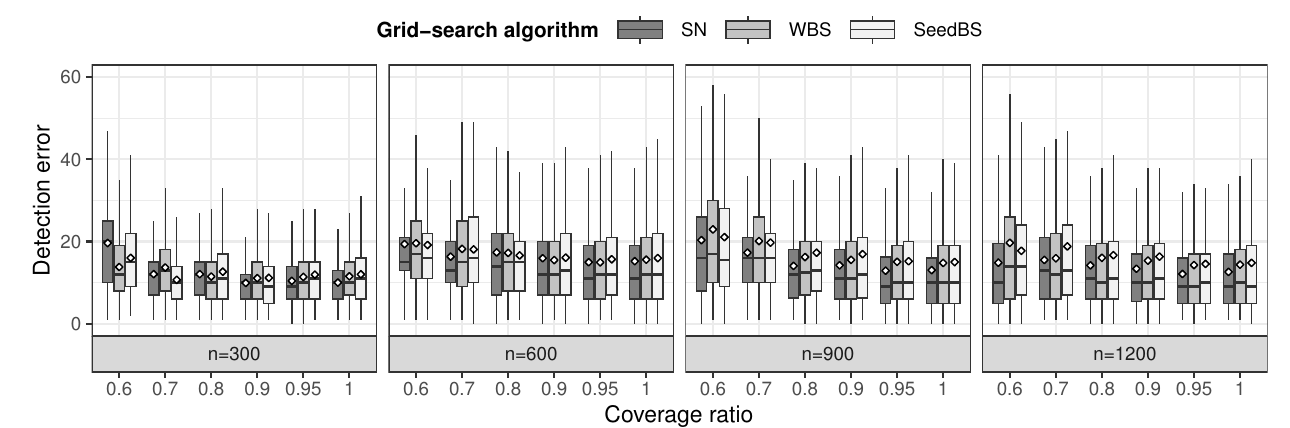}
    \end{subfigure}
    \begin{subfigure}{1\textwidth}
        \centering
        \includegraphics[width=0.99\linewidth]{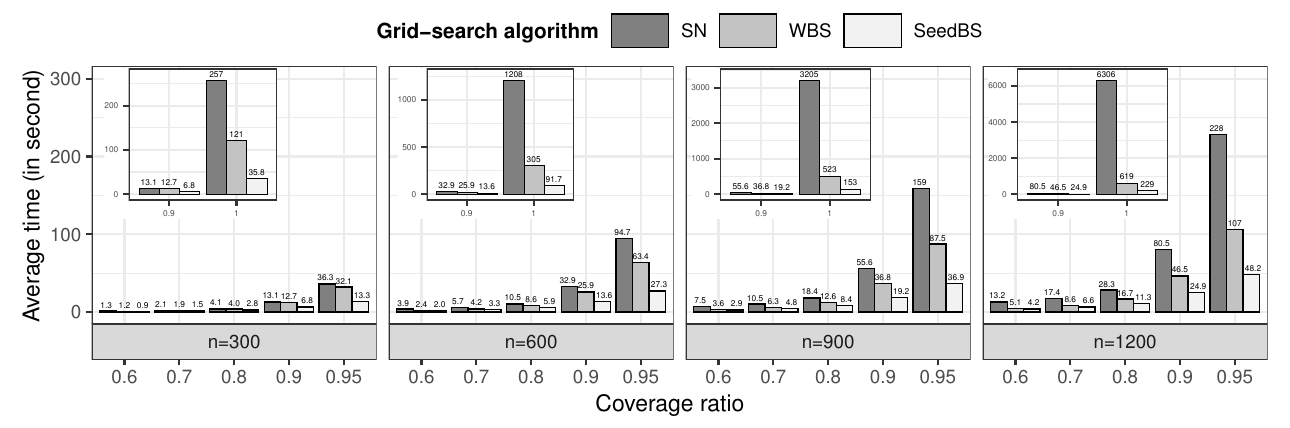}
    \end{subfigure}
    \caption{Performance under temporal dependence in high-dimensional linear models.}
    \label{fig:lasso_temporal}
\end{figure}

\subsection{A Priori Estimation of Model-Fitting Time}\label{sec:estimated_time}

In settings where the grid-search intervals can be anticipated, total fitting time for Reliever can be approximated as follows:
\begin{itemize}[leftmargin=0em]
    \item \textit{Build relief layers.} Follow Definition~\ref{itv_construction}, form the layers $\mathcal{R}_k$ so that $\mathcal{R} = \bigcup_{k=0}^{\lfloor \log_{b} \{(1 + w) n/\delta_{\mathrm{m}}\} \rfloor} \mathcal{R}_k$.
    \item \textit{Sample a few intervals.} From each layer $\mathcal{R}_k$, draw $c$ intervals (with replacement if $|\mathcal{R}_k| < c$).
    \item \textit{Measure single-fit times.} Fit the model on those samples to obtain $\hat{\textsf{time}}_k$, the average time per fit in layer $k$.
    \item \textit{Count expected visits.} Estimate $\textsf{num}_k$, the number of layer-$k$ intervals the grid-search algorithm will actually visit.
    \vspace{-1em}
    \begin{itemize}
        \item For OP or SN, $\textsf{num}_k$ follows directly from their deterministic schedules.
        \item For WBS or SeedBS, generate the wild/seeded intervals $\mathcal{W}$ and set $\textsf{num}_k$ as the size of
        $\{R \in \mathcal{R}_k: \exists (a, b] \in \mathcal{W} \,\&\, t \in (a, b], R \mbox{ is the relief interval of } (a, t] \mbox{ or } (t, b]\}$.
        \item For PELT, use the OP count as an upper bound when pruning is inactive.
    \end{itemize}
    \item \textit{Total time estimate.} Compute $\sum_{k} \hat{\textsf{time}}_k \times \textsf{num}_k$.
\end{itemize}

\begin{table}[htb]
\setlength\tabcolsep{0pt}
\centering
\begin{threeparttable}
\caption{\small Estimated and realized model-fitting times (in seconds) for the changepoint detection in the high-dimensional linear model in Section~4.1 with the grid search methods including SN, WBS and SeedBS, under a single run with $n = 1200$.}
\label{tab:time_est}
\tabcolsep=0.44em
\begin{tabular*}{.97\linewidth}{cc@{\extracolsep{\fill}}*{10}{r}}
\toprule
Method & $r$ & 0.5 & 0.6 & 0.7 & 0.8 & 0.9 & 0.95 & 0.97 & 0.98 & 1.0 \\
\midrule
\multirow{2}{*}{SN} & Estimate & 3.1 & 6.8 & 11.2 & 25.3 & 94.1 & 276.6 & 591.1 & 1000.0 & 7011.1 \\
& Realization  & 3.0 & 6.1 & 9.6 & 22.3 & 80.7 & 267.2 & 600.9 & 979.6 & 7617.3 \\
\midrule
\multirow{2}{*}{WBS} & Estimate & 2.7 & 5.0 & 10.3 & 18.0 & 55.9 & 131.4 & 240.5 & 313.4 & 889.7 \\
& Realization  & 2.9 & 5.1 & 9.1 & 19 & 56.4 & 131.5 & 241.8 & 299.4 & 885.8 \\
\midrule
\multirow{2}{*}{SeedBS} & Estimate & 2.0 & 4.2 & 6.9 & 10.2 & 23.5 & 48.1 & 74.6 & 96.8 & 221.2 \\
& Realization & 2.6 & 3.9 & 6.5 & 11.9 & 27.6 & 52.5 & 76.9 & 108.3 & 235.5 \\
\bottomrule
\end{tabular*}
\end{threeparttable}
\end{table}

Using $c=1$, we applied this recipe to SN, WBS, and SeedBS on a high-dimensional linear model; Table~\ref{tab:time_est} shows the estimated versus actual fitting times, which match closely. Users can therefore select $r$ to match a target runtime with reasonable confidence.

\bibliography{ref}

\end{document}